\documentclass[10pt]{amsart}
\usepackage{verbatim}
\usepackage{amsmath}
\usepackage{amssymb}
\usepackage{amsthm}
\usepackage{latexsym}
\usepackage{url}
\usepackage{hyperref}
\usepackage{caption, subcaption}
\usepackage{graphicx}
\usepackage{tikz}
\usetikzlibrary{matrix,arrows}
\usepackage{enumerate}
\usepackage{appendix}
\usepackage{color}
\usepackage{eqnarray}
\usepackage[colorinlistoftodos]{todonotes}
\usepackage{pdfcolfoot}  % to have colored footnotes split across pages
\usepackage{tikz-cd}
\usetikzlibrary{decorations.pathmorphing}
\usetikzlibrary{arrows,decorations.pathreplacing,decorations.markings,shapes.geometric}
\usetikzlibrary{calc,intersections,through,backgrounds,patterns}
\tikzset{
every loop/.style={},
tadpole/.style={loop, min distance=15mm, in=-40, out=40},
bdry/.style={shape=circle, draw, inner sep=2pt,fill=black!40},
bulk/.style={shape=circle, draw, inner sep=2pt,fill=black},
bubu/.style={black},
bubo/.style={decorate,decoration=snake,red},
bobo/.style={decorate,decoration={coil,amplitude=0.1cm, aspect=1.2,segment length=0.2cm},purple},
bobo2/.style={decorate, decoration={zigzag,amplitude=0.2cm},blue}
}
\newtheorem{theorem}{Theorem}[section]

\newtheorem{conjecture}[theorem]{Conjecture}
\newtheorem{corollary}[theorem]{Corollary}

\newtheorem{definition}[theorem]{Definition}
\newtheorem{example}[theorem]{Example}

\newtheorem{lemma}[theorem]{Lemma}

\newtheorem{prop}[theorem]{Proposition}
\newtheorem{remark}[theorem]{Remark}

\newtheorem{assumption}[theorem]{Assumption}

\numberwithin{equation}{section}

\DeclareMathOperator{\tr}{tr}

\DeclareMathOperator{\proj}{q}

\def\C{\mathbb{C}}

\def\R{\mathbb{R}}

\def\A{\mathrm{A}}
\def\Aut{\operatorname{Aut}}
\def\exp{\operatorname{Exp}}

\def\Cn{C^{\infty}}

\def\de{\Delta}
\def\FS{{\mathcal{F}_{\Sigma}}}

\def\sp{{\mathrm{split}}}
\def\dt{{\mathrm{dt}}}
\def\d{{\mathrm{d}}}
\def\fp{{\mathrm{\textbf{f}}.\mathrm{\textbf{p}}.}}
\def\dvol{\mathrm{dVol}}
\def\Gr{\mathrm{Gr}}

\newcommand{\dd}{\partial}
\newcommand{\ra}{\rightarrow}
\newcommand{\mr}{\mathrm}

\newcommand{\RR}{\mathbb{R}}
\newcommand{\ZZ}{\mathbb{Z}}
\newcommand{\RG}{\mathcal{R}}
\newcommand{\Fun}{\mathrm{Fun}(\Sigma)}
\renewcommand{\exp}{\mathrm{exp}}
\newcommand{\DD}{\mathcal{D}}
\newcommand{\OO}{\mathrm{O}}

\newcommand{\til}{\widetilde}

\newcommand{\funlog}{C^\infty_\mathrm{adm}}
\newcommand{\pre}{\mathrm{pre}}

\begin{document}
\title{Two-dimensional perturbative scalar QFT and Atiyah-Segal gluing}

\author{Santosh Kandel}  

\address{University of Freiburg}
\email{skandel1@alumni.nd.edu}

\author{Pavel Mnev}

\address{University of Notre Dame}
\address{St. Petersburg Department of V. A. Steklov Institute of Mathematics of the Russian Academy of Sciences}
\email{pmnev@nd.edu}

\author{Konstantin Wernli}
\address{Centre for Quantum Mathematics - %SDU
University of Southern Denmark
}
\email{kwernli@imada.sdu.dk}

\thanks{
S.K. and K.W. would like to thank the University of Zurich where a part of this work was written, and acknowledge partial support of NCCR SwissMAP, funded by the Swiss National Science Foundation, and by the
COST Action MP1405 QSPACE, supported by COST (European Cooperation in Science and Technology), and the SNF grant No. 200020\_172498/1 during their affilitation with the University of Zurich. 
%P. M. acknowledges partial support of RFBR Grant No. 17-01-00283a.\\
K. W.  acknowledges further support from a BMS Dirichlet postdoctoral fellowship and the SNF Postdoc Mobility grant P2ZHP2\_184083, and would like to thank the Humboldt-Universit\"at Berlin, in particular the group of Dirk Kreimer, and the University of Notre Dame for their hospitality.
}

\maketitle

\begin{abstract}
We study the perturbative quantization of 2-dimensional massive scalar field theory with polynomial (or power series) potential on manifolds with boundary. We prove that it fits into the functorial quantum field theory framework of Atiyah-Segal.
In particular, we prove that the perturbative partition function defined in terms of integrals over configuration spaces of points on the surface satisfies an Atiyah-Segal type gluing formula. Tadpoles (short loops) behave nontrivially under gluing  and play a crucial role in the result.
\end{abstract}
\setcounter{tocdepth}{3}
\tableofcontents

%\newpage

%\listoftodos

%\newpage

\section{Introduction}

In recent years, \emph{Functorial Quantum Field Theories} (FQFTs), as proposed by Atiyah and Segal \cite{Atiyah},\cite{Seg1}, have been the subject of intense mathematical investigation, see e.g. \cite{CMRQ}, \cite{Reshetikhin2010}, \cite{Stolz} and references therein. The rough idea is that a quantum field theory corresponds to a functor 
$$Z \colon \mathbf{Cob} \to \mathbf{Hilb}$$
from a cobordism category, possibly equipped with extra structure, to the category of Hilbert spaces. Examples of such functors from \emph{topological} cobordism categories (called TQFTs, short for \emph{Topological Quantum Field Theories}) abound, see e.g. \cite{Reshetikhin1991}, \cite{Turaev1992}, \cite{Dijkgraaf1990}. 
%and have been studied in great detail over the last decades (see e.g. REFERENCE). 
On the other hand, there are very few examples known for geometric %Riemannian 
cobordism categories.
% arising from path integral quantization. 
The first examples (in dimension greater than one) are: 2-dimensional Yang-Mills theory (Migdal-Witten, \cite{Migdal1975},  \cite{Witten1991}) and 2-dimensional free fermion conformal field theory (Segal \cite{Seg1} and %James 
Tener \cite{Tener2017}) -- for the cobordism category endowed with area form or conformal structure, respectively. An example of an invertible FQFT for the Spin Riemannian cobordism category is constructed by Dai and Freed \cite{Daifreed1994}. 
%\marginpar{dubious reference?}
%\textcolor{blue}{In \cite{P}, it was shown that a certain class of two dimensional interacting scalar theories with polynomial interaction give rise to examples of such a FQFT} 
In \cite{Santosh}
it was shown that free massive scalar field theory provides an example of such a FQFT in  even dimensions, for the Riemannian cobordism category. \\

In this paper we give a new example of an interacting FQFT on the Riemannian cobordism category arising from the perturbative path integral.

\subsection{Main results}
We are considering the perturbative quantization of the scalar field theory defined classically by the action functional
$$ S_\Sigma(\phi)= \int_\Sigma \frac12 d\phi \wedge * d\phi + \int_{\Sigma}\frac{m^2}{2} \phi^2\, \dvol_\Sigma+ \int_{\Sigma}p(\phi)\, \dvol_\Sigma $$
with $\Sigma$ an oriented surface endowed with Riemannian metric, $\phi\in C^\infty(\Sigma)$ the field, $m>0$ a parameter (``mass'') and $p(\phi)=\sum_{ k\geq 3 }\frac{p_k}{k!}\phi^k$ a polynomial (or possibly power series) interaction potential.

The first main result of this paper is that there is an Atiyah-Segal gluing formula for the perturbative partition function (Definition \ref{def: Z tadpole}). The result is set up as follows. First, we define a vector space 
$$H_Y=\Big\{\Psi(\eta)=\sum_{n\geq 0}\int_{C^\circ_n(Y)}dx_1\cdots dx_n\, \psi_n(x_1,\ldots,x_n) \eta(x_1)\cdots \eta(x_n) \Big\}$$
 associated to a 1-dimensional Riemannian manifold $Y=S^1\sqcup\cdots \sqcup S^1$ (Definition \ref{def:bdry_space}). Vectors in $H_Y$ are functionals on $C^\infty(Y)\ni \eta$ and %\footnote{
% In the bulk of the paper it is denoted $\tilde\eta$; here we are suppressing the tilde. Up to a normalization factor, $\eta$ is the boundary value of the field $\phi$, $\phi|_Y=\sqrt\hbar\,\eta$.
 %} 
  are parameterized by the ``$n$-particle wave functions'' $\psi_n$ -- 
 formal power series in $\hbar^{1/2}$ with coefficients given by 
 smooth symmetric functions on the open configuration space of $n$ points on $Y$ with certain types of singularities on diagonals allowed (in particular, logarithmic singularities on codimension $1$ diagonals), see Definition \ref{def: adm singularities}. 
 %with values in $\RR[[\hbar^{1/2}]]$.  
   
  The perturbative partition function of a surface $\Sigma$  is then given as\footnote{
For the sake of preliminary exposition, we are being slightly imprecise writing that $Z_\Sigma$ is an element of the space of boundary states  -- see Section \ref{sec:Feynmanrules} for details (in particular, Remarks \ref{rem: Z is not in H but Zhat is in H}, \ref{rem: DN kernel singularity} and \ref{rem:Zfactors}). 
  }
  %\marginpar{\textcolor{brown}{We are cheating a bit: it is not vector in $H$, strictly speaking. Maybe it's ok to cheat here a bit, for a lighter exposition.}}
\begin{equation}\label{Z intro}
Z_\Sigma(\eta)=e^{-\frac{1}{2}\int_{\partial \Sigma} \dvol_{\partial \Sigma}\, \eta D_\Sigma(\eta)} {\det}^{-\frac{1}{2}}(\Delta+m^2)\sum_\Gamma \frac{\hbar^{E-N-\frac{n}{2}} F_\Gamma(\eta)}{|\Aut(\Gamma)|}  \quad \in H_{\partial \Sigma}
\end{equation}
where:
\begin{itemize}
\item $\hbar$ is a formal parameter of quantization.
\item $D_\Sigma$ is the Dirichlet-to-Neumann operator. %, mapping $\eta\mapsto \partial_n \phi_\eta$ where
\item $\det(\Delta+m^2)$ is the zeta-regularized determinant over functions on $\Sigma$ with Dirichlet boundary conditions.
\item The sum runs over graphs $\Gamma$ with $N$ bulk vertices and $n$ boundary vertices (which are univalent), with no boundary-boundary edges.\footnote{In fact, the exponential prefactor in (\ref{Z intro}) can be seen as the contribution of boundary-boundary edges -- see Remark \ref{rem:Zfactors}.} Here %$V=N+n$ the total number of vertices and 
$E$  is the total number of edges and $|\Aut(\Gamma)|$ is the order of the automorphism group of the graph.
\item The Feynman evaluation of a graph $F_\Gamma(\eta)$ is calculated as follows. 
\begin{itemize}
\item Boundary vertices of $\Gamma$ are placed at points $x_i\in \partial \Sigma$ and are decorated with $\eta(x_i)$; 
\item bulk vertices are placed at points $y_j\in \Sigma$ and are decorated by $-p_v$ (the coefficient of the interaction polynomial, with $v$ the valence of the vertex);
\item edges between distinct vertices are decorated by the Green's function for $\Delta+m^2$ (or %\textcolor{blue}{minus} 
its normal derivative for bulk-boundary edges), 
\item an edge connecting a vertex to itself is decorated by the zeta-regularized evaluation of the Green's function on the diagonal, see Definition \ref{def: tau zeta-reg}.
\end{itemize}
Then $F_\Gamma(\eta)$ is calculated as the product of all decorations integrated over positions of all points $x_i$, $y_j$.
\end{itemize}

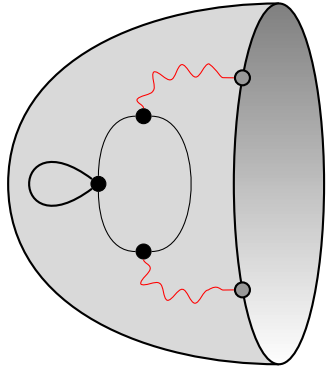
\begin{figure}[h!]
\centering 
\begin{tikzpicture}[baseline=(z2.base),scale=0.6] %cup

    \coordinate (z0) at (1,4);
    \coordinate (z1) at (1,-4);
    \coordinate (z2) at (-5,0);
    
    \filldraw[fill=black!15, thick, draw=black] (z0) to[out=180, in=90] (z2) to[out=-90, in=180] (z1); % to[out=0, in=0]  (z0);
    
    %\draw[thick, draw=black] (z0) to[out=240, in=120] (z1) to[out=, in=0] (z0) ;
    \shadedraw[ thick, draw=black] (2,0) arc [start angle=0,   
                  end angle=360,
                  x radius=1cm, 
                  y radius=4cm]
    node [bdry, pos=.4] (x1) {} node [pos = .6, bdry] (x2) {};
    \node[bulk] (b1) at (-2,1.5) {} edge[bubo, out=90, in=180] (x1);
    \node[bulk] (b2) at (-2,-1.5) {} edge[bubo, out=270,in=180] (x2);
    \node[bulk] (b3) at (-3,0) {};
    \draw[bubu] (b1) to[out=180,in=90] (b3);
       \draw[bubu] (b3) to[out=270,in=180] (b2);
    \draw[bubu] (b1) to[out=0,in=360] (b2);
    \draw[thick] (b3) .. controls (-5,1.5) and (-5,-1.5) .. (b3);
%     \node (k1) at (1,0) {BF};
  \end{tikzpicture}
  \caption{A typical Feynman graph on a hemisphere. Straight lines are decorated with the Green's function (for Dirichlet boundary conditions) while wavy lines are decorated with its normal derivative. }
  \end{figure}  
%%%%%
%  then lives in the vector space associated to its boundary: $Z_\Sigma \in H_{\partial\Sigma}$. 
  
  Finally, if $\Sigma$ is a closed surface with a decomposition $\Sigma = \Sigma_L \cup_Y \Sigma_R$, where $Y$ is the boundary of $\Sigma_L$ and $\Sigma_R$, we define a pairing 
\begin{equation}\label{intro <> DN}
\langle\cdot,\cdot\rangle_{\Sigma_L,Y,\Sigma_R} \colon %H_{\Sigma_L} \otimes H_{\Sigma_R} \to  H_\Sigma
H_Y\otimes H_Y\rightarrow \mathbb{R}[[\hbar^{1/2}]]
\end{equation}
as follows: 
\begin{multline}\label{intro <,> DN explicit}
 \langle \Psi^L(\eta),\Psi^R(\eta) \rangle 
 ={\det}^{-\frac12}(D_{\Sigma_L}+D_{\Sigma_R}) \cdot \\
\cdot \sum_{\mathfrak{m}\in \mathfrak{M}_{m+n}} \int_{C^\circ_{m+n}(Y)}dx_1\cdots dx_{m+n}\,\psi^L_m(x_1\ldots,x_m) \psi^R_n(x_{m+1},\ldots,x_{m+n}) \prod_{(i,j)\in\mathfrak{m}}K(x_i,x_j) 
 \end{multline}
%(Definition \ref{def:pairing}). 
Here $\psi^{L,R}$ are the wavefunctions for the states $\Psi^{L,R}$ and $K$ is the Green's function for the operator $D_{\Sigma_L}+D_{\Sigma_R}$. The sum runs over perfect matchings $\mathfrak{m}$ of $m+n$ elements. 

The Atiyah-Segal gluing formula can then be formulated as follows. %(Theorem \ref{thm:ASgluing}):
\begin{theorem}\label{thm:ASgluingIntro}
Let $\Sigma = \Sigma_L \cup_Y \Sigma_R$. Then 
$$Z_\Sigma = \langle \widehat{Z}_{\Sigma_L},\widehat{Z}_{\Sigma_R}\rangle_{\Sigma_L,Y,\Sigma_R}$$ 
where 
%\marginpar{\textcolor{pink}{$Z_L$,$Z_R$ don't appear in the theorem - what notation to keep?}}
%$\hat{Z}_{\Sigma_L}(\eta)= e^{\frac{1}{2\hbar}\int_Y \dvol_Y\, \eta D_{\Sigma_L}(\eta) }Z_{\Sigma_L}(\eta)$ and similarly for $\hat{Z}_{\Sigma_R}$. 
$\widehat{Z}_{\Sigma_L}$, $\widehat{Z}_{\Sigma_R}$ are given by (\ref{Z intro}) specialized to $\Sigma_{L}, \Sigma_R$ with the exponential prefactor omitted.
\end{theorem}

The proof is based, roughly, on the idea that the value of a Feynman graph $\Gamma$ on $\Sigma$ can be presented in terms of values of its subgraphs $\Gamma_L,\Gamma_R$ located on $\Sigma_L$ and $\Sigma_R$, glued using the interface Green's functions $K$, 
%(more precisely, one is summing over all the ways to cut $\Gamma$ into subgraphs $\Gamma_L$ and $\Gamma_R$). 
see Figure \ref{fig: gluing of graphs}. 
  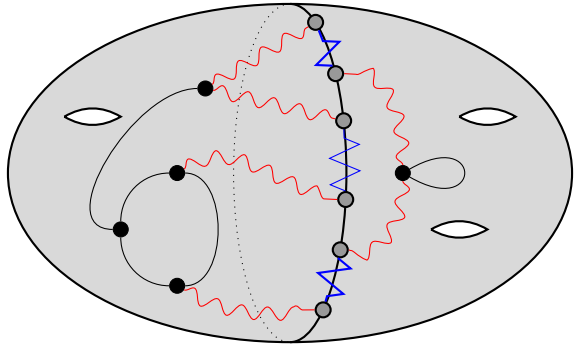
\begin{figure}[h!]
\centering 
\begin{tikzpicture}[scale=0.75] %cup
\filldraw[,fill=black!15, thick] (0,0) ellipse (5cm and 3cm); 
\draw[thick] (0,-3) arc (270:450:1cm and 3cm) 
%bdry vertices left
node[bdry, pos=.2] (bdry1) {} 
node[bdry, pos=.45] (bdry2) {} 
node[bdry, pos=.6] (bdry3) {}
 node[bdry, pos=.85] (bdry4) {}
%bdry vertices right 
node[bdry, pos=.35] (bdry5){} edge[bobo2] (bdry1) 
node[bdry,pos=.7] (bdry6) {} edge[bobo2] (bdry4);
%arc behind
\draw[dotted] (0,3) arc (90:270:1cm and 3cm);
%left bulk vertices
 \node[bulk] (b1) at (-2,0) {} edge[bubo, out=30, in=180] (bdry2);
 \node[bulk] (b2) at (-2,-2) {} edge[bubo, out=330,in=180] (bdry1);
 \node[bulk] (b3) at (-3,-1) {};
 \draw[bubu] (b1) to[out=180,in=90] (b3);
 \draw[bubu] (b3) to[out=270,in=180] (b2);
 \draw[bubu] (b1) to[out=0,in=360] (b2);
\node[bulk] (b4) at (-1.5,1.5) {} edge[bubo] (bdry3) edge[bubo] (bdry4);
   \draw[bubu] (b4) to[out=180,in=180] (b3);
   \draw[bobo2] (bdry2) -- (bdry3);
%right bulk vertex 
\node[bulk] (b1r) at (2,0) {} edge[bubo, in=0, out=270] (bdry5) edge[bubo, in=0,out=90] (bdry6);
\draw[tadpole] (b1r) to[out=30,in=330] (b1r);
%genus
\coordinate (x) at (3,1);
\coordinate (y) at (4,1);
\filldraw[fill=white, thick] (x) to[out=330,in=210] (y) to [out=150,in=30] (x);
\coordinate (y1) at (-3,1);
\coordinate (x1) at (-4,1);
\filldraw[fill=white, thick] (x1) to[out=330,in=210] (y1) to [out=150,in=30] (x1);

\coordinate (x) at (2.5,-1);
\coordinate (y) at (3.5,-1);
\filldraw[fill=white, thick] (x) to[out=330,in=210] (y) to [out=150,in=30] (x);
%\draw[thick] (2,0) arc (180:360:5mm) node[pos=0.2] (y) {} node[pos=0.8] (x) {};
%draw (x) to[out=120,in=30] (y);
\end{tikzpicture}
  \caption{ \label{fig: gluing of graphs} When gluing contributions from two Feynman graphs, their boundary vertices are connected by the Green's function of the Dirichlet-to-Neumann operator (zig-zag line).}
  \end{figure}

This result has a generalization for $\Sigma$ a non-closed surface, see Theorem \ref{thm:ASgluing}. 
%\\
%{\color{orange}In this case, $\widehat{Z}_{\Sigma_L}$ and $\widehat{Z}_{\Sigma_R}$ \marginpar{\textcolor{orange}{The notation as explained here contradicts the notation used later. I think we can remove this (orange) paragraph without any real loss of information.}} are understood as given by (\ref{Z intro}) without the exponential prefactor but with boundary-boundary edges in $\Gamma$ allowed, except for edges starting and ending on $Y$; in the pairing (\ref{intro <,> DN explicit}), one should replace $D_{\Sigma_L}$, $D_{\Sigma_R}$ with the ``$Y-Y$ blocks'' of the respective Dirichlet-to-Neumann operators. With these corrections, the Theorem above works in general case. 
%}\\

The second result is that the construction above can be upgraded to a functor from the symmetric monoidal semi-category (i.e. without identity morphisms) of Riemannian cobordisms to Hilbert spaces. The main problem obstructing functoriality of the previous result is the fact that the pairing (\ref{intro <> DN}) depends not just on the gluing interface $Y$ but also on the adjacent surfaces.

We define the adjusted partition function 
$$\overline{Z}_\Sigma(\eta)=e^{\frac{1}{2}\int_{\partial \Sigma} \dvol_{\partial \Sigma}\, \eta \varkappa(\eta)}Z_\Sigma(\eta)$$ 
%\textcolor{blue}{with} 
where 
$\varkappa=(\Delta+m^2)^{1/2}\big|_{\partial \Sigma}$  is the square root of the Helmholtz operator on the boundary. We also define a new adjusted (functorial) pairing $\langle,\rangle_{2\varkappa}$ on $H_Y$, given by a similar formula to (\ref{intro <,> DN explicit}) where the operator $D_{\Sigma_L}+D_{\Sigma_R}$ is replaced by $2\varkappa$. %We denote $\mathcal{H}_Y$ the $L^2$ completion of $H_Y$.
Furthermore, we replace $H_Y$ with its $L^2$ completion.
%\footnote{\textcolor{blue}{For the sake of simplicity of the exposition, here we blur the distinction between $H_Y$ and its $L^2$ completion.}}

\begin{theorem} \label{thm:functor intro} The assignment of the Hilbert space $H_Y$ to each closed Riemannian $1$-manifold $Y$ (endowed with a two-sided collar) and of the adjusted partition function $\overline{Z}_\Sigma$ to each Riemannian 2-cobordism constitutes a functor $%\mathbf{Cob}_2^\mathrm{Riem}
\mathbf{Riem}^2
\rightarrow \mathbf{Hilb}$.
\end{theorem}

%[REMARK ON TADPOLES?]

\subsection{Tadpoles}
In various treatments of scalar theory, tadpole diagrams were set to zero (this corresponds to a particular renormalization scheme -- in flat space, this is tantamount to normal ordering, see e.g. \cite{GJ}, \cite{SimonB}).
However, in our framework this prescription %explicitly 
contradicts locality in Atiyah-Segal sense, see Section \ref{sec: why tadpoles?}. 
%E.g. if we set the tadpole diagram to zero on a surface $\Sigma_1$ and then glue another surface $\Sigma_2$ its boundary, 
One good solution is to prescribe to the tadpole diagrams the zeta-regularized diagonal value of the Green's function. 
We prove that 
%that
assigning to a surface its zeta-regularized tadpole is compatible with locality, see  Proposition \ref{prop: tau reg locality}.
However there are other consistent prescriptions (for instance, the tadpole regularized via point-splitting and subtracting the singular term, see Section \ref{sec: point-splitting tadpole}). This turns out to be related to Wilson's idea of RG flow in the space of interaction potentials, see Section \ref{sec: RG flow}.

In the free theory, the zeta-regularized tapdole can be interpreted in terms of the trace of the (classical) stress-energy tensor.
Generally, the trace of the quantum stress-energy tensor 
$\big \langle \tr T_q(x) \big\rangle$ 
(the reaction of the partition function to an infinitesimal Weyl transform of the metric) differs from the expectation value of the classical stress-energy tensor by a ``trace anomaly'' -- an effect well-known in conformal field theory. In Appendix \ref{app: trace anomaly} we obtain an expression for the trace anomaly in the interacting massive scalar theory 
(Proposition \ref{prop: trace anomaly}):
$$ 
\big \langle \tr T_q(x) \big\rangle - \big\langle \tr T_\mr{cl}(x) \big\rangle = \frac{\hbar}{4\pi}\Big\langle \frac{K(x)}{6}-m^2-\frac{\dd^2}{\dd \phi^2}p(\phi) \Big\rangle
$$
with $K(x)$ the scalar curvature of $\Sigma$ at the point $x$. Brackets  $\langle\cdots\rangle$  
stand for the expectation value defined using the perturbative path integral, as a sum of connected graphs with a single marked vertex (except in 
$\big \langle \tr T_q(x) \big\rangle$ where brackets are a part of notation).
We also comment on how to compare the $m \to 0$ limits of trace anomaly and partition function with conformal field theory, see Remark \ref{rem:comparison to CFT}.

\subsection{Plan of the paper}
Let us briefly outline the plan of the paper. 
\begin{itemize}
\item In Section \ref{sec:CFT} we recall some facts on classical field theory and free massive scalar field theory.
%In Section \ref{sec:Green} we review some facts on Green's functions of elliptic operators on manifolds with boundary. 
\item In Section \ref{sec:pert_quant}, we define the perturbative quantization of scalar field theory on manifolds with boundary\footnote{Our approach is slightly different from the one in \cite{CMRQ}: 
we use the second-order formalism whereas \cite{CMRQ} works in the first order. Also, a technical point: we use a different extension of boundary fields into the bulk, for the splitting of fields as boundary fields plus fluctuations, -- we use the extension as a solution of equations of motion whereas \cite{CMRQ} uses a ``discontinuous extention.'' 
%splitting of fields as boundary fields extended into the bulk plus fluctuations (i.e., we use a different ).} 
We plan to discuss the relation between the two approaches in a future publication.}. We use heuristics of path integrals to motivate our construction and then give rigorous definitions and proofs.
\item In Section \ref{sec:gluing}, we show how to heuristically derive gluing formulae for regularized determinants and Green's functions from formal Fubini theorems for path integrals. These gluing formulae have been proven in other contexts in the literature, and we briefly review these mathematical results.
\item In Section \ref{sec:tadpoles} we study the regularization of tadpole diagrams and how it interacts with gluing.
\item In Section \ref{sec:Fubini}, we state and prove Theorem \ref{thm:ASgluingIntro} above (in the general form). 
\item In Section \ref{sec:functor}, we promote our results to the functorial framework and state and prove Theorem \ref{thm:functor intro} above.  
\item In Appendix \ref{App: Examples}, we present a collection of explicit examples of zeta-regularized determinants and tadpole functions (and gluing thereof) and Dirichlet-to-Neumann operators.  
%We also prove a (sharp) bound on the order of the difference of the Dirichlet-to-Neumann operator and the square root of Helmholtz operator on the boundary: it turns out, this difference is a pseudodifferential operator of order $\leq -2$ (Proposition \ref{prop:delta regularity}).
We prove that the Dirichlet-to-Neumann operator differs from the square root of the Helmholtz operator on the boundary by a pseudodifferential operator of order $\leq -2$ (a sharp bound) -- Proposition \ref{prop:delta regularity}. 
\item In Appendix \ref{app: trace anomaly} we discuss the relation of the tadpole and the stress-energy tensor, and obtain the trace anomaly.
%we give the interpretation of the tadpole as the expectation value of the trace of the classical stress-energy tensor (in free case). We also obtain, in the interacting theory, the general formula for the ``trace anomaly'' -- the difference between the trace of quantum stress-energy tensor (the reaction of the partition function to an infinitesimal Weyl transform of the metric) and the expectation value of the classical stress-energy tensor.
\end{itemize}

\subsection{Related work}
%\marginpar{Is this too short? Do you have ideas on how to expand this?}
%There is a literature on cutting-gluing in 
The cutting and gluing of perturbative partition functions in the context of first-order gauge theories is discussed in \cite{CMRQ}. An example of a computation in the context of Chern-Simons theory was done in \cite{Wernli2018}, \cite{CMW1}. %Computing 
An approach to 
perturbative Chern-Simons invariants of 3-manifolds via cutting and gluing in the context %from the viewpoint 
of algebraic topology was discussed in \cite{KT}. Another approach to gluing of functional integrals - with a view towards supersymmetry -  is considered in \cite{D1},\cite{D2}. Gluing in the context of the perturbative path integral approach to quantum mechanics was considered in \cite{JF}.

\subsection{Acknowledgements}
We would like to thank 
Alberto S. Cattaneo, Liviu Nicolaescu, Nicolai Reshetikhin, Stephan Stolz and Donald Youmans for helpful remarks and discussions.
%Liviu Nicolaescu for helpful remarks  and Stephan Stolz for discussions on Riemannian cobordism category. 
%P.M. would like to thank Alberto S. Cattaneo and Nicolai Reshetikhin for inspiring discussions on the scalar field QFT in the context of BV-BFV quantization around 2012--13.

\section{%Classical scalar field theory
Preliminaries
}\label{sec:CFT}
\subsection{Classical field theory}
A classical (Lagrangian) field  theory on a compact oriented $d$-dimensional Riemannian manifold $\Sigma$, possibly with nonempty boundary, consists of:
\begin{enumerate}[(i)]
\item The space of fields $\FS$, which is the space sections of a vector bundle\footnote{
Or, more generally, a sheaf. However, the case relevant for this paper is the one of the trivial $\R$-bundle.}
over $\Sigma$. 
\item
A Lagrangian $\mathcal{L}_{\Sigma}$, which is a function on $\FS$ with values in the space of $d$-forms on $\Sigma$, i.e. $\mathcal{L}_{\Sigma}:\FS\to \Omega^d(\Sigma), \phi\mapsto \mathcal{L}_{\Sigma}(\phi)$ such that it is \emph{local} in the sense that it depends only on the fields and finitely many derivatives thereof. 
\end{enumerate}
Using the Lagrangian, we define the action functional $S_{\Sigma}:\FS\to\R, S_{\Sigma}(\phi)=\int_{\Sigma}\mathcal{L}_{\Sigma}(\phi).$ 

The variation $\delta S_{\Sigma}$ of the action functional $S_{\Sigma}$ has the so-called Euler-Lagrange term and a boundary term. The Euler-Lagrange term gives rise to the equation of motion and its solutions are called the classical solutions. Let $EL_{\Sigma}\subset \FS$ denote the space of the classical solutions. The boundary term induces a one-form on $EL_{\Sigma}$. 

Let $\textbf{Y}$ denote  a $(d-1)$-dimensional manifold $Y$ together with a one-sided $d$-dimensional collar around $Y$. Associated to $Y$ we have a space $\Phi_{Y}$ which consists of the solutions to the equation of motion on $\textbf{Y}$, and a one-form $\alpha_{Y}$ on $\Phi_{Y}$ that arises from the boundary contribution of the variation %of 
$\delta S_{\textbf{Y}}$ on $\mathcal{F}_{\textbf{Y}}$. We can identify $\Phi_Y$, at least when $\mathcal{L}_{\Sigma}$ is nice, with the Cauchy data space $C_{Y}$ which is the information on the fields and their derivatives along $Y$ so that the equation of motion has a unique solution. Let $\omega_{Y}=\delta\alpha_{Y}$, then, $\omega_{Y}$ defines a presymplectic structure on $C_{Y}$ and it is symplectic\footnote{$C_Y$ is typically infinite-dimensional in the setting of field theory, and nondegeneracy of symplectic structure is understood in the weak sense, i.e. the map $TM \to T^*M$ induced by $\omega$ is injective, but not necessarily surjective.} in nice situations. For instance, for the free scalar field theory $(C_Y, \omega_Y)$ is a symplectic vector space.  

Now assume that $\partial\Sigma =Y$, then we have a surjective submersion $\pi_{\Sigma}:\FS\to C_{Y}.$ Let ${L}_{\Sigma}=\pi_{\Sigma}(EL_{\Sigma}).$ Then,  ${L}_{\Sigma}$ is an isotropic submanifold of $C_{Y}$ and it will be a Lagrangian submanifold when $\mathcal{L}_{\Sigma}$ is nice. We assume that $L_{\Sigma}$ is Langangian for this discussion.

Hence, a 
classical %scalar 
field theory assigns to a compact oriented Riemannian $(d-1)$-dimensional manifold $Y$ (more precisely, $Y$ has a one-sided collar) a symplectic manifold $(C_{Y}, \omega_{Y})$. Moreover, if $\partial\Sigma=Y$, then $L_{\Sigma}$ is a Lagrangian submanifold of   $C_{Y}$. More generally if $\partial\Sigma= \overline{{\partial\Sigma^{\text{in}}}}\sqcup\partial\Sigma^{\text{out}}$, then $L_{\Sigma}$ is a Lagrangian submanifold of $\overline{C_{\partial\Sigma^{\text{in}}}}\times C_{\partial\Sigma^{\text{out}}}$ and it can be regarded as a relation which if often called a canonical relation \cite{Weinstein_Alan}. Here $\overline{\partial\Sigma^{\text{in}}}$ is used to emphasize that the intrinsic orientation on $\partial\Sigma^{\text{in}}$ is the opposite of the induced orientation from $\Sigma$ (whereas the bar in $\overline{C_{\partial\Sigma^{\text{in}}}}$ denotes the change of sign of the symplectic form). Furthermore, if $\Sigma=\Sigma_1\cup_{Y}\Sigma_2$, then $L_{\Sigma}=L_{\Sigma_1}\circ L_{\Sigma_2}$ where the $\circ$ means the composition of relations. We refer to  \cite{Weinstein_Alan} for a discussion of the symplectic category and %\cite{CMRC}
\cite{CMRCQ}  
 for a more elaborate discussion and axiomatization of classical field theory in a more general setting.

\subsection{Free massive scalar theory} 
Here, we illustrate the discussion above using the free massive scalar field theory. Let $\Sigma$ be a compact oriented Riemannian manifold of dimension $d$ with $\partial\Sigma=Y$. Let $m$ be a positive real number. For the massive free scalar field theory on $\Sigma$, the space of fields is $\mathcal{F}_{\Sigma}=C^{\infty}(\Sigma)$ and the Lagrangian is given by
\begin{equation*}\mathcal{L}(\phi)=\dfrac{1}{2}(d\phi\wedge \ast d\phi+m^2\phi\wedge\ast \phi) \hskip1.5mm\text{and}\hskip1.5mm
        S(\phi)=\frac12\int_{\Sigma}d\phi\wedge \ast d\phi+m^2\phi\wedge\ast \phi.
        \end{equation*} Moreover, 
\begin{eqnarray*}\delta S=\int_{\Sigma}\delta\phi\wedge (d\ast d\phi+m^2\ast\phi)+\int_{\partial\Sigma}\delta\phi\wedge\ast d\phi 
\end{eqnarray*}
and  the equation of motion is 
\begin{eqnarray*}(\Delta_{\Sigma}+m^2)\phi=0
\end{eqnarray*} which is also known as \textit{Helmholtz} equation. Here, and throughout the paper, $\Delta_\Sigma$ is the Laplace-de Rham operator $ \Delta = dd^* + d^*d$ (where $d^*$ is the codifferential) restricted to 0-forms.\footnote{In particular, it has nonnegative spectrum, so coincides with \emph{minus} the usual Laplace operator on flat space.} The Cauchy data space $C_{Y}$ is given by $\Cn(Y)\oplus\Cn(Y)$, and $\pi_{\Sigma}:\FS\to C_{Y}$ is given by $\phi\mapsto \left(\iota_{Y}^*(\phi), \iota_{Y}^*\left(\dfrac{\partial\phi}{\partial\nu}\right)\right),$ where $\nu$ is the outward pointing unit normal vector field along $Y.$ Furthermore, $\alpha_{Y}=\int_{Y}\chi\delta\phi\, \dvol_Y$ and    the symplectic form $\omega_Y$ is given by $\omega_Y = \delta\alpha_Y = \int_Y \delta \chi \delta \phi\, \dvol_Y$. Here $\delta$ is understood as de Rham differential on $C^\infty(Y)$, this means that when evaluated on two vectors $(\phi_1,\chi_1)$ and $(\phi_2,\chi_2)$ the result is 
\begin{equation*}
\omega_Y((\phi_1,\chi_1)(\phi_2,\chi_2)) = \int_Y( \chi_1\phi_2 - \phi_1\chi_2)\,\dvol_Y.
\end{equation*} 
$L_{\Sigma}$ is the graph of the Dirichlet-to-Neumann operator $D_{\Sigma}$ on $Y$, which is defined as follows:

\begin{definition} Let $\eta\in\Cn(Y)$ and $\phi_{\eta}\in\Cn(\Sigma)$ be the solution of the Helmholtz equation on $\Sigma$ with  $\iota_{Y}^*(\phi_{\eta})=\eta$ given by Lemma \ref{BVP}. Then we define $D_\Sigma\colon C^\infty(Y) \to C^\infty(Y)$ by 
 \[D_{\Sigma}(\eta):=\iota_{Y}^*\left(\dfrac{\partial\phi_{\eta}}{\partial\nu}\right)
 .\]
 \end{definition} It is known that $D_{\Sigma}$ is symmetric from which it follows that $L_{\Sigma}$ is Lagrangian. Furthermore, if $\Sigma=\Sigma_1\cup_{Y}\Sigma_2$, one can verify that $L_{\Sigma}=L_{\Sigma_1}\circ L_{\Sigma_2}$ \cite{Santosh, CMRC}. 

\begin{remark} \label{rem:DtoNgluing} When $\Sigma=\Sigma_1\cup_{Y}\Sigma_2$, the Dirichlet-to-Neumann operator $D_{\Sigma_1,\Sigma_2}$ along $Y$ is defined as the sum of  normal derivatives, with respect to the induced orientations on $Y$, of solutions of Helmholtz equation on $\Sigma_1$ and $\Sigma_2$: 
%\textcolor{blue}{extending the same $\eta$ on $Y$}:   
\begin{equation}
D_{\Sigma_1,\Sigma_2} = D_{\Sigma_1} + D_{\Sigma_2}. \label{eq:DtoNgluing}
\end{equation}
\end{remark}

\subsection{Green's functions}\label{sec:Green}
Let $\Sigma$ be a closed oriented Riemannian manifold and $P$ be an elliptic differential operator on $\Sigma$ such that $P$ is invertible, then $P$ has a unique Green's function $G(x,y)$, see for example \cite{TayII}, Chapter 7: 

The PDE 
%\marginpar{\textcolor{blue}{Why are we applying the operator in the second argument? Is that the usual convention? Also: are we assuming $P$ to be self-adjoint/symmetric with given b.c.? -PM}}
\begin{eqnarray*} P_{y}G(x, y)=\delta_{x}(y)
\end{eqnarray*} 
has unique distributional solution $G(x,y)$ and it is the integral kernel of $P^{-1}.$  Moreover, \[G(x,y)\in C^{\infty}\left(\Sigma\times\Sigma\setminus \mathrm{diag}\right).\]

%\paragraph{Green's function on manifolds with boundary} 
More generally, if $\Sigma$ is a compact oriented manifold with boundary, then one can define Green's function by imposing boundary conditions. For example, for the Dirichlet boundary condition, we have the following definition.

\begin{definition}\label{Green's funcn} Let $\Sigma$ be a compact manifold with $\partial\Sigma\ne\varnothing$ and $P$ be an elliptic operator on $\Sigma.$ Then the boundary value problem 
\begin{eqnarray*} P_{y}G(x, y)=\delta_{x}(y)
\end{eqnarray*} with $G(x,y)=0$ on $\partial\Sigma$ has a unique distributional solution $G(x,y).$ We call such a $G(x,y)$ Green's function with Dirichlet boundary condition and denote it by $G_{\Sigma}^{D}(x,y).$
\end{definition} 

\begin{remark} Green's functions for other boundary conditions are defined similarly and Green's functions may not be unique for a general boundary condition.
\end{remark} 

In this paper, we are mostly interested in the Green's function of $\Delta_{\Sigma}+m^2$ where $\Delta_{\Sigma}$ is the nonnegative Laplacian on $\Sigma$. One well-known technique to construct Green's functions of an elliptic operator $P$ on a manifold with boundary is the method of images (see for example, \cite{GJ}) which we describe below. 

Let $\Sigma$ be a smooth compact oriented Riemannian manifold with boundary and $\partial\Sigma=\partial_1\Sigma\sqcup\partial_2\Sigma$ We want to construct a Green's function that satisfies the Dirichlet boundary condition on $\partial_1\Sigma$ and the Neumann boundary condition on $\partial_2\Sigma.$ The idea here is to use the ``doubling twice trick" (we took this name from \cite{CMRQ}): we first glue a copy of $\Sigma$ along $\partial_1\Sigma$ and denote the resulting manifold by $\Sigma'.$ Note that there is a canonical  isomorphism $S_1$ which is the reflection about $\partial_1\Sigma.$ Next, we glue $\Sigma'$ with itself to get a closed  Riemannian manifold $\Sigma''.$ Let $S_2$ denote the  reflection along the boundary of $\Sigma'.$ Since $\Sigma''$ is a closed manifold we have the Green's function $G''$ for $\Sigma''.$  We use $G''$  to define Green's function on $\Sigma$ with the desired properties. For this purpose, we define
\begin{eqnarray*} G(x,y)= G''(x,y)+G''(x,S_2(y))-G''(x,S_1y)-G''(x, S_2\circ S_1(y)).
\end{eqnarray*} Let us verify that $G(x,y)$ is indeed a desired Green's function.

\begin{lemma} G(x,y) is a Green's function for $P$ that satisfies Dirichlet boundary condition on $\partial_1\Sigma$ and the Neumann boundary condition on $\partial_2\Sigma.$
\begin{proof} Let $x\in \Sigma.$ Then $P_yG(x,y)=\delta_{x}(y)$ as $\delta_{x}(S(y))=0$ for $y\in \Sigma.$ By construction $G(x,y)$ satisfies the Dirichlet boundary condition on $\partial_{1}\Sigma$ and the Neumann boundary condition on $\partial_{2}\Sigma.$
\end{proof}
\end{lemma}

\begin{example} Consider $\Sigma$ to be the first quadrant in the $\R^2$ and $P=\Delta$. We want a Green's function that satisfies the Dirichlet boundary condition on $y=0$ and the Neumann boundary condition on $x=0.$ In this case $\Sigma''$ is the whole $\R^2.$ We know that $-\frac{1}{2\pi} \log((d(x,y),(\alpha,\beta))$ is the Green's function on $\R^2.$ Thus, in this case, 
\begin{eqnarray*} G((x,y),(\alpha,\beta))=-\frac{1}{2\pi} \log((d((x,y),(\alpha,\beta)))-\frac{1}{2\pi} \log((d(x,y),(-\alpha,\beta))\\
+\frac{1}{2\pi} \log((d(x,y),(\alpha,-\beta))+\frac{1}{2\pi} \log((d(x,y),(-\alpha,-\beta))
\end{eqnarray*}
One can easily check $G$ has the desired properties. 
\end{example}

From now onward a Green's function always refers to a Green's function associated to $\Delta+m^2$. The following fact about Green's functions on a compact oriented Riemannian manifold $\Sigma$ is well known, see for example \cite{TayII}, Chapter 7; the item (\ref{item: Green fun asymptotics}) follows from the expansion of the Green's function near the diagonal (see for example \cite{Ludewig2017}).
\begin{lemma}\label{DGreen's function} 
\begin{enumerate}[(i)]
\item The Green's function $G_{\Sigma}^D(x,y)$, satisfying Dirichlet boundary condition, is symmetric and it defines a positive bounded operator on $L^2(\Sigma).$
\item \label{item: Green fun asymptotics} In the case $\dim\Sigma=2$, in a neighborhood of the diagonal $\Delta(\Sigma)$ in $\Sigma\times\Sigma$, away from the boundary, we have
%{\color{cyan} we have  \marginpar{REMOVE THE COMPLICATED ONE?}
%\begin{equation}
%G_{\Sigma}^D(x,y)= \frac{1}{2\pi}K_0(m\cdot d(x,y))\mathsf{D}_{VVM}^{1/2} + \tilde{H}(x,y)\label{eq:Gsing plus Gsmooth}
%\end{equation}
%where $K_0(z)$ is the modified Bessel function of the second kind and $\mathsf{D}_{VVM}(x,y)$ is the Van Vleck-Morette determinant\footnote{Modified Bessel functions of the second kind are defined by $K_{\nu}(z) = \lim\limits_{\alpha \to \nu} \frac{\pi}{2\sin \alpha \pi}(I_{-\alpha} - I_{\alpha}),  $
%where $I_\alpha = \sum_{m=0}^\infty \frac{1}{m!\Gamma(m + \alpha +1)}\left(\frac{x}{2}\right)^{2m+\alpha}= i^{-\alpha}J_\alpha (ix)$ is the modified Bessel function of the first kind (and $J_\alpha$ is the usual Bessel function of the first kind). The Van Vleck-Morette determinant is defined by $\mathsf{D}_{VVM}(x,y) = \frac{1}{\sqrt{g(x)}\sqrt{g(y)}}\det\left(\frac{-\partial^2 \sigma(x,y)}{\partial x \partial y}\right)$ with $\sigma(x,y) = \frac{1}{2}d(x,y)^2$.} and $\tilde{H}(x,y)$ is a $C^\infty$ function in $\Sigma \times \Sigma$. In particular, } we have 
\[G_{\Sigma}^D(x,y)=
%                    \begin{cases}
                    -\frac{1}{2\pi}\log(d(x,y))+H(x,y)   %&\text{if}\hskip2mm \mathrm{dim}(\Sigma)=2
                    %\\                                   Kd(x,y)^{2-\mathrm{dim}(\Sigma)}+H(x,y)&\text{if}\hskip2mm \mathrm{dim}(\Sigma)\geq3
%                    \end{cases}
,\] where %$K>0$ and 
$H$ is in  %continuous 
$C^1$.
%function in a neighborhood of the diagonal in \textcolor{cyan}{$\Sigma\times\Sigma$}. %\marginpar{\textcolor{blue}{CHECK THIS!} Done}
\end{enumerate}
 
\end{lemma}

We can use Green's function on a manifold with boundary to construct solutions to boundary value problems \cite{F2}:
\begin{lemma}\label{BVP} Let $\eta\in \Cn(\partial\Sigma).$ Define $\phi^\Sigma_{\eta}$ on $\Sigma$ by 
\[\phi^\Sigma_{\eta}(x)= -
\int_{\partial\Sigma}\frac{\partial G_{\Sigma}^{D}(x,y)}{\partial\nu}\eta(y)\,dy,\] then $\phi^\Sigma_{\eta}$ is the unique solution of the Dirichlet boundary value problem for $\Delta_{\Sigma}+m^2$ on $\Sigma$ with boundary value $\eta.$ Moreover, $\phi_{\eta}$ is smooth on $\Sigma$. 
\end{lemma}
We will sometimes drop the subscript $\Sigma$ if it is clear from the context. We will also drop $D$ from $G_{\Sigma}^D$ because, in this paper, we consider the Green's function either on a closed manifold $\Sigma$ or with respect to the Dirichlet boundary condition when $\partial\Sigma$ is non-empty.

\section{Perturbative Quantization}\label{sec:pert_quant}

In this section we consider the perturbative quantization of scalar field theory with a potential $p \in  \RR[[\phi]]    
%C^\infty(\R)
$ %\marginpar{\textcolor{blue}{maybe $\RR[[\phi]]$ instead of $C^\infty(\RR)$}}
- i.e. the  evaluation of the partition function by formally applying the method of steepest descent. As usual, terms in the resulting power series are labeled by Feynman graphs that are evaluated according to Feynman rules. The result is a functional on the boundary fields (more precisely, on the leaf space of a polarization on $C_Y$). 

% From now on, unless stated otherwise, $\Sigma$ is assumed to be a \emph{two-dimensional} Riemannian manifold.

\subsection{Formal Gaussian integrals and moments}\label{sec:formalints}
The path integrals appearing in this paper are all integrals over vector spaces, and can be reduced to expressions of the form 
\begin{equation*}
\int_{\phi \in C^{\infty}(\Sigma,\partial \Sigma)} e^{-\frac{1}{2\hbar} (\phi, A \phi)}\phi(x_1)\cdots \phi(x_n)D\phi,
\end{equation*}
where $\Sigma$ is a Riemannian manifold, $C^{\infty}(\Sigma,\partial \Sigma)$ denotes smooth functions which vanish on the boundary, $A\colon C^{\infty}(\Sigma,\partial \Sigma)\to C^{\infty}(\Sigma,\partial \Sigma)$ is a linear operator,  $(\phi,\psi) = \int_\Sigma \phi \psi\,\dvol_\Sigma$, and $\hbar$ is a formal parameter. One way to define these integrals is just to simply postulate the rules for finite-dimensional Gaussian moments in infinite dimensions. We will very briefly review this idea, as it is essential to the paper. Details can be found in many places, for instance \cite{Polyak2005},\cite{Reshetikhin2010},\cite{Mnev2019}.
 For $n=0$, we want to define the ``formal Gaussian integral''
\begin{equation}
\int_{\phi \in C^{\infty}{(\Sigma,\partial \Sigma)}} e^{-\frac{1}{2\hbar} (\phi, A \phi)}D\phi := \frac{1}{(\det A)^\frac12} \label{eq:InfGaussian}
\end{equation}
Heuristically this defines a certain normalization of the path integral measure (absorbing an infinite power of $2\pi\hbar$). However, since $A$ is an operator on an infinite-dimensional space, we need to be careful about the determinant. In this paper, %added Hawking's reference here
following \cite{hawking}, we will use the zeta-regularization.  \\ %changed stick to to use
The values of Gaussian moments can be elegantly given using the notion of perfect matchings:
\begin{definition} If $S$ is a set, a \emph{perfect matching} on $S$ is a collection $\mathfrak{m}$ of disjoint two-element subsets of $S$ such that $\bigcup \mathfrak{m} = S$.  The set of perfect matchings on $\{1,\ldots,n\}$ is denoted $\mathfrak{M}_n$.
\end{definition}

For instance, $\mathfrak{m} = \{\{1,3\},\{2,4\}\}$ is a perfect matching on $S=\{1,2,3,4\}$. 
Again, simply extending the finite-dimensional result to infinite dimensions yields the following definition:
\begin{equation}\int_{\phi \in C^{\infty}{(\Sigma,\partial \Sigma)}} e^{-\frac{1}{2\hbar} (\phi, A \phi)}\phi(x_1)\cdots \phi(x_n)D\phi := \frac{\hbar^n}{(\det A)^\frac12}\sum_{\mathfrak{m} \in \mathfrak{M}_n} \prod_{\{i,j\}\in\mathfrak{m}}A^{-1}(x_i,x_j)\label{eq:InfGaussianMom}
\end{equation}
where $A^{-1}$ is the integral kernel of the inverse of $A$. This definition works fine as long as $x_i \neq x_j$ for all $i\neq j$, but if two $x_i$'s coincide, we run into trouble because  $A^{-1}$ is typically singular on the diagonal. In this section, we will resolve this issue by \emph{normal ordering}, which has the effect of neglecting any terms containing $A^{-1}(x_i,x_i)$. However, for the purpose of gluing, we will need to resort to another mechanism explained in Section \ref{sec:tadpoles}. \\
A standard combinatorial argument, for which we again refer to the literature (e.g. the references above), then shows that one can conveniently label all terms in integrals such as 
\begin{equation*}
\int_{C^\infty(\Sigma,\partial \Sigma)} e^{-\frac{1}{2\hbar} (\phi, A \phi) - \frac{1}{\hbar}p(\phi)},
\end{equation*}
where $p$ is a polynomial, by graphs. These graphs are called Feynman graphs and the rules to evaluate them are called Feynman rules. Below, we will define the path integrals in question through these graphs and rules.

\subsection{The path integral on a manifold with boundary}
\subsubsection{The general picture}
Let $\Sigma$ be a compact oriented Riemannian manifold with $\partial\Sigma=Y$. We may have $Y=\varnothing.$ 
{  
Recall from Section \ref{sec:CFT} that we then have a presymplectic manifold $(C_Y, \omega_Y)$. Let us assume it is symplectic. 
}For the quantization we need some extra data, namely a polarization $P_{Y}$ of $C_{Y}$. We assume that $P_{Y}$ is such that the space of leaves $B_{Y}$ of the associated foliation is a smooth manifold. Let $\proj_{Y}:C_{Y}\to B_{Y}$ be the quotient map. If $\pi_{\Sigma}^{-1}(\proj_{Y}^{-1}(\eta))\cap EL_{\Sigma}$ is a finite set, then the  formal expression
\begin{eqnarray}
\begin{array}{lll} Z_{\Sigma}(\eta)=\int_{\pi_{\Sigma}^{-1}(\proj_{Y}^{-1}(\eta))} e^{-\frac{S(\phi)}{\hbar}}\, D\phi 
\end{array}
\label{pertdef}
\end{eqnarray}
can be defined using the formal version of the method of steepest descent.  \\

\subsubsection{Free scalar theory}
 Let us again consider our main example, the free massive scalar field theory defined by the action 
$$ S(\phi)=\frac12\int_\Sigma  d\phi\wedge *d\phi +m^2\phi \wedge *\phi $$ 
     We recall that 
\begin{align*} C_Y &= C^{\infty}(Y) \oplus C^\infty(Y)\ni(\eta,\psi), \\ \pi_\Sigma(\phi) &= \Big(\iota_Y^*(\phi),
 \iota^*_Y \Big(\frac{\partial \phi}{\partial \nu}\Big)\Big), \\\omega_Y &= \int_Y \delta\eta\delta\psi. 
 \end{align*} 
In a symplectic vector space, a nice class of polarizations is given by Lagrangian subspaces. In particular, the splitting $C_Y = C^{\infty}(Y) \oplus C^\infty(Y)$ is Lagrangian, so that there are two obvious polarizations on $C_Y$. 
%\marginpar{\textcolor{blue}{Why ``for now'' - do we ever use any other polarization? -PM}} 
We will choose the polarization for which $q_Y$ is the projection on the first component. Thus, for $\eta \in C^\infty(Y)$, we have $$\pi_\Sigma^{-1}(q^{-1}_Y(\eta)) = \{\phi \in C^\infty(\Sigma),\iota_Y^*(\phi) = \eta\}$$ and $$\proj_{Y}^{-1}(\pi_{\Sigma}^{-1}(\eta))\cap EL_{\Sigma}=\{\phi_{\eta}\}$$ where $\phi_{\eta}$ is the unique solution of the Dirichlet problem for $\Delta_{\Sigma}+m^2$ with boundary value $\eta$. The assignment $s\colon\eta\mapsto \phi_{\eta}$,  defines a section of the short exact sequence of vector spaces
$$\begin{tikzcd}
{C^\infty(\Sigma,Y)} \arrow[r, hook] & C^\infty(\Sigma) \arrow[r, "q_Y \circ \pi_\Sigma", two heads] & C^\infty(Y) \arrow[l, "s", bend left]
\end{tikzcd} $$ Hence, we can write $\phi=\hat{\phi}+\phi_{\eta}$ where $\hat{\phi}$ vanishes on $Y.$ Moreover, $S(\phi)=S(\hat{\phi})+S(\phi_{\eta}).$ We can then rewrite\footnote{Assuming translation invariance of the functional measure $D\phi$.} the formal expression (\ref{pertdef}) as follows: 
$$\int_{\pi_{\Sigma}^{-1}(\proj_{Y}^{-1}(\eta))} e^{-\frac{S(\phi)}{\hbar}}\,D\phi = e^{-\frac{S(\phi_\eta)}{\hbar}} \int_{\pi_{\Sigma}^{-1}(\proj_{Y}^{-1}(0))}  e^{-\frac{S(\hat{\phi})}{\hbar}}\,D\hat{\phi}$$ 
The latter expression is, formally, a Gaussian integral over the vector space $C^{\infty}_0(\Sigma)$ functions which vanish on the boundary of $\Sigma$. Thus it makes sense to define it, analogously to the finite-dimensional case, as 
\begin{eqnarray}
\label{eq:partition_function_free_theory_with_boundary} Z_{\Sigma}(\eta)=\int_{\pi_{\Sigma}^{-1}(\proj_{Y}^{-1}(\eta))} e^{-\frac{S(\phi)}{\hbar}}\,D\phi:=(\det (\Delta_{\Sigma}^{}+m^2))^{-\frac12}e^{-\frac{S(\phi_{\eta})}{\hbar}}
\end{eqnarray}
where, %deleted Hawkin's reference from here and moved to the section 3.1
$S(\phi_{\eta})$ is given by 
\begin{eqnarray}\label{eq:boundary_action}S(\phi_{\eta})=\frac12\int_{\partial\Sigma}\eta D_{\Sigma}\eta\, \dvol_{\partial\Sigma}
\end{eqnarray} and the Dirichlet boundary condition is used for the zeta-regularized determinant. In this paper, unless stated otherwise, the zeta-regularized determinant will be taken with respect to the Dirichlet boundary condition when the boundary is present.  %deleted $D$ from the zeta-regularized determinants and changed the text slightly. 

Later on, we will be interested in a decomposition of the boundary into two components, $\partial\Sigma = \partial_L\Sigma \sqcup \partial_R\Sigma$ (these components are allowed to be empty or disconnected). Then $\eta = \eta_L + \eta_R$, where $\eta_i$ is supported on $\partial_i\Sigma$. Since the Dirichlet-to-Neumann operator is symmetric, we can rewrite  
\begin{equation*}
S(\phi_\eta) = S(\phi_{\eta_L}) + S(\phi_{\eta_R}) + \int_{\partial \Sigma_L}\eta_LD_\Sigma \eta_R\,\dvol_{\partial\Sigma_L},
\end{equation*}
and the latter term can be expanded\footnote{To condense the notation, we let $\int_Yf(y)dy := \int_Yf(y) \,\dvol_Y(y)$.} as 
\begin{align*}
\int_{\partial \Sigma_L}\eta_LD_\Sigma \eta_R \, \dvol_{\partial\Sigma_L} &= 
-
\int_{(y,y')\in\partial \Sigma_L \times \partial\Sigma_R}\frac{\partial}{\partial\nu(x)}\frac{\partial G_{\Sigma}^{}(x,y)}{\partial\nu(y)}\eta_L(x)\eta_R(y)dydy' \\
&=: S_{L,R}(\eta_L,\eta_R) 
\end{align*}
Thus, the partition function can be expanded as 
\begin{align}
Z_{\Sigma}(\eta_L,\eta_R)&=\int_{\pi_{\Sigma}^{-1}(\proj_{Y}^{-1}(\eta))} e^{-\frac{S(\phi)}{\hbar}}\,D\phi \notag \\
&:=(\det (\Delta_{\Sigma}^{}+m^2))^{-\frac12}e^{-\frac{S(\phi_{\eta_L})}{\hbar}}e^{-\frac{S(\phi_{\eta_R})}{\hbar}}e^{-\frac{S_{L,R}(\eta_L,\eta_R)}{\hbar}}.\label{eq:Zfree} 
\end{align} %deleted $D$.

\subsubsection{Interacting theory}
From now on, unless stated otherwise, $\Sigma$ is assumed to be a \emph{two-dimensional} Riemannian manifold.

Let  
$$p(\phi) = \sum_{k\geq 0} \frac{p_k}{k!}\phi^k$$ 
be a formal power series. 
We are interested in the interacting massive scalar field theory where the Lagrangian  has the form 
\begin{equation*}L(\phi)=\dfrac{1}{2}(d\phi\wedge\ast d\phi+m^2\phi\wedge\ast\phi)+\ast p(\phi)
\end{equation*} so that the action functional is $S=S_0+S_{\text{int}}$ with
\begin{equation*}S_0=\frac12\int_{\Sigma}d\phi\wedge\ast 
d\phi+m^2\phi\wedge\ast\phi \hskip1.5mm\text{and}\hskip1.5mm S_{\text{int}}=\int_{\Sigma}\ast p(\phi).
\end{equation*}
We will consider this theory as a perturbation of the free theory, in order to define the perturbative partition function 
\begin{equation*}
Z_{\Sigma}(\eta,\hbar)=\int_{\pi_{\Sigma}^{-1}(\proj_{Y}^{-1}(\eta))} e^{-\frac{S_0(\phi)+ S_{\text{int}}(\phi)}{\hbar}}\, D\phi
\end{equation*}

Let $\eta\in \Cn(Y)$. The assignment $\eta\mapsto \phi_{\eta}$, where $\phi_{\eta}$ is the unique solution to the Dirichlet boundary value problem with boundary value $\eta$, defines a section of $q_{Y}\circ\pi_{\Sigma}:\FS\to B_{Y}.$ Hence, we can write $\phi=\hat{\phi}+\phi_{\eta}$ where $\hat{\phi}$ vanishes on $Y.$ Moreover, $S_0(\phi)=S_0(\hat{\phi})+S_0(\phi_{\eta}).$ Now, we can write 
\begin{equation}Z_{\Sigma}(\eta,\hbar)=e^{-\frac{S_0(\phi_{\eta})}{\hbar}}\int_{\pi_{\Sigma}^{-1}(\proj_{Y}^{-1}(0))} e^{-\frac{1}{\hbar}(S_0(\hat{\phi})+ S_{\text{int}}(\hat{\phi}+\phi_{\eta}))}\, D\hat{\phi}. \label{eq:formalpathint}
\end{equation}
The integral on the right hand side is again a formal path integral, and as such we would like to define as a formal perturbed Gaussian integral as explained above.     

There is a subtlety here: one would prefer the partition function to be a formal power series in $\hbar$, i.e. it should not contain negative powers of $\hbar$.\footnote{
%One reason for this is that multiplication and exponentiation of two formal Laurent series with infinitely many negative powers are not well defined.
For instance, because a product of two formal Laurent series with infinitely many negative powers is generally ill-defined. %Likewise, for $f(\hbar)\in \RR[[\hbar]]$, the expression  $e^{\hbar^{-1}f(\hbar)}$ is generally
} In the closed case, it is enough to assume $p_0 = p_1 = 0$ to achieve this. In the presence of boundary however, it is necessary to express the partition function instead in terms of the rescaled boundary field 
\begin{equation*}
\tilde{\eta}= \hbar^{-1/2}\eta \quad  \Leftrightarrow \quad  \eta = \hbar^{1/2}\tilde{\eta}.
\end{equation*}
In terms of $\tilde{\eta}$, \eqref{eq:formalpathint} reads
\begin{equation}Z_{\Sigma}(\tilde{\eta},\hbar)=e^{-S_0(\phi_{\tilde{\eta}})}\int_{\pi_{\Sigma}^{-1}(\proj_{Y}^{-1}(0))} e^{-\frac{1}{\hbar}(S_0(\hat{\phi})+ S_{\text{int}}(\hat{\phi}+\sqrt{\hbar}\phi_{\tilde{\eta}})}\, D\hat{\phi}. \label{eq:formalpathint2}
\end{equation}
For the rest of the paper, we will work with the rescaled boundary field $\tilde{\eta}$, and we will treat it as an element of $C^\infty(\partial\Sigma)$. Unless otherwise stated, we will assume $p_0 = p_1 = p_2 = 0$.\footnote{
This condition guarantees that there are finitely many Feynman diagrams contributing to each order in $\hbar$. We can also relax this assumption and allow nonzero $p_0,p_1,p_2$, see the discussion in Section \ref{sec:low valence vertices}. 

%
%\begin{itemize}
%\item $p_0$ term can be carried outside the path integral.
%\item $p_2$ adds a binary vertex to Feynman rules. Diagrams containing binary vertices can be resummed in terms of a theory with ``dressed propagator'' corresponding to the shifted mass.
%\item $p_1$ term generates trees which can be resummed and correspond to the perturbation theory around the shifted critical point (appearance of the linear term in the action implies that $\phi=0$ is no longer a critical point of $S$).
%\end{itemize}
}

\subsubsection{The space of boundary states - a perturbative model} %\marginpar{section title? we only introduce  ``pre-space of states here now''}
  Heuristically, the space of boundary states should be the space of square integrable functions %defined 
  on the space of leaves of the polarization: $H_Y = L^2(B_Y) = L^2(C^\infty(Y))$. Of course, in the field theory setting, one has to be very careful about the measure used to define these $L^2$ spaces, and in many cases it is convenient to drop the measure theory altogether and work with a different model for the space of states. %In this paper, we will present two different models. 
%  In this subsection, 
Here we will present (Definitions \ref{def:bdry_space}, \ref{def: H completion DN}) a natural model in the context of perturbative quantization, in which a gluing formula can be formulated. 
%(more precisely, we first introduce a ``pre-space of states'' and then the space of states proper as its completion with respect to certain pairing). 
In Section \ref{sec:functor}, we will revisit it from a more measure-theoretic perspective, and will introduce another pairing of states leading to a functorial interpretation of the gluing formula. 

%The perturbative model for the space of boundary states is constructed as follows. 
Let $C^\circ_n(Y)$ denote the open configuration space of $n$ points in $Y$: 
$$C^\circ_n(Y) = \{(y_1,\ldots,y_n)\in Y^n|y_i\neq y_j, \forall i \neq j\}.$$

To introduce the model for the space of states, we will need the following auxiliary definition (which is motivated by the properties of Feynman graphs, see Proposition \ref{prop: Feynman singularities} below).
\begin{definition}\label{def: adm singularities}
We say that a smooth function $f(y_1,\ldots,y_n)$ on the open configuration space $C^\circ_n(Y)$ has \emph{admissible singularities on diagonals} if:
\begin{enumerate}[(a)]
\item $f$ has (at most) logarithmic singularity when two points collide: 
$$f =  %\underset{y_i\ra y_j}{\sim}
 \mathrm{O}(\log d(y_i,y_j)) $$
 as $y_i\ra y_j$.
\item Assume that a set of $k\geq 3$ points $y_{i_1},\ldots,y_{i_k}$ coalesce, so that pairwise distances satisfy $\epsilon< d(y_{i_r},y_{i_s})<C\epsilon$  (with $C$ a constant). Then 
$$f=\mathrm{O}\Big(\frac{1}{\epsilon^{k-2}}\Big)$$ 
as $\epsilon\ra 0$.
\item Behavior near a general diagonal: assume that several disjoint subsets  $S_1,\ldots, S_p \subset \{1,\ldots,n\}$ of points coalesce at $p$ different points on $Y$, with each coalescing cloud of points at pairwise distances of order $\epsilon$. Then 
$$ f = \mathrm{O}\Big(\prod_{j=1}^p g_{|S_j|}(\epsilon)\Big) $$
as $\epsilon\ra 0$, 
where 
$$ 
g_k(\epsilon) = 
\begin{cases}
 \log\epsilon, \;\; k=2\\
\frac{1}{\epsilon^{k-2}},\;\; k\geq 3
\end{cases}
$$
\end{enumerate}%\marginpar{can alternatively denote $C^\infty_\mr{adm}$. Which is better?}
We will denote the space of smooth functions on $C^\circ_n(Y)$ with admissible singularities on diagonals allowed as $\funlog(C^\circ_n(Y))$.
%(with subscript $F$ for Feynman).
\end{definition}
We remark that all functions in $\funlog(C^\circ_n(Y))$ are integrable. However, they are generally not in $L^p$ for $p\geq 2$ due to the singularities arising at a collapse of $\geq 3$ points.

The perturbative model for the space of boundary states is constructed in two steps: first we introduce the ``pre-space of states'' and then (Definition \ref{def: H completion DN}) we will introduce the space of states proper as its appropriate completion.
%\marginpar{remove/rephrase this sentence (``pre-space'')?}

\begin{definition}\label{def:bdry_space}
Let $Y$ be a closed 1-dimensional manifold. 
%\begin{enumerate}
%\item 
For $n \in \mathbb{N}, n \geq 1$, we define $H^{(n)}_Y$ to be the space of    functionals $\Psi \colon C^{\infty}(Y) \to \R$ of the form 
\begin{equation}
\Psi(\tilde{\eta}) = \int_{C^\circ_n(Y)}\psi(y_1,\dots,y_n)\tilde{\eta}(y_1) \cdots \tilde{\eta}(y_n)\, dy_1\ldots dy_n,\label{eq:defalpha}
\end{equation}
 where $\psi$ is the ``wave function'' of the state $\Psi$ and it is a smooth, 
symmetric function on $C^\circ_n(Y)$ with admissible singularities (in the sense of Definition \ref{def: adm singularities}) allowed on diagonals.
 We say that $\psi$ represents $\Psi$. We are allowing the wave function $\psi$ to take values in formal power series $\RR[[\hbar^{1/2}]]$.
Moreover, we define $H^{(0)} = \mathbb{R}[[\hbar^{1/2}]]$. Finally, we define the pre-space of states\footnote{
This model for the pre-space of states is twice larger than what we need, in the following sense. The space $H^\mr{pre}_Y$ is $\frac12\mathbb{Z}\times  \mathbb{Z}$-bigraded by the $\hbar$-degree and polynomial degree in $\tilde\eta$, $d_1$ and $d_2$ respectively. The perturbative state induced from a surface on the boundary always has only monomials satisfying the ``selection rule'' $d_1-\frac{d_2}{2}\in \mathbb{Z}$.
}  
as 
$$H^\mr{pre} = \bigoplus_{n\geq0} H^{(n)}.$$ 

%\item 
%Define $\mathcal{H}^{(n)}(Y)=H^{(n)}(Y)\otimes \RR[[\hbar^{1/2}]]$ -- the version of the space of states where wave functions are allowed to take values in power series in $\hbar^{1/2}$.\footnote{
%This model for the space of states is twice larger than what we need, in the following sense. The space $\mathcal{H}_\mr{pre}(S^1)$ is $\frac12\mathbb{Z}\times  \mathbb{Z}$-bigraded by the $\hbar$-degree and polynomial degree in $\tilde\eta$, $d_1$ and $d_2$ respectively. The perturbative state induced from a surface on the boundary always has only monomials satisfying the ``selection rule'' $d_1-\frac{d_2}{2}\in \mathbb{Z}$.
%} 
% \textcolor{cyan}{We also set $\mathcal{H}_\mr{pre}(Y)=\bigoplus_{n\geq 0} \mathcal{H}^{(n)}(Y) = H_\mr{pre}(Y) \otimes \R[[\hbar^{1/2}]]$}.
% $$\mathcal{H}_{Y} = \left(\prod_{n \geq 0} H^{(n)}(Y)\right)[[\hbar^{1/2}]].$$
%\end{enumerate}
\end{definition} 
In particular, the vector space associated to the empty manifold $Y = \varnothing$ is just $\mathbb{R}[[\hbar^{1/2}]]$. 

%Here we defined the ``pre-space of states'': later (definition \ref{def: H completion DN}) we will introduce its appropriate $L^2$ completion -- the space of states proper.

%By usual arguments, the map $\psi \mapsto \Psi$ is bijective, so that we can identify 
%\begin{equation}
%H^n(Y) \cong (L^2(C^\circ_n(Y)))^{S_n}.
%\end{equation}

%We denote the space of smooth functions on the open configuration space with logarithmic singularities allowed on diagonals by
%$ C^\infty_{\log{}} (C^\circ_n(Y))$.

\begin{remark}\label{rem: DN kernel singularity}
Notice that\footnote{One way to prove Eq. \eqref{eq:reg second derivative integral} is as follows. In $S_0[\phi_{\til{\eta}} ]= \frac12 \int_\Sigma d\phi_{\til{\eta}} * d\phi_{\til{\eta}}  +  m^2\phi_{\til{\eta}} * \phi_{\til{\eta}}   = \frac12 \int_\Sigma d(\phi_{\til\eta}*d\phi_{\til\eta})$ use the expansion $\phi_{\eta}(x)=-
\int_{\partial\Sigma}\frac{\partial G_{\Sigma}(x,y')}{\partial\nu(y)}\eta(y')\,dy'$ for the second factor and use Stokes' theorem for the complement of a small half-disk of radius $\epsilon$ around $y'$. The second term in the second line in \eqref{eq:reg second derivative integral} arises (asymptotically) as the contribution of the boundary of that half-disk.} %$e^{-S(\phi_{\tilde{\eta}}^\Sigma)} \in \mathcal{H}_{\partial \Sigma}$, as we can express 
\begin{equation}
\begin{aligned}
S_0&(\phi_{\tilde{\eta}}^\Sigma) = -
\left[\frac12 \int_{\partial \Sigma \times \partial \Sigma} dy\,dy'\,\frac{\partial^2 G_{\Sigma}^{}(y,y')}{\partial\nu(y)\partial\nu(y')}{\tilde{\eta}}(y){\tilde{\eta}}(y')\right]_\mr{reg}\\
&:= - \lim_{\epsilon\ra 0} \left(\frac12 \int_{(y,y')\in\partial \Sigma \times \partial \Sigma,\, d(y,y')>\epsilon} dy\,dy'\,\frac{\partial^2 G_{\Sigma}^{}(y,y')}{\partial\nu(y)\partial\nu(y')}{\tilde{\eta}}(y){\tilde{\eta}}(y') -\frac{1}{\pi \epsilon} \int_{\dd \Sigma} dy\, \til\eta(y)^2\right)\label{eq:reg second derivative integral}
\end{aligned}
\end{equation} %deleted $D$.
Since the second normal derivative of the Green's function behaves as 
$\mathrm{O}\left(\frac{1}{d(y,y')^2}\right)$, it is worse than a logarithmic singularity allowed for a $2$-point collapse and thus 
$e^{-S_0(\phi_{\tilde{\eta}}^\Sigma)} \notin H^\mr{pre}_{\partial \Sigma}$.

In fact this singularity is strong enough to be non-integrable on the diagonal of the configuration space and the integral needs to be understood in the regularized sense, as in the second line above.
\end{remark}
  \begin{remark} %\marginpar{check..}
If $Y$ has several components $Y = Y_1 \sqcup \ldots \sqcup Y_n$, then the associated pre-space of states factorizes as a (projective) tensor product 
\begin{equation}
H^\mr{pre}_Y \cong H^\mr{pre}_{Y_1} \otimes \cdots \otimes H^\mr{pre}_{Y_n}.
\end{equation}
\end{remark}

\subsubsection{The gluing pairing}

In this subsection we define the pairing that will be used to formulate the gluing theorem. 
The notation is as follows. We consider a cobordism $(\Sigma,\partial_L\Sigma,\partial_R\Sigma)$. We then consider a decomposition of $\Sigma$ along a %\marginpar{\textcolor{blue}{Why call $Y$ a hypersurface and not a curve? - PM}}
 hypersurface (curve) $Y$: $\Sigma = \Sigma_L \cup_Y \Sigma_R$, such that $\partial\Sigma_L = \partial_L\Sigma \cup Y$ and $\partial\Sigma_R = \partial_R\Sigma \cup Y$. {  The heuristic idea to define the pairing is as follows: if $\Psi_1$ is a functional of boundary fields of the left cobordism ${\tilde{\eta}}_L,{\tilde{\eta}}_Y$, and $\Psi_2$ is a functional of the boundary fields ${\tilde{\eta}}_Y,{\tilde{\eta}}_R$ on the right cobordism, then we want to define
$$\langle \Psi_1,\Psi_2\rangle({\tilde{\eta}}_L,{\tilde{\eta}}_R) = \int_{{\tilde{\eta}}_Y}\Psi_1({\tilde{\eta}}_L,{\tilde{\eta}}_Y)\Psi_2({\tilde{\eta}}_Y,{\tilde{\eta}}_R)\, \mathcal{D}{\tilde{\eta}}_Y,$$ 
where $\mathcal{D}{\tilde{\eta}}_Y$ is the ``Lebesgue measure'' on $C^\infty(Y)$. To get to a mathematical definition, we notice that partition functions always include a factor of $e^{-S_0(\phi_{{\tilde{\eta}}_Y})}$. Thus, it makes sense to extract that factor and thus arrive at a formal \emph{Gaussian} measure on $C^\infty(Y)$, for which we can use the ideas of Section \ref{sec:formalints}: 
%\marginpar{\textcolor{blue}{changed $\eta_{L,R}$ to $\eta_Y$ in the 2nd integral and the sentence after. Another point: I also removed the prefactor ${\color{red}e^{-S_0(\phi^{\Sigma_L}_{{\tilde{\eta}}_L})}e^{-S_0(\phi^{\Sigma_R}_{{\tilde{\eta}}_R})}} $ here - I think it is confusing the issue: we are not yet talking about the adjusted partition functions here (haven't introduced them at this point). -PM}}
\begin{align*}
&\int_{{\tilde{\eta}}_Y}\Psi_1({\tilde{\eta}}_L,{\tilde{\eta}}_Y)\Psi_2({\tilde{\eta}}_Y,{\tilde{\eta}}_R) \,\mathcal{D}{\tilde{\eta}}_Y \\
&\text{``=''} %{\color{red}e^{-S_0(\phi^{\Sigma_L}_{{\tilde{\eta}}_L})}e^{-S_0(\phi^{\Sigma_R}_{{\tilde{\eta}}_R})}} 
\int_{{\tilde{\eta}}_Y}\widehat{\Psi_1}({\tilde{\eta}}_L,{\tilde{\eta}}_Y)\widehat{\Psi_2}({\tilde{\eta}}_Y,{\tilde{\eta}}_R)e^{-S_0(\phi_{{\tilde{\eta}}_Y}^{\Sigma_L}) - S_0(\phi_{{\tilde{\eta}}_Y}^{\Sigma_R})}\,\mathcal{D}{\tilde{\eta}}_Y
\end{align*}   }
With this idea in mind, we now define a map describing the formal integral over $C^\infty(Y)$ with respect to $e^{-S_0(\phi_{{\tilde{\eta}}_Y}^{\Sigma_L}) - S_0(\phi_{{\tilde{\eta}}_Y}^{\Sigma_R})}\mathcal{D}{\tilde{\eta}}_Y$. 

\begin{definition}
Let $\Sigma = \Sigma_L \cup_Y \Sigma_R$ and $D_{\Sigma_L,\Sigma_R}$ the Dirichlet-to-Neumann operator along $Y$ defined in Remark \ref{rem:DtoNgluing}. Let $K$ be the integral kernel of the inverse of $D_{\Sigma_L,\Sigma_R}$.
%\marginpar{\textcolor{blue}{ Notation $H^{(n)}$ did not appear before. }}
We define the map $\langle\cdot\rangle_{\Sigma_L,Y,\Sigma_R}\colon H^{(n)}_Y \to \C$, called the \emph{expectation value map}, by 
\begin{equation}
\left\langle \Psi\right\rangle_{\Sigma_L,Y,\Sigma_R}  = \dfrac{1}{\det(D_{\Sigma_L,\Sigma_R})^{\frac12}}\sum_{\mathfrak{m}\in \mathfrak{M}_n}\int_{C_n^\circ(Y)}\psi(y_1,\ldots,y_n)\prod_{\{v_1,v_2\}\in\mathfrak{m}}
K(y_{v_1},y_{v_2})\; dy_1\cdots dy_n.\label{eq:exp_value_bdry}
\end{equation}
The extension of this map to $H^\mr{pre}_Y$ is also denoted by $\langle\cdot\rangle_{\Sigma_L,Y,\Sigma_R}$. 
\end{definition}

Since $K(y,y')=\mathrm{O}(\log d(y,y'))$ at $y\ra y'$ and $\psi$ 
has admissible singularities
%is logarithmic
 on diagonals, the integral is convergent.
%is square integrable on $C^\circ_2(Y)$ (and hence the product of $K$'s appearing in (\ref{eq:exp_value_bdry}) is square integrable in $C^\circ_n(Y)$)  and $\psi$ are square integrable, this map is well-defined. }

%\textcolor{blue}{Moreover, by the same reason it is defined for $\psi\in L^1(C_n^0(Y))$.}
%\todo{Santosh:$K$ is not bounded because we are no longer in the product metric scenerio, even in that case, it is not bounded (I made an error in one of the  lemmas from in appendix in earlier draft)}

\begin{remark} This map is of course nothing but a formal integration over the field ${\tilde{\eta}} \in C^\infty(Y)$. It has an interesting interpretation as the expectation value of the observable $\Psi \colon C^\infty(Y) \to \C$ with respect to the theory on $Y$ with space of fields $C^{\infty}(Y)$ and (non-local) action functional $S_Y = \int_Y {\tilde{\eta}} D_{\Sigma_L,Y,\Sigma_R} {\tilde{\eta}}\, \dvol_Y$. This explains the notation.
\end{remark}
%\todo{K: remove this remark? Rephrase the definition?}

Since sets of odd cardinality do not have any perfect matchings, the map $\langle\cdot\rangle$ vanishes on $H^{(n)}$, for $n$ odd.

%To define the pairing, we remark that $H$ is closed under multiplication of functionals.    This follows from the following general fact. If $(X,\mu)$ is a finite measure space (i.e. $\mu(X) < \infty$) then $L^2(X,\mu) \subset  L^1(X,\mu)$ and hence the product of square integrable functions is again square integrable.  Thus, if $\Psi,\Psi'$ are represented by square integrable functions $\psi$, $\psi'$ respectively, then $\psi\psi'$ is also a square integrable function. 

For $\Psi\in H^{(n)}$, $\Psi'\in H^{(m)}$ two states with wave functions $\psi\in \funlog(C^\circ_n(Y))$, $\psi'\in \funlog(C^\circ_m(Y))$, we can form a new state $\Psi\odot\Psi'\in H^{(n+m)}$ whose wave function is the symmetrized tensor product $\psi\odot \psi' \in \funlog(C^\circ_{n+m}(Y))$ given by 
\begin{equation}(\psi\odot \psi')(y_1,\ldots,y_{n+m})=\frac{1}{(n+m)!}\sum_{\sigma} \psi(y_{\sigma(1)},\ldots,y_{\sigma(n)})\, \psi'(y_{\sigma(n+1)},\ldots,y_{\sigma(n+m)})\label{eq:def symm tens}
\end{equation}
with $\sigma$ running over permutations of $n+m$ elements.

The pairing is then simply the composition of the expectation value map with the multiplication $\odot$: 
%\todo{Here there is an analytic argument about the coefficients $\tau_\alpha$. They should be integrable and probably also square-integrable (i.e. in $L^1 \cap L^2$). Is this easy to see?}
\begin{definition}\label{def:pairing}
Let $\Sigma = \Sigma_L \cup_Y \Sigma_R$. Then, we define a pairing $$\langle\cdot,\cdot\rangle_{\Sigma_L,Y,\Sigma_R}  \colon H^\pre_Y \times H^\pre_Y \to \R[[\hbar^{1/2}]]$$ by 
%\marginpar{replace with $\otimes$ or $\odot$}
\begin{equation}\label{<>_D def}
\left\langle  \Psi, \Psi'\right\rangle_{\Sigma_L,Y,\Sigma_R}  = \left\langle\Psi  \odot  \Psi'\right\rangle_{\Sigma_L,Y,\Sigma_R} 
\end{equation}
and extending bilinearly.
\end{definition}
%\textcolor{blue}{
%Here we understand that, if $\Psi,\Psi'$ are represented by square integrable functions $\psi,\psi'$, then $\Psi\Psi'$ is represented by the product $\psi\psi'$, which only has $L^1$ regularity on the open configuration space, which is still enough to have a well-defined expectation value, by the remark above.}

\begin{definition}\label{def: H completion DN}
We define the space of states $H_Y$ associated to a Riemannian $1$-manifold $Y$ as the completion (order-by-order in $\hbar$) of $H^\pre_Y$ with respect to the pairing $\langle\cdot,\cdot \rangle_{\Sigma_L,Y,\Sigma_R} $.
\end{definition}

\begin{remark}\label{rem:ext_pairing}
If $\Sigma = \Sigma_L \cup_Y \Sigma_R$ with 
$\partial \Sigma_i = Y_i \sqcup Y, \hskip1.5mm i\in\{L,R\}$, then $\partial \Sigma = Y_L \sqcup Y_R$. By using the isomorphisms $H_{\partial\Sigma_{i}} \cong H_{Y_{i}} \otimes H_Y$ and $H_{\partial\Sigma} \cong H_{Y_L}\otimes H_{Y_R}$, the pairing extends to a map 
$$\langle\cdot,\cdot\rangle_{\Sigma_L,Y,\Sigma_R}  \colon H_{\partial\Sigma_L} \otimes H_{\partial\Sigma_R}\to H_{\partial\Sigma}.$$

%If $\alpha \in H_{\partial \Sigma_L}$ and $\beta \in H_{\partial\Sigma_R}$, then $\alpha \equiv \alpha(\eta_L, \eta)$, $\beta \equiv \beta(\eta,\eta_R)$, where $\eta_i \in C^\infty(Y_i)$ and $\eta \in C^\infty(Y)$. Then, by virtue of the decomposition \eqref{} we have
%\begin{multline}
%\left\langle e^{-S_0(\phi_\eta^{\Sigma_L})}\alpha, e^{-S_0(\phi_\eta^{\Sigma_R})}\beta\right\rangle_{\Sigma_L,Y,\Sigma_R}(\eta_L,\eta_R) \\=  e^{-S_0(\phi_{\eta_L}^{\Sigma_L})}  e^{-S_0(\phi_{\eta_R}^{\Sigma_R})}\left\langle  e^{-S_0(\phi`\eta^{\Sigma_L})} e^{-S^{\Sigma_L}_{L,R}(\phi_{\eta_L}^{\Sigma_L},\phi_{\eta}^{\Sigma_L})}\alpha(\eta_L,\eta),e^{-S_0(\phi_\eta^{\Sigma_L})} e^{-S^{\Sigma_R}_{L,R}(\phi_{\eta}^{\Sigma_R},\phi_{\eta_R}^{\Sigma_R})}\beta(\eta,\eta_R)\right\rangle_{\Sigma_L,Y,\Sigma_R}.
%\end{multline}
 \end{remark}

\subsection{Feynman graphs}
In this subsection we introduce the Feynman graphs relevant for this paper. 
\begin{definition}\label{def:FeynmanGraph}
A Feynman graph $\Gamma$ is given by the following data: 
\begin{enumerate}
\item Three disjoint finite sets $(V_b,V_{L},V_R)$, called the set of bulk and left resp. right boundary vertices. Their union, $V = V_b \sqcup V_L \sqcup V_R$ is called the set of vertices. $V_\partial = V_L \sqcup V_R$ is called the set of boundary vertices. 
\item A finite set $H$ with an incidence map $i \colon H \to V$
\item An involution $\tau \colon H \to H$ without fixed points (representing the edges)
\end{enumerate} 
such that 
%\item For all $v \in V_I$, we have $3 \leq |i_1^{-1}(v)| + |i_2^{-1}(v)| \leq N$ (bulk vertices are at least trivalent),
for all $v \in V_{L} \sqcup V_R$, we have $|i^{-1}(v)| = 1$ (boundary vertices are univalent).
%\item $i_1(h) \neq i_1(\tau_1(h))$ for all $h \in H_1$ (Type 1 edges have no short loops),
%\item for all $h \in H_2$, $|\{i_2(\tau(h),i_2(h)\} \cap V_I| = 1$ (type 2 edges connect bulk and boundary vertices). 

%\end{enumerate}
%\end{definition}
\end{definition}
The edge set $E(\Gamma)$ of the graph is by definition the set of orbits of $\tau$. We denote by $E_i(\Gamma)$ the edges that contain $i$ boundary vertices. Thus $E(\Gamma) = E_0(\Gamma) \sqcup E_1(\Gamma) \sqcup E_2(\Gamma)$. We give them different graphical representations (see Table \ref{tab:edges}). Some examples of graphs are shown in Figure \ref{fig:examplegraphs}. 

\begin{table}[!h]
\centering
\begin{tabular}{c|c|c}
edge & set & name \\
\hline 
\begin{tikzpicture} \node[] (a) at (-1,0){};
\node at (1,0) {}
edge (a); \end{tikzpicture} & $E_0$ & bulk edge \\
\begin{tikzpicture} \node[coordinate] (a) at (-1,0) {};
\node at (1,0){} 
edge[bubo] (a); 
\end{tikzpicture} &  $E_1$ & bulk-boundary edge \\
\begin{tikzpicture} \node[coordinate] (a) at (-1,0) {};
\node at (1,0) {}
edge[bobo] (a); 

\end{tikzpicture}& $E_2$ & boundary-boundary edge\\
\end{tabular}
\caption{Edges in Feynman diagrams}\label{tab:edges}
\end{table}

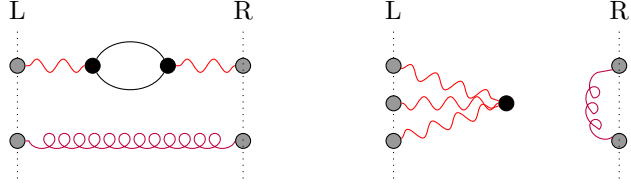
\begin{figure}[h!]
\centering
\begin{tikzpicture}
\draw[dotted] (0,0) -- (0,2);
\node[coordinate,label=above:{L}] at (0,2) {};
\node[bdry] (b1) at (0,1.5){};
\node[bulk] (bu1) at (1,1.5) {}
edge[bubo] (b1);
\node[bulk] (bu2) at (2,1.5) {}
edge[bubu, bend left=60] (bu1)
edge[bubu, bend right=60] (bu1) ;
\node[bdry] (b2) at (3,1.5){} 
edge[bubo] (bu2);
\node[bdry] (b3) at (0,0.5) {};
\node[bdry] at (3,0.5) {} edge[bobo] (b3);
\draw[dotted] (3,0) -- (3,2);
\node[coordinate,label=above:{R}] at (3,2) {};
\begin{scope}[shift={(5,0)}]
\draw[dotted] (0,0) -- (0,2);
\node[coordinate,label=above:{L}] at (0,2) {};
\node[bdry] (b1) at (0,1){};
\node[bulk] (bu1) at (1.5,1) {}
edge[bubo] (b1);
\node[bdry] (b2) at (0,1.5){}
edge[bubo] (bu1);
\node[bdry] at (0,0.5) {} edge[bubo] (bu1);
\node[bdry] (b3) at (3,1.5) {};
\node[bdry] at (3,0.5) {} edge[bobo, bend left=60] (b3);
\draw[dotted] (3,0) -- (3,2);
\node[coordinate,label=above:{R}] at (3,2) {};
\end{scope}
\end{tikzpicture}
\caption{Some examples of Feynman graphs}\label{fig:examplegraphs}
\end{figure}
We shall also require the notion of automorphism of a graph. 
\begin{definition}
An automorphism $\varphi$ of a graph $\Gamma$ is given by a pair of bijections: 
\begin{align*} 
V(\varphi)\colon V &\to V 
\end{align*} and 
\begin{align*}
H(\varphi) \colon H &\to H
\end{align*}
which commute with the incidence map $i$ and the involution $\tau$, i.e. 
\[\begin{tikzcd}
H \arrow{r}{i}  \arrow{d}{H(\varphi)} & V  \arrow{d}{V(\varphi)}\\
H \arrow{r}{i} & V  
\end{tikzcd}
\]
$V$ is required to respect the decomposition $V = V_b \sqcup V_L \sqcup V_R$. 
\end{definition}

The automorphism also induces maps on the sets of edges, denoted by $E_i(\varphi)$ or $E(\varphi)$.

Below we will rely on the following simple observation: 
\begin{prop}
Suppose all bulk vertices are at least trivalent. Then $\ell(\Gamma):=|E(\Gamma)|-|V_b(\Gamma)| - \frac12|V_{\partial}(\Gamma)| \geq 0$, with equality if and only if there are no bulk vertices.
\end{prop}
\begin{proof}
The assumption implies that the number of half-edges in the graph is at least $3|V_b(\Gamma)| + |V_{\partial}(\Gamma)|.$ This implies 
$$|E(\Gamma)| - \frac32 |V_b(\Gamma)| - \frac12|V_{\partial}(\Gamma)| \geq 0, $$ 
which in turn implies the statement. 
\end{proof}

%end of old version from 12/13

\subsection{Feynman rules and the perturbative path integral}
\label{sec:Feynmanrules}
Associated to a graph $\Gamma$ is a certain configuration space $C_{\Gamma}$:
\begin{definition} Given a Feynman graph $\Gamma$, we define the associated configuration space of $\Gamma$ in a cobordism $(\Sigma, \partial_L\Sigma,\partial_R\Sigma)$ as
\begin{equation*}
C^\circ_{\Gamma}(\Sigma) \equiv C^\circ_{\Gamma} := \{f\colon V \to \Sigma, f \text{ injective}, f(V_{L}) \subset \partial_L \Sigma,f(V_{R}) \subset \partial_R \Sigma\}
\end{equation*}
\end{definition}

%If $\Gamma$ has $k$ boundary and $l$ bulk vertices, picking an enumeration of $V_I$ and $V_{II}$ identifies $C^\circ_{\Gamma}$ as the open subset of $\Sigma^k \times \partial \Sigma^l$ given by removing all diagonals. Next, let us define a certain space of functionals on boundary fields: 
%\begin{definition} 
%\begin{enumerate}
%\item For $n \in \mathbb{N}, n \geq 1$, we define $H^{(n)}(\partial \Sigma)$ to be the space of functionals $\alpha$ on $C^{\infty}(\partial \Sigma)^n$ of the form 
%\begin{equation}\alpha(f_1,\ldots,f_n) = \int_{C^\circ_n(\partial \Sigma)}\tau_{\alpha}\pi^*_1f_1 \ldots \pi^*f_n\dvol,\label{eq:defalpha}\end{equation}
%and $H^{(0)} = \mathbb{C}$
%where $\tau_\alpha$ is a smooth, integrable function on $C^\circ_n(\partial \Sigma)$ and $\pi_i\colon C^\circ(\partial\Sigma)\to \partial \Sigma$ denotes the natural projection.
%\item Define $$H(\partial \Sigma) = \bigoplus_{n \geq 0} H^{(n)}[[\lambda]].$$
%\end{enumerate}
%\end{definition}

If $\Gamma$ has $k_l$ resp $k_r$ left resp. right boundary vertices and $l$ bulk vertices, picking an enumeration of $V_b, V_{L}, V_R$ identifies $C^\circ_{\Gamma}$ as the open subset of $\Sigma^l \times \partial_L\Sigma^{k_l} \times \partial_R\Sigma^{k_r}$ given by removing all diagonals. 
We now define the weight $F(\Gamma)$ as a functional of the boundary fields by associating a certain function (depending on the boundary fields) on $C^\circ_{\Gamma}$ to the graph and integrating it over $C^\circ_{\Gamma}$ against the measure $\dvol_{C^\circ_\Gamma(\Sigma)}$ induced by the embedding into $\Sigma^l \times \partial_L\Sigma^{k_l} \times \partial_R\Sigma^{k_r}$. Namely, $F(\Gamma)$ can be defined as follows:
\begin{definition}
Let $\Gamma$ be a Feynman graph with $n$ boundary vertices $v_1,\ldots, v_n$ {   and no short loops}. Let $\pi_i\colon C^\circ_\Gamma(\Sigma) \to Y$ denote the projection to the $i$-th boundary point.  Then $F(\Gamma)$ is the map $F(\Gamma)\colon C^\infty(\Sigma) \to \R$ defined by 
\begin{equation}F(\Gamma)[{\tilde{\eta}}] = \int_{C_\Gamma(\Sigma)}\omega_{\Gamma}\pi_1^*{\tilde{\eta}}\cdots \pi_n^*{\tilde{\eta}} \,\dvol_{C_{\Gamma}^\circ(\Sigma)}\label{eq:def_F_Gamma}
\end{equation}
where 
\begin{multline}\omega_{\Gamma} = \prod_{v\in V_b(\Gamma)}(-p_{val(v)})
\prod_{\{\alpha,\beta\} \in E_0} G_{\Sigma}(x_\alpha,x_\beta) \cdot \\ \cdot \prod_{\{x_\alpha,y_i\} \in E_1}
\Big( -\, \frac{\partial G_{\Sigma}(x_\alpha,y_i)}{\partial \nu(y_i)} \Big)
\prod_{\{y_i,y_j\} \in E_2}
\Big( - \,  \frac{\partial^2 G_{\Sigma}(y_i,y_j)}{\partial \nu(y_i)\partial\nu(y_j)} \Big)\label{eq:FeynmanRules}
\end{multline}   
%\marginpar{\textcolor{red}{Strange/inconsistent notations for endpoints in (\ref{eq:FeynmanRules}).}}
is the product of propagators and their normal derivatives according to the combinatorics of the graph.\footnote{
In the case if $\Gamma$ contains boundary-boundary edges connecting a boundary component to itself, the integral (\ref{eq:def_F_Gamma}) needs to be regularized as in Remark \ref{rem: DN kernel singularity}.
}
\end{definition} %\todo{Santosh:we need to discuss Pavel's comment}

\begin{remark}\label{rem:symmetry}
The configuration space $C^\circ_\Gamma$ has a natural map $p$ to $C^\partial_\Gamma(\Sigma):=C^\circ_{V_L}(\partial_L\Sigma) \times C^\circ_{V_R}(\partial_R  \Sigma)$ given by forgetting the bulk points. The fiber of this map over a pair $f_1,f_2$ of configurations is the open configuration space of $\Sigma \setminus (f_1(V_L)\cup f_2(V_R))$. Thus, we can define $\psi_\Gamma := p_*\omega_\Gamma$, where $p_*$ denotes integration (pushforward) along the fibers of $p$. $\psi_\Gamma$ is a function on $C^\partial_\Gamma(\Sigma)$ whose regularity we will study below. We can then rewrite \eqref{eq:def_F_Gamma} as 
\begin{equation}
F(\Gamma)[{\tilde{\eta}}] = \int_{C^\partial_\Gamma(\Sigma)}\psi_\Gamma \pi_1^*{\tilde{\eta}}\ldots \pi_n^*{\tilde{\eta}}\,\dvol_{(\partial\Sigma)^n}. \label{eq:def_F_2}
\end{equation}
We will call $\psi_\Gamma$ the \emph{wave function} associated to $\Gamma$. 
Even though $\psi_\Gamma$ does not need to be  symmetric under permutation of the boundary points, only its symmetric part will contribute to the integral \eqref{eq:def_F_2}. 
\end{remark}

We will now show that the coefficients of Feynman graphs have the nice regularity properties 
that we want. 
\begin{prop}\label{prop: Feynman singularities}
Let $\Gamma$ be a graph without short loops and without boundary-boundary edges connecting  $\dd_L$ to $\dd_L$ or $\dd_R$ to $\dd_R$. Then the corresponding wave function $\psi_\Gamma$ is a smooth function on the open configuration space with admissible singularities on diagonals, as in Definition \ref{def: adm singularities}:
$$ \psi_\Gamma\in \funlog(C^\dd_\Gamma(\Sigma)) $$
\end{prop}
\begin{proof}
Given a graph $\Gamma$, we are interested in the integral 
\begin{equation} \label{prop 3.16 integral}
\begin{aligned}
 f(y_1,\ldots,y_n) &=\int_{\Sigma^N}d^2x_1\cdots d^2 x_N \prod_{(\alpha,\beta)\in E_0} G^0(x_\alpha,x_\beta)\cdot \\
 &\cdot \prod_{(\alpha,i)\in E_1}  G^1(x_\alpha,y_i)\cdot  \prod_{(i,j)\in E_2}   G^2(y_i,y_j) 
 \end{aligned}
\end{equation}
as a function of $n$ pairwise distinct boundary points $y_1,\ldots,y_n \in Y$. Note that $\psi_\Gamma= f\cdot \,(-1)^{\# E_1 + \# E_2}
\prod_{v\in V_b}(-p_{val(v)})$ -- a constant multiple of $f$.  
Here we denoted $N=|V_b|$ the number of bulk vertices, $n=|V_L|+|V_R|$ the number of boundary vertices; we denoted $G^k$ the Green's function ($k=0$),  its first ($k=1$) or second ($k=2$) normal derivative at the boundary appearing in (\ref{eq:FeynmanRules}). 
Note that we have the following asymptotics:
\begin{equation}
\begin{gathered}\label{G estimates}
G^0 (x_\alpha,x_\beta) \underset{x_\alpha\ra x_\beta}{\sim}  -\frac{1}{2\pi}\log d(x_\alpha,x_\beta),\\ 
G^1(x_\alpha, y_i) \underset{x_\alpha\ra y_i}{=} \mathrm{O}\Big(\frac{1}{d(x_\alpha,y_i)}\Big), \quad G^2(y_i,y_j) \underset{y_i\ra y_j}{=} \mathrm{O}\Big(\frac{1}{d(y_i,y_j)^2}\Big).  
\end{gathered}
\end{equation} 
In fact, the arguments of $G^2$  never approach each other in (\ref{prop 3.16 integral}), since in $\Gamma$ we didn't allow boundary-boundary edges connecting a boundary component to itself.

First note that, for $y_1,\ldots,y_n$ fixed pairwise distinct boundary points, (\ref{prop 3.16 integral}) is a convergent integral: if $p\geq 2$ bulk points $x_{\alpha_1},\ldots,x_{\alpha_p}$ coalesce at pairwise distances  $\epsilon < d(x_{\alpha_r},x_{\alpha_s}) <C\epsilon$, the integrand behaves as $\OO(\log^a\epsilon)$ (with $a$ the number of edges in the collapsing subgraph), which gives an integrable singularity. If $p\geq 1$ bulk points collapse at a boundary point $y_i$, the integrand behaves as $$\OO\Big(\frac{\log^a\epsilon}{\epsilon^b}\Big)$$ where $b=1$ if in $\Gamma$ the vertex $y_i$ is connected to one of the collapsing bulk points and $b=0$ otherwise. This again gives an integrable singularity.\footnote{Note that here it is essential that the boundary vertices are univalent - otherwise we could have gotten $b>1$ which would lead to a non-integrable singularity.}

%SMOOTHNESS.....
%{\color{cyan} To prove smoothness, consider a single boundary point $y$, connected to a single bulk vertex $x$. Integrating out %Pushing forward over 
%all bulk points but $x$ yields, by the analysis above, a bounded function of $x$. Hence it is enough to show that the integral $F(y) = \int_\Sigma G^1(x,y)dx$ depends smoothly on $y$. To see this, split the integral into the integral over a small half-disk of radius $\delta$ around $y$ and its complement $\Sigma_\delta$. Then $\int_{\Sigma_\delta} G^1(x,y)dx$ is clearly smooth since the Green's function is smooth away from the diagonal. On the other hand, for $\delta$ small enough we have \marginpar{CORRECT THIS F-LA}
%$G^1(x,y) = -m/(2\pi)\left(\dfrac{\partial d(x,y)}{\partial \nu(y)}\right) K_1(md(x,y)) + \tilde{H}(x,y)$, where $H$ is smooth. The integral of the first term can be estimated in geodesic normal coordinates - where the metric satisfies $g_{ij} = \delta_{ij} + \mathrm{O}(\delta^2)$ -  by %\marginpar{$\sin^2\phi$ or just $\sin\phi$?}
%$C_1 \int_0^\pi d\varphi\, \sin\varphi\int_0^\infty dr\, r K_1(r)(1 + \mr{O}(\delta^2)) = C_2 (1 + O(\delta^2))$. Since the constant is independent of $y$, the integral of the first term depends smoothly on $y$ as well. 
Smoothness on the open configuration space of the boundary follows from standard arguments.

Next, we turn to the analysis of the singularities. 
Consider the asymptotic regime for (\ref{prop 3.16 integral}) when $y_i$ approaches $y_j$. In the limit $y_i\ra y_j$, the integral converges unless there is a bulk vertex $x_\alpha$ connected to both $y_i$ and $y_j$. If there is such a vertex $x_\alpha$, the integral is divergent as $y_i \to y_j$. Consider the integral over a half-disk $D^+_\delta$ of radius $\delta\gg d(y_i,y_j)$ centered at $y_j$, where $D^+_\delta$ is contained in a geodesic normal chart. The integral in $x_{\alpha}$ over the complement $\Sigma \setminus D^+_\delta$ does not create a singularity. The integral $\int_{D^+_\delta} \frac{d^2 x_\alpha}{d(x_\alpha,y_i) d(x_\alpha,y_j)}$ can be modeled by the corresponding integral for the flat case, since the leading singularities are the same in both cases and the metric in geodesic normal coordinates satisfies $g_{ij} = \delta_{ij} + \mathrm{O}(\delta^2)$. 
Thus, one obtains the estimate\footnote{
Indeed, denote that the integral in the l.h.s. of (\ref{prop 3.16 log model integral}) by $J$. It is a function of $\delta$ and the distance $d(y_i,y_j)$. Set $y_j=0$ for convenience. Making a rescaling $x_\alpha\mapsto\Lambda x_\alpha, y_i\mapsto \Lambda y_i $ under the integral, we see that $J(\Lambda\delta,\Lambda d(y_i,y_j))=J(\delta,d(y_i,y_j))$. Therefore, $J=J(\frac{\delta}{d(y_i,y_j)})$. Next, the integrand behaves in $J$ behaves as $\sim\frac{1}{r^2}$, with $r=d(x_\alpha,y_j)$, when $r\gg d(y_i,y_j)$. Therefore, in the asymptotic regime $\delta\gg d(y_i,y_j)$, we have $J\sim \int_0^\pi d\theta \int_{C d(y_i,y_j)}^\delta \frac{r dr}{r^2} \sim \pi \log \frac{\delta}{d(y_i,y_j)}$.
} 
\begin{equation}\label{prop 3.16 log model integral}
\int_{D^+_\delta} \frac{d^2 x_\alpha}{d(x_\alpha,y_i) d(x_\alpha,y_j)} \;\; \underset{y_i\ra y_j}{=} \;\; \OO\Big(\log\frac{\delta}{d(y_i,y_j)}\Big)
\end{equation}

%(for the model integral, one can consider a half-disk in the standard flat half-plane).
%REMOVE THIS FOOTNOTE?
%\footnote{ 
%More precisely, we cut out of $\Sigma$ a half-disk $\mathbb{D}_\delta^+$ of small radius $\delta\gg d(y_i,y_j)$ around $y_j$. We split the integral over $x_\alpha$ in (\ref{prop 3.16 integral}) into integrals over two regions: the half-disk $x_\alpha\in \mathbb{D}_\delta^+$  and the complement $x_\alpha\in \Sigma\backslash \mathbb{D}_\delta^+$. The latter integral does not generate a singularity as $y_i\ra y_j$, whereas the former is estimated by the model integral over a flat half-disk and generates a logarithmic singularity.
%}
% The integral over a region of the configuration space where $x_\alpha$ and $p$ other bulk points are in $\mathbb{D}_\delta^+$ and the other points are outside $\mathbb{D}_\delta^+$, is dominated by.... 
%Note that a collapse of more than one bulk vertex onto $y_i,y_j$ does not produce a singular behavior as $y_i\ra y_j$ (as one can see e.g. from scaling analysis of the respective model integral). 
Thus, the worst possible singularity of (\ref{prop 3.16 integral}) at a codimension one diagonal of the configuration space is the logarithmic one.

Consider a subset of boundary points $y_{i_1},\ldots, y_{i_k}$, with $k\geq 3$ coalescing at pairwise distances $\epsilon< d(y_{i_r},y_{i_s})<C\epsilon$. In this asymptotic regime, the strongest singularity in (\ref{prop 3.16 integral}) arises from the situation when a single bulk vertex $x_\alpha$ connected to each of the coalescing $y$'s by an edge, is colliding onto them. By a similar argument to the above, this situation is modeled by an integral over a (flat) half-plane
\begin{equation}\label{prop 3.16 power integral}
\int_{\Pi_+}\frac{d^2 x_\alpha}{d(x_\alpha,y_{i_1})\cdots d(x_\alpha,y_{i_k})}\;\;\underset{\epsilon\ra 0}{=}\;\; \OO\Big(\frac{1}{\epsilon^{k-2}}\Big)
\end{equation}
This estimate follows from a scaling argument: denoting the l.h.s. by $I(y_{i_1},\ldots, y_{i_k})$, we have $I(\Lambda y_{i_1},\ldots, \Lambda y_{i_k})=\Lambda^{2-k}I(y_{i_1},\ldots, y_{i_k})$ for any $\Lambda>0$, as follows from a scaling substitution in the integral $y_{i_r}\mapsto \Lambda y_{i_r}$, $x_\alpha\mapsto \Lambda x_\alpha$.

Finally, if we have a simultaneous collapse of several subsets of boundary points (at different points on the boundary), the respective worst-case-scenario asymptotics is given by a product of model integrals (\ref{prop 3.16 log model integral}), (\ref{prop 3.16 power integral}) corresponding to the collapsing subsets.
\end{proof}

It should be noted that naively extending the definition of the Feynman rules to diagrams with self-loops would yield  ill-defined results, as the Green's function is singular on the diagonal. One way to overcome this divergence problem is to not apply the formal integral to the exponential of the action, but to apply \emph{normal ordering} before applying the formal integral. Put simply, this has the effect of removing short loops\footnote{We refer to the literature, e.g. \cite{GJ} for an explanation of why this is the case.}. This leads to the following definition\footnote{This definition is just a neat way to rewrite the result of a formal computation of the path integral \eqref{eq:formalpathint} using the methods sketched in Section \ref{sec:formalints}, for a deeper discussion, we refer again to the literature, e.g. \cite{Polyak2005},\cite{Reshetikhin2010},\cite{Mnev2019}. }. 
 \begin{definition}
 We define the \emph{normal ordered perturbative partition function} \eqref{eq:formalpathint} by 
\begin{equation}\label{pert path int sigma} 
%\int_{\pi_{\Sigma}^{-1}(\proj_{Y}^{-1}(0))}D\hat{\phi} e^{-S_0(\hat{\phi})+\lambda S_{\text{int}}(\hat{\phi}+\phi_{{\tilde{\eta}}})} \\ :=
Z^\mathrm{no}_\Sigma({\tilde{\eta}}_L,{\tilde{\eta}}_R):=\frac{1}{\det(\Delta_{\Sigma}^{}+m^2)^{\frac12}}\sum_{\Gamma}\frac{\hbar^{\ell(\Gamma)}F(\Gamma)[{\tilde{\eta}}_L,{\tilde{\eta}}_R]}{|\Aut(\Gamma)|} 
%\inH_{\partial_L\Sigma} \otimes H_{\partial_R\Sigma}
\end{equation} %deleted $D$.
where the sum is over all Feynman graphs without self-loops, $\ell(\Gamma) = |E(\Gamma)| - |V_b(\Gamma)| -|\frac12 V_{\partial}(\Gamma)|\in \frac12 \mathbb{Z}_{\geq 0}$ and $F(\Gamma)$ is the Feynman weight of the Feynman graph $\Gamma$. 
\end{definition}
%{\color{orange}
%\begin{remark}\label{rem:finitelymany}
%Notice that by construction there are only finitely  \marginpar{Can remove?} many terms contributing to a fixed power of $\hbar$ and number of boundary vertices. 
%%Thus there are no problems of convergence in the individual terms of the sum on the right hand side. 
%
%\end{remark}
%}
{
\begin{remark}\label{rem:Zfactors}
Since boundary vertices are univalent, the contributions of the $E_2$ edges can be factored out. They yield precisely the exponential of $-S_0(\phi_{\tilde{\eta}})$. Hence, we can write
\begin{align}
Z^\mathrm{no}_\Sigma({\tilde{\eta}}_L,{\tilde{\eta}}_R)&=\frac{e^{-S_0(\phi_{{\tilde{\eta}}_L+{\tilde{\eta}}_R})}}{\det(\Delta_{\Sigma}^{}+m^2)^{\frac12}}\sum_{\{\Gamma\colon E_2(\Gamma)=\emptyset\}}\frac{\hbar^{\ell(\Gamma)}F(\Gamma)[{\tilde{\eta}}_L,{\tilde{\eta}}_R]}{|\Aut(\Gamma)|}. \notag \\ 
&= Z_{\Sigma}^\mr{free}(\til{\eta}_L,\til{\eta}_R)Z_\Sigma^{\mr{pert},\mr{no}}(\til{\eta}_L,\til{\eta}_R) \notag
\end{align} %deleted $D$.
Here $Z_{\Sigma}^\mr{free}$ is the partition function of the free theory \eqref{eq:Zfree} and $Z_\Sigma^{\mr{pert},\mr{no}}(\til{\eta}_L,\til{\eta}_R)$ is given by the sum of all diagrams containing no boundary edges. By construction, there are only finitely many diagrams at each order in $\hbar$ contributing to $Z_\Sigma^{\mr{pert},\mr{no}}$. Thus, $Z_\Sigma^{\mr{pert},\mr{no}} \in H^\mr{pre}_{\dd \Sigma}$. 
Expanding $S_0(\phi_{{\tilde{\eta}}_L + {\tilde{\eta}}_R}) = S_0(\phi_{{\tilde{\eta}}_L}) + S_0(\phi_{{\tilde{\eta}}_R}) + S_{L,R}({\tilde{\eta}}_L,{\tilde{\eta}}_R)$, we see that the first two terms generate Feynman diagrams connecting the left resp. right boundary to themselves, while the third term generates diagrams connecting the two, see Figure \ref{fig:DNbdrybdry} below. This observation will be important in the proof of the gluing formula. 
\end{remark}
}

{
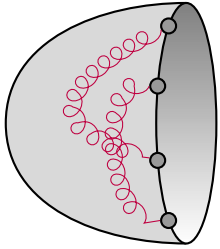
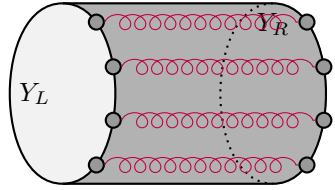
\begin{figure}[h]
\centering
\begin{subfigure}[b]{0.45\linewidth}
\begin{tikzpicture}[scale=.4]
   \coordinate (z0) at (1,4);
    \coordinate (z1) at (1,-4);
    \coordinate (z2) at (-5,0);
    \filldraw[fill=black!15, thick, draw=black] (z0) to[out=180, in=90] (z2) to[out=-90, in=180] (z1); % to[out=0, in=0]  (z0);
     %\draw[thick, draw=black] (z0) to[out=240, in=120] (z1) to[out=, in=0] (z0) ;
    \shadedraw[ thick, draw=black] (2,0) arc [start angle=0,   
                  end angle=360,
                  x radius=1cm, 
                  y radius=4cm]
    node [bdry, pos=.35] (x1) {} node [pos = .45, bdry] (x2) {}
    node [bdry, pos=.55] (x3) {} node [pos = .65, bdry] (x4) {};
   \draw[bobo] (x1) ..controls (-3,1.5) and (-4,0) .. (x3);
    \draw[bobo] (x2) to[out=180,in=180] (x4);
%     \node (k1) at (1,0) {BF};
  \end{tikzpicture}
 \caption{Graphs contributing to the exponential prefactor $e^{-\frac12\int_Y \eta D_{\Sigma}\eta \dvol_Y}$. Curled red edges are decorated with normal derivatives of Green's functions at both boundary points. }
 \end{subfigure}\hskip2mm \qquad
 \begin{subfigure}[b]{0.45\linewidth}
  \begin{tikzpicture}[scale=.4]
   \coordinate (a0) at (-1,-3);
  \coordinate (a1) at (-1,3);
  \coordinate (a2) at (6,3);
  \coordinate (x3) at (6,-3);
  \coordinate (x) at (0,0);  
    \filldraw[fill=black!30](a0) to[out=0, in=0] (a1)to[out=0, in=180](a2) node[below]{$Y_R$} to[in=0, out=0](x3)to[out=180, in=0]node[above]{} (a0) ;
  \filldraw[fill=black!5, draw=black!5](a0) to[out=180, in=180]node[right]{$Y_L$} (a1)to[in=0, out=0](a0);
  %%LEFT
    \draw[color=black, thick] (a0) to[out=0, in=0]
    node[bdry, pos=.2] (b1) {} node[bdry,pos=.4] (b2) {}
    node[bdry, pos=.6] (b3) {} node[bdry,pos=.8] (b4) {} (a1);
  \draw[color=black, thick] (a1) to[out=180, in=180] (a0);
  %%RIGHT
  \draw[color=black, thick] (x3) to[out=0, in=0] 
   node[bdry, pos=.2] (c1) {} node[bdry,pos=.4] (c2) {}
    node[bdry, pos=.6] (c3) {} node[bdry,pos=.8] (c4) {} (a2) ;
  \draw[color=black, thick, dotted] (a2) to[out=180, in=180] (x3);
%%STRAIGHT  LINES 
  \draw[color=black, thick](a1) to[out=0, in=180] (a2) ;
  \draw[color=black, thick](a0) to[out=0, in=180] (x3);
  %%EDGES 
  \draw[bobo] (b1) -- (c1);
    \draw[bobo] (b2) -- (c2);

  \draw[bobo] (b3) -- (c3);

  \draw[bobo] (b4) -- (c4);
 \end{tikzpicture}
  \caption{Graphs contributing to the exponential factor containing the ``off-diagonal $Y_L$-$Y_R$'' block of the Dirichlet-to-Neumann operator.} \label{fig: bdry-bdry} 
\end{subfigure}
\caption{Expansion of different ``blocks'' of Dirichlet-to-Neumann operator}\label{fig:DNbdrybdry}
 \end{figure}
}

We further define the adjusted partition function as
\begin{equation} \label{Z hat def}
\begin{aligned}
\widehat{Z^\mr{no}_\Sigma}(\til\eta_L,\til\eta_R) &=e^{S_0(\phi_{\til\eta_L})+S_0(\phi_{\til\eta_R})}Z^\mr{no}_\Sigma(\til\eta_L,\til\eta_R)  \\
&= \frac{e^{-S_{L,R}(\til\eta_L,\til\eta_R)}}{\det(\Delta_{\Sigma}^{}+m^2)^{\frac12}}\sum_{\{\Gamma\colon E_2(\Gamma)=\emptyset\}}\frac{\hbar^{\ell(\Gamma)}F(\Gamma)[{\tilde{\eta}}_L,{\tilde{\eta}}_R]}{|\Aut(\Gamma)|} \qquad  \in H_{\partial_L\Sigma} \otimes H_{\partial_R\Sigma}
\end{aligned}
\end{equation}
\begin{remark}\label{rem: Z is not in H but Zhat is in H}
Note that (\ref{pert path int sigma}) is not an element of the space of states (due to bad singularity in the exponential prefactors, cf. Remark \ref{rem: DN kernel singularity}), whereas (\ref{Z hat def}) is in the space of states.\footnote{Note however that (\ref{Z hat def}) in general is not in the pre-space of states (which was was defined via direct sum, not a direct product) due to infinitely many diagrams of type \ref{fig: bdry-bdry} contributing in the order $\OO(\hbar^0)$. In the completion, (\ref{Z hat def}) is a legitimate element.
}
\end{remark}

%{
%\begin{remark} 
%When pairing partition functions, we will ``absorb'' the part of $S_0(\phi_{{\tilde{\eta}}_L + {\tilde{\eta}}_R})$ which is supported on the ``glued'' boundary into the measure on the corresponding space of boundary fields. We will denote the partition function normalized in this way with a hat. This operation is only defined inside a pairing, and also depends on the position (left or right):  Thus, 
%\begin{equation}
%\left\langle \widehat{Z^\mathrm{no}_{\Sigma_L}},\widehat{Z^\mathrm{no}_{\Sigma_R}}\right\rangle_{\Sigma_L,Y,\Sigma_R} = \left\langle e^{S_0(\phi_{{\tilde{\eta}}_Y}^{\Sigma_L})}Z^\mathrm{no}_{\Sigma_L},e^{S_0(\phi_{{\tilde{\eta}}_Y}^{\Sigma_R})}Z^\mathrm{no}_{\Sigma_R}\right\rangle
%\end{equation}
%\end{remark}
%}
\subsubsection{Digression: Allowing coefficients $p_0,p_1,p_2$}\label{sec:low valence vertices}
The assumption that $p_0 = p_1 = p_2 = 0$ is important to have finitely many diagrams contributing in every order in $\hbar$ (excluding the boundary-boundary edges). In fact, 
we can reduce the general case to the case where these coefficients are absent  by resumming the corresponding diagrams.
Consider again the path integral  
\begin{equation*}
Z^{m,p}_{\Sigma}(\eta,\hbar)=\int_{\pi_{\Sigma}^{-1}(\proj_{Y}^{-1}(\eta))} e^{-\frac{S_0(\phi)+ S_{\text{int}}(\phi)}{\hbar}}\, D\phi
\end{equation*}
for non-rescaled boundary field, with $p_0,p_1,p_2$ possibly non-zero.   

\begin{itemize}
\item The linear term $p_1 \phi $ in $p(\phi)$ shifts
 the critical point of $\frac12 m^2\phi^2+p(\phi)$ away from zero to some value $\phi_\mr{cr}\in p_1\mathbb{R}[[p_1]]$. 
% -- the solution of the critical point equation  
% $m^2\phi+p'(\phi)=0$. The critical point equation can be written as
% $$\phi=\underbrace{-\frac{p_1}{m^2}-\sum_{n\geq 2}\frac{p_n}{(n-1)!}\phi^{n-1}}_{\Xi(\phi)}$$
%and its solution is the limit of iterations $\phi_\mr{cr}=\lim_{N\ra\infty}\Xi^N(0)$. 
Let 
$$p(\phi_\mr{cr}+\psi)=\til{p}(\psi)=\sum_{n\geq 0, n\neq 1}\frac{\til{p}_n}{n!}\psi^n$$
Note that $\til{p}(\psi)$ is a power series in the shifted field $\psi$ with vanishing linear term. Thus we obtain 
 $$ Z^{m,p}_{\Sigma}(\eta,\hbar) = Z^{m,\tilde{p}}(\eta - \phi_\mr{cr},\hbar).$$
 This corresponds to resumming the tree diagrams generated by the term $p_1\phi$ in the action. In fact, one can show\footnote{
This follows from writing the critical point equation $m^2\phi+p'(\phi)=0$ in the form $\phi=-\frac{p_1}{m^2}-\frac{1}{m^2}\sum_{n\geq 2}\frac{p_n}{(n-1)!}\phi^{n-1}=:\Xi(\phi)$. Its solution can be then written as the limit %$N\ra\infty$ 
of iterations %$\phi_{N}=\Xi(\phi_{N-1})$ with $\phi_0=0$.
 $\phi_\mr{cr}=\lim_{N\ra\infty}\Xi^N(0)$ which can in turn be presented as a sum over rooted trees. Here we note that for $\phi$ a constant function, one has 
 $\Xi(\phi)(x)=\int d^2 y\, G(x,y) (-p_1)+\int d^2 y\, G(x,y) \sum_{n\geq 2}\frac{-p_n}{(n-1)!}\phi(y)^{n-1}$ (which is also a constant function).
 } that 
 $$\phi_\mr{cr}(x) = \sum_{T \text{ rooted tree }}F(T)$$ 
 where the root of the tree is labeled by $x \in \Sigma$, see Figure \ref{fig:trees}. 
 \item The constant term $\til{p}_0$ in the new potential $\til{p}(\psi)$ can be carried out of the path integral. This corresponds to resumming the disconnected vertices generated by $\til{p}_0$.  Thus, we obtain 
 $$  Z^{m,\tilde{p}}(\eta - \phi_\mr{cr},\hbar) = e^{-\frac{\til{p}_0}{\hbar}\mr{Area}(\Sigma)}\; Z^{m,\til{p}_{\geq 2}}_\Sigma(\eta-\phi_\mr{cr},\hbar).$$  
 \item Finally, the quadratic term can be accounted for as a shift of mass, $\til{m}^2=m^2+p_2$. Indeed, resumming diagrams containing a binary vertex (see Figure \ref{fig:binary vertices}) we obtain a new propagator $G' = \sum_{k\geq 0}(-p_2)^k G^{k+1}$ which coincides with the Neumann series for $A + p_2$, where $A = \Delta + m^2$: 
 $$ (A + p_2I)^{-1} = A^{-1}\left(I - (-p_2 A^{-1})\right)^{-1}= A^{-1}\sum_{k\geq 0} (-p_2)^kA^{-k}. $$ 
\end{itemize}
\begin{figure}
\begin{tikzpicture}[scale=1, bulk/.style={shape=circle, draw, inner sep=0pt,fill=black,minimum size=4pt}]
\begin{scope}[shift={(-10,0)}]
\node[bulk,label=right:{$\til{p}_k$}] (c) at (0,0) {};
\node[coordinate, label=above:{$\psi$}] at (120:1) {} edge[bulk] (c);
\node[coordinate,label=left:{$\psi$}] at (210:1) {} edge[bulk] (c);
\node[coordinate,label=below left:{$\psi$}] at (240:1) {} edge[bulk] (c);
\draw[dotted, thick] (130:0.75) arc (130:200:0.75) node[midway, left] {$k$ legs};
\node[] at (2,0) {\large{$=\sum\limits_l$}};
\end{scope}
\begin{scope}[shift={(-5.5,0)}]
%\node[bulk,label=below:{\tiny{$p_{k+l}$}}] (c) at (0,0) {};
\node[bulk] (c) at (0,0) {};
\node[coordinate, label=above:{$\psi$}] at (120:1) {} edge[bulk] (c);
\node[coordinate,label=left:{$\psi$}] at (210:1) {} edge[bulk] (c);
\node[coordinate,label=below left:{$\psi$}] at (240:1) {} edge[bulk] (c);
\draw[dotted, thick] (130:0.75) arc (130:200:0.75);
\node[coordinate, label=above:{$\phi_\mr{cr}$}] at (70:1) {} edge[bulk] (c);
\node[coordinate,label=right:{$\phi_\mr{cr}$}] at (-20:1) {} edge[bulk] (c);
\node[coordinate,label=right:{$\phi_\mr{cr}$}] at (-50:1) {} edge[bulk] (c);

%\node[coordinate,label=below right:{$\phi_\mr{cr}$}] at (-70:0.75) {} edge[dashed] (c);
\draw[dotted, thick] (-10:0.5) arc (-10:60:0.5) node[midway,right] {$l$ legs};
\node[] at (3,0) {\large{$=\sum\limits_{\substack{\text{trees} \\ T_1,\ldots,T_l}}$}};
\end{scope}

\node[bulk] (c) at (0,0) {};
\node[coordinate, label=above:{$\psi$}] at (120:1) {} edge[bulk] (c);
\node[coordinate,label=left:{$\psi$}] at (210:1) {} edge[bulk] (c);
\node[coordinate,label=below left:{$\psi$}] at (240:1) {} edge[bulk] (c);
\draw[dotted, thick] (130:0.75) arc (130:200:0.75);

% \node[bulk] (B) at (-5,0){};
%   \node[label=left:{$x$}] at (B) {};
%   \node[bulk] (v1) at (-3.5,0) {};
\begin{scope}[scale=0.75,bulk/.style={shape=circle, draw, inner sep=0pt,fill=black,minimum size=3pt}]
\node[bulk] (v1) at (70:1) {} edge[bulk] (c);
\node[bulk] (v2) at (-20:1) {} edge[bulk] (c);
\node[bulk] (v3) at (-50:1) {} edge[bulk] (c);
\draw[dotted, thick] (-10:0.5) arc (-10:60:0.5);
%TREE 1
  \node[bulk](v11) at ($(v1)+(100:1)$){} edge[bulk] (v1);
  \node[bulk](v111) at ($(v11)+(100:0.7)$){} edge[bulk] (v11);
  \node[bulk](v112) at ($(v11)+(70:0.7)$){} edge[bulk] (v11);
  \node[bulk](v1111) at ($(v111)+(100:0.7)$){} edge[bulk] (v111);
  \node[bulk](v1112) at ($(v111)+(70:0.7)$){} edge[bulk] (v111);
  \node[bulk](v12) at ($(v1)+(60:1)$){} edge[bulk] (v1);
  \node[bulk](v121) at ($(v12)+(60:0.7)$){} edge[bulk] (v12);
  \node[bulk](v122) at ($(v12)+(30:0.7)$){} edge[bulk] (v12);
  \node[bulk](v123) at ($(v12)+(0:0.7)$){} edge[bulk] (v12);
%TREE2
  \node[bulk](v21) at ($(v2)+(40:1)$){} edge[bulk] (v2);
  \node[bulk](v211) at ($(v21)+(70:0.7)$){} edge[bulk] (v21);
  \node[bulk](v212) at ($(v21)+(40:0.7)$){} edge[bulk] (v21);
  \node[bulk](v2111) at ($(v211)+(10:0.7)$){} edge[bulk] (v211);
  \node[bulk](v2112) at ($(v211)+(70:0.7)$){} edge[bulk] (v211);
  \node[bulk](v22) at ($(v2)+(10:1)$){} edge[bulk] (v2);
     \node[bulk](v221) at ($(v22)+(40:0.7)$){} edge[bulk] (v22);
  \node[bulk](v222) at ($(v22)+(10:0.7)$){} edge[bulk] (v22);
    \node[bulk](v2221) at ($(v222)+(10:0.7)$){} edge[bulk] (v222);
\node[bulk](v2222) at ($(v222)+(-10:0.7)$){} edge[bulk] (v222);

%TREE3
  \node[bulk](v31) at ($(v3)+(-40:1)$){} edge[bulk] (v3);
  \node[bulk](v311) at ($(v31)+(-70:0.7)$){} edge[bulk] (v31);
  \node[bulk](v312) at ($(v31)+(-40:0.7)$){} edge[bulk] (v31);
\node[bulk](v32) at ($(v3)+(-10:1)$){} edge[bulk] (v3);
     \node[bulk](v321) at ($(v32)+(-40:0.7)$
     ){} edge[bulk] (v32);
  \node[bulk](v322) at ($(v32)+(-10:0.7)$){} edge[bulk] (v32);
  \end{scope}
%   \node[bulk](v5) at ($(v3)+(-80:1.5)$){};
%   \node[bulk] (v6) at ($(v4)+(-40:1.5)$){};
%
%   \node[bulk,label=above:{$-p_1\phi$}] (b1) at ($(v2)+(800:1)$) {};
%    \node[bulk,label=right:{$-p_1\phi$}] (b2) at ($(v2)+(0:1)$){};
%    \node[bulk,label=right:{$-p_1\phi$}] (b3) at ($(v4)+(40:1)$){};
%    \node[bulk,label=below:{$-p_1\phi$}] (b4) at ($(v5)+(-40:1)$){};
%    \node[bulk,label=below:{$-p_1\phi$}](b5) at ($(v5)+(-120:1)$){};
%    \node[bulk,label=right:{$-p_1\phi$}] (b6) at ($(v6)+(0:1)$){};
%   \node[bulk,label=right:{$-p_1\phi$}] (b7) at ($(v6)+(-80:1)$){};
%\draw[dotted,thick] (b2) arc (0:60:1); 
%     
%      
%   \draw[bulk] (B) to (v1);   
%    
%   \draw[bulk] (v1) to (v2);      
%   
%   \draw[bulk] (v1) to (v4);    
%   
%   \draw[bulk] (v1) to (v5);      
%        
%   \draw[bulk] (v4) to (v6);    
%  
%    \draw[bulk]  (v2) to (b1);         
%    \draw[color=black, thick]  (v2) to (b2);   
%    
%    \draw[color=black, thick]  (v4) to (b3);   
%    
%    \draw[color=black, thick]  (v5) to (b4);   
%
%    
%    \draw[color=black, thick]  (v5) to (b5);   
%    
%    \draw[color=black, thick]  (v6) to (b6);   
%    
%    \draw[color=black, thick]  (v6) to (b7);  
     
   \end{tikzpicture}
   \caption{Resumming trees into $\til{p}$.}\label{fig:trees}
\end{figure}
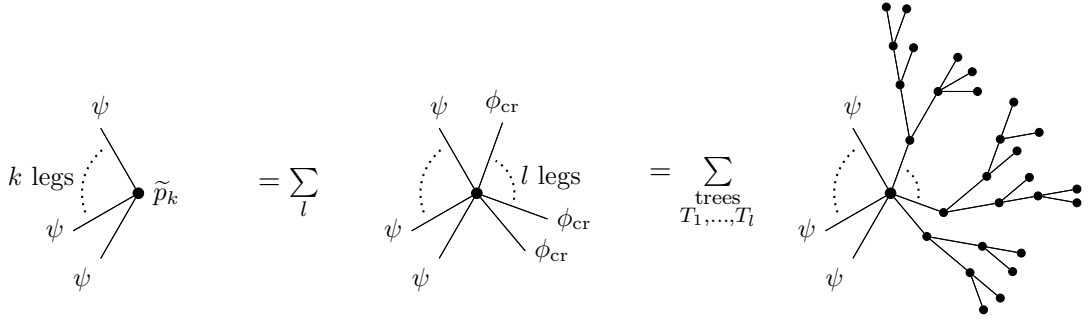
 \begin{figure}[h]
\centering
  \begin{tikzpicture}[scale=1]
   \node[bulk] (b11) at (0,0) {};
  \node[coordinate] at (120:1) {} edge[bulk] (b11);
\node[coordinate] at (210:1) {} edge[bulk] (b11);
\node[coordinate] at (240:1) {} edge[bulk] (b11);
\draw[dotted, thick] (130:0.75) arc (130:200:0.75);
  %\node[label=left:{$x$}] at (b11) {};
  \begin{scope}[shift={(2,0)}]
  \node[bulk] (b12) at (0,0) {} edge[double,thick] node[above] {$G'$} (b11);
   \node[coordinate]  at (70:1) {} edge[bulk] (b12);
\node[coordinate]  at (-20:1) {} edge[bulk] (b12);
\node[coordinate]  at (-50:1) {} edge[bulk] (b12);
\draw[dotted, thick] (-10:0.75) arc (-10:60:0.75);
  %\node[label=right:{$y$}] at (b12) {};
  \end{scope}
  \begin{scope}[shift={(5,0)}]
   \node at (-1.5,-0.25) {$=\sum\limits_{k\geq 0}$};
  \node[bulk] (b1) at (0,0) {};
  %\node[label=left:{$x$}] at (b1) {};
    \node[coordinate] at (120:1) {} edge[bulk] (b1);
\node[coordinate] at (210:1) {} edge[bulk] (b1);
\node[coordinate] at (240:1) {} edge[bulk] (b1);
\draw[dotted, thick] (130:0.75) arc (130:200:0.75);
    \node[bulk,label=below:{$\til{p_2}$}] (b2) at (1,0) {} edge[bubu] node[above] {$G$} (b1);
    \node[bulk,label=below:{$\til{p_2}$}] (b3) at (2,0) {} edge[bubu] node[above] {$G$} (b2);
    \node[coordinate] (b4) at (3,0) {} edge[bubu] node[above] {$G$} (b3);
    \node[label=below:{($k$ vertices})] at (4,0) {$\cdots$};
    \node[coordinate] (b5) at (5,0) {}; 
     \node[bulk,label=below:{$\til{p_2}$}] (b6) at (6,0) {} edge[bubu] node[above] {$G$}(b5);
     \node[bulk] (b7) at (7,0) {} edge[bubu] node[above] {$G$}(b6);
     %\node[label=right:{$y$}] at (b7) {};
 \node[coordinate]  at ($(b7)+(60:1)$) {} edge[bulk] (b7);
\node[coordinate]  at ($(b7)+(-30:1)$){} edge[bulk] (b7);
\node[coordinate]  at ($(b7)+(-60:1)$) {} edge[bulk] (b7);
\draw[dotted, thick] ($(b7)+(-20:0.75)$) arc (-10:50:0.75);
  \end{scope}
 \end{tikzpicture}
  \caption{Resumming binary vertices}\label{fig:binary vertices}
 \end{figure}
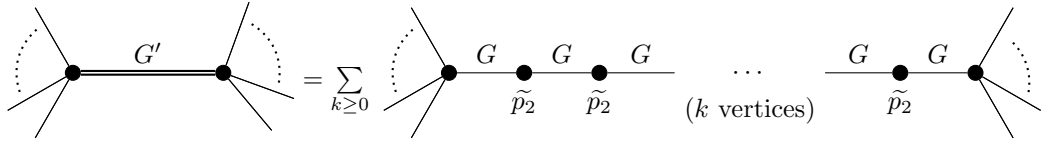 
 Thus, the path integral for the interaction potential $p$ reduces, by these manipulations (shift of the integration variable by a constant $\phi\ra \psi$, carrying a constant out and absorbing a quadratic term into the free part), to a path integral with interaction potential $\til{p}_{\geq 3}=\sum_{n\geq 3}\frac{\til{p}_n}{n!}\psi^n$ and mass $\til{m}$:
 $$ Z^{m,p}_\Sigma(\eta,\hbar)=  e^{-\frac{\til{p}_0}{\hbar}\mr{Area}(\Sigma)}\;\;Z^{\til{m},\til{p}_{\geq 3}}_\Sigma(\eta-\phi_\mr{cr},\hbar). $$
Mathematically, the r.h.s. here should be regarded as the definition of the left hand side. 
%On the level of Feynman diagrams, the shift of the field accounts for the resummation of trees with leaves decorated by $p_1$. The shift of mass accounts for the resummation of edges with several insertions of the $p_2$ vertex into a single edge  decorated with the propagator corresponding to a new mass.

\section{Heuristic analysis of path integrals and gluing formulae}\label{sec:gluing}
In this section, we discuss the heuristic analysis of path integrals associated to the free massive scalar field theory on compact oriented Riemannian manifolds and explain how they lead to the gluing formula for zeta regularized determinants and Green's functions. Furthermore, we give a rigorous proof of the gluing formula (see Proposition (\ref{Proposition:gluing_formula_green}) below) for the Green's functions. %Added a line which emphasizes there is a rigorous proof of the gluing formula for Green's function. 

In this section $\Sigma$ is a compact oriented Riemannian manifold, not necessarily of dimension two.
  
\subsection{BFK gluing formula for the zeta-regularized determinants}
First, we consider the partition function of the free massive scalar field theory and explain how it leads to the Burghelea-Friedlander-Kappeler (BFK) gluing formula for the zeta-regularized determinants \cite{BFK, Lee}. 

%Let $\Sigma$ be a compact oriented Riemannian manifold. 
We are interested in the path integrals of the form:\footnote{For the purpose of this section we can set $\hbar = 1$.} 
\begin{equation}\label{eq:partition_function_free_theory} Z_{\Sigma}=\int_{\mathcal{F}_{\Sigma}}e^{-S_0(\phi)}\,D\phi
\end{equation}     
When $\Sigma$ is closed, we can rewrite
 \[S_0(\phi)=\frac12\int_{\Sigma}\phi(\Delta_{\Sigma}+m^2)\phi\, \dvol(\Sigma)
 \]
This means that the integral in (\ref{eq:partition_function_free_theory}) is a Gaussian integral. 
%If $\mathcal{F}_{\Sigma}$ were a finite-dimensional vector space, then $Z_{\Sigma}$ would be simply $\det(\Delta_{\Sigma}+m^2)^{-\frac12}$ up to a nonzero scaling constant. Hence, we need a generalization of determinant for the infinite-dimensional case, and the zeta-regularized determinant \cite{LF1989} is one such generalization and following \cite{hawking}, we use the zeta-regularized determinant to define $\det(\Delta_{\Sigma}+m^2)$. 
Hence, we have \[Z_{\Sigma}={\det}(\Delta_{\Sigma}+m^2)^{-\frac12}.\] %deleted some overlap here.
More generally, if $\partial\Sigma=Y$ and $Y\ne\varnothing$, the partition function is defined by (\ref{eq:partition_function_free_theory_with_boundary}), which is:
\[Z_{\Sigma}({{\eta}})=(\det (\Delta_{\Sigma}^{}+m^2))^{-\frac12}e^{-S(\phi_{{\eta}})}.\] %deleted $D$ and added $tilde$.

Let $\Sigma$ be a closed oriented Riemannian manifold obtained by gluing two compact oriented Riemannian manifolds glued along a common boundary component $Y$: $\Sigma=\Sigma_L\cup_{Y}\Sigma_R$, $\partial\Sigma_L=Y$ and $\partial\Sigma_R=\overline{Y}$.   Recall that for $\eta\in \Cn(Y)$, we have for $i\in\{L,R\},$ 
\[Z_{\Sigma_{i}}(\eta)=(\det(\Delta_{\Sigma_i}^{}+m^2))^{-\frac12}e^{-S_{i}(\phi_{{{\eta}}})}\] where $S_i$ are defined as in (\ref{eq:boundary_action}). Let us assume that there is a formal Fubini's theorem (also known as locality of path integrals): %deleted $D$.
\begin{eqnarray}\label{formal_fubini_partition_function}
\begin{array}{lll}
\int_{\mathcal{F}_{\Sigma}} e^{-S(\phi)}\, D\phi 
 =\int_{\eta}\left(\int_{\pi^{-1}({{\eta}})} e^{-\left(S_L(\phi_L)+S_R(\phi_R)\right)}\,D\Phi\right)\,D{\eta}.
\end{array} %changed formatting.
\end{eqnarray} Then, this suggests the following gluing relation for the zeta-regularized determinants:

\begin{eqnarray}\label{eq:gluing_zeta_determinants}
\det(\Delta_{\Sigma}+m^2)=\det(\Delta_{\Sigma_L}^{}+m^2)\det(\Delta_{\Sigma_R}^{}+m^2)\det(D_{\Sigma_L,\Sigma_R})
\end{eqnarray}  %deleted $D$.
In summary, the locality of path integrals suggests a gluing formula for the zeta-regularized determinants. In fact, (\ref{eq:gluing_zeta_determinants}) is a theorem first proved by BFK in \cite{BFK} when $\Sigma$ is two-dimensional. It was later generalized to the 
%higher any even dimensional case 
case of arbitrary even dimension
by Lee \cite{Lee} under the assumption that the Riemannian metric is a product metric near the boundary. 

\begin{remark} 
%For the gluing relation (\ref{eq:gluing_zeta_determinants}) to hold, it is not necessary that $Y$ is  a dividing hypersurface, it can be any compact hypersurface. 
The gluing relation (\ref{eq:gluing_zeta_determinants}) admits a generalization to the case when $Y$ is not necessarily a dividing hypersurface: for $Y\subset \Sigma$ any compact hypersurface, one has $\det(\Delta_\Sigma+m^2)=\det(\Delta_{\Sigma\backslash Y}+m^2)\,\det(D_{\Sigma})$. Here $\Sigma\backslash Y$ is understood as $\Sigma$ with two additional boundary components, $Y$ and $\bar{Y}$; $D_\Sigma$ is the sum of Dirichlet-to-Neumann operators at $Y$ and at $\bar{Y}$.
%\marginpar{\textcolor{blue}{maybe some comment is needed on what is $D_\Sigma$ in this case?}}
\end{remark}

\subsection{Path integral representation of Green's function and a gluing relation}
In this subsection, we again consider the free massive scalar theory. Let $\Sigma$ be a compact oriented Riemannian manifold with $\partial\Sigma=Y$ and $f\in\Cn(\Sigma).$ Let us define an observable $O_f$, which is by definition a function on the space of fields, by
\[O_f(\phi)=\int_{\Sigma}f\phi\,\dvol_{\Sigma}\] 
The expectation of value of $O_{f}$ defines a function on the space of boundary fields: 
%\marginpar{\textcolor{blue}{Shouldn't we be integrating against $\hat\phi$ in the second line? -PM}}
\begin{eqnarray*}
\begin{array}{lll}
\left<O_{f}\right>(\eta)&=\int_{\pi^{-1}(\eta)} e^{-S(\phi)} O_{f}\,D\phi \\
                                     &=\int_{\pi^{-1}(\eta)} e^{-S(\hat\phi)-S(\phi_{\eta})} \left(\hat{O}_{f}(\hat\phi)+O_{f}(\phi_{\eta})\right)\,D\phi \\
                                     &={\det\left(\Delta_{\Sigma}^{}+m^2\right)^{-\frac12}}{e^{-S(\phi_{\eta})}O_{f}(\phi_{\eta})}\\
 \end{array} %deleted $D$.
\end{eqnarray*} where $\hat{O}_{f}$ is the observable on the space of fields which vanish on the boundary and it is defined by 
\[\hat{O}_f(\hat\phi)=\int_{\Sigma}f\hat{\phi}\,\dvol_{\Sigma}.\] 
Given $f,g\in\Cn(\Sigma)$ and a boundary field $\eta$, one can show that: 
\begin{eqnarray*}
\begin{array}{lll} \left<O_{f} O_{g}\right>(\eta)&
                          %&=\int_{\pi^{-1}(\eta)} D\phi e^{-S(\phi)} O_{f}O_{g}\\
                                                                       %&=\int_{\pi^{-1}(\eta)} D\phi e^{-S(\hat\phi)-S(\phi_{\eta})} \left(O_{f}(\phi_{\eta})O_{g}(\phi_{\eta})+\hat{O}_{f}\hat{O}_{g}\right)\\
                                                                       ={\det\left(\Delta_{\Sigma}+m^2\right)^{-\frac12}}{e^{-S(\phi_{\eta})}}\\&\left(O_{f}(\phi_{\eta})O_{g}(\phi_{\eta})
                                                 +\int_{\Sigma\times\Sigma}f(x)G_{\Sigma}(x,x')g(x')d^2x d^2x'\right)\,
\end{array}
\end{eqnarray*}
In particular,
\[
\left<O_{f} O_{g}\right>(0)={\det\left(\Delta_{\Sigma }+m^2\right)^{-\frac12}}\int_{\Sigma\times\Sigma}f(x)G_{\Sigma}(x,x')g(x')\,d^2x d^2x'
\]
 and taking $f=\delta_x$ and $g=\delta_{x'}$ we get the path integral represention of Green's function.

\subsubsection{Gluing relation for Green's functions}
The path integral representation of the Green's function and the formal Fubini type argument suggest that the Green's function with respect to the Dirichlet boundary condition satisfy a gluing relation and it can be proven rigorously (cf. Proposition \ref{Proposition:gluing_formula_green}). Let $\Sigma$ be a compact oriented Riemannian manifold obtained by gluing two compact oriented Riemannian manifolds glued along a common boundary component $Y$: $\Sigma=\Sigma_L\cup_{Y}\Sigma_R$, $\partial\Sigma_L=\overline{Y_L}\sqcup Y$ and $\partial\Sigma_R=\overline{Y}\sqcup Y_R$. Let $i\in\{L,R\}$.  Let $G_{\Sigma_i}$ be Green's functions on $\Sigma_i$ and $G_{\Sigma}^{}$ be the Green's function on $\Sigma.$ For $f,g\in\Cn(\Sigma)$ and $\eta_i\in \Cn(Y_i)$, by definition we have
\begin{eqnarray*}\left<O_{f} O_{g}\right>(\eta_L,\eta_R)
                          =\int_{\pi^{-1}(\eta_L,\eta_R)} e^{-S(\phi)} O_{f}O_{g}\,D\phi .
\end{eqnarray*}Let us assume that there is a formal Fubini's theorem:
\begin{eqnarray}\label{formal_fubini_green's_function}
\begin{array}{lll}
&\int_{\pi^{-1}(\eta_L,\eta_R)} e^{-S(\phi)} O_{f}O_{g}\,D\phi \\
 & =\int_{\eta}\left(\int_{\pi^{-1}(\eta,\eta_L,\eta_R)}e^{-\left(S(\phi_L)+S(\phi_R)\right)}O_{f}O_{g}\,D\Phi \right)\, D\eta,
\end{array}
\end{eqnarray} where $\Phi=(\phi_L,\phi_R)$, $\phi_i\in\Cn(\Sigma_i)$ and $\eta\in \Cn(Y)$. 
 
Next, we want to analyze (\ref{formal_fubini_green's_function}) in various situations. We first fix some notations. Given $\eta_i\in \Cn(Y_i),$  we use $\phi_{(\eta_L,\eta_R)}$ to denote the unique solution to Dirichlet boundary value problem on $\Sigma$ associated to $\Delta_{\Sigma}+m^2$ with boundary values $\eta_L$ and $\eta_R$. Similarly, given $\eta\in\Cn(Y)$, we will use $\phi^{(L)}_{(\eta_L,\eta)}$ and $\phi^{(R)}_{(\eta,\eta_R)}$ for the solutions to Dirichlet boundary value problems on $\Sigma_L$ and $\Sigma_R$ respectively. 
\\[1mm]
\underline{\textbf{Case (i):}} Assume both $f$ and $g$ are supported in $\Sigma_L.$ Then,
\begin{eqnarray*}&\left<O_{f} O_{g}\right>(\eta_L,\eta_R)\\
&={\det(\Delta_{\Sigma_L}^{}+m^2)^{-\frac12}\det(\Delta_{\Sigma_R}^{}+m^2)^{-\frac12}}\int_{\eta} e^{-\left(S(\phi_{\eta,\eta_R}^{(R)})+S(\phi_{\eta_L,\eta}^{(L)})\right)}\\
&\cdot\bigg\{\int_{\Sigma_L}f\phi_{\eta_L,\eta}^{(L)}\,\dvol_{\Sigma_L}\cdot\int_{\Sigma_L}g\phi_{\eta_L,\eta}^{(L)}\,\dvol_{\Sigma_L}\\
&+\int_{\Sigma_L\times\Sigma_L}f(x)g(x'))G_{\Sigma_L}^{}(x,x')\,d^2xd^2x'\bigg\}\,D\eta \\
&={\det(\Delta_{\Sigma_L}^{}+m^2)^{-\frac12}\det(\Delta_{\Sigma_R}^{}+m^2)^{-\frac12}\det D_{\Sigma_L,\Sigma_R}^{-\frac12}}
e^{-S(\phi_{(\eta_L,\eta_R)})}\\
&\bigg\{\int_{\Sigma_L\times\Sigma_L}f(x)g(x')G_{\Sigma_L}^{}(x,x')\,d^2x d^2 x'+\\
&\int_{\Sigma_L\times\Sigma_L}\left(\int_{Y\times Y}\dfrac{\partial G_{\Sigma_L}(x,y)}{\partial\nu(y)}f(x)g(x')\dfrac{\partial G_{\Sigma_L}(y',x')}{\partial\nu(y')}K(y,y')\,dy dy'\right)\,d^2x d^2x' \bigg\}\\
&=\det(\Delta_{\Sigma}^{}+m^2)^{-\frac12}e^{-S(\phi_{(\eta_L,\eta_R)})}
\bigg\{\int_{\Sigma_L\times\Sigma_L}f(x)g(x')G_{\Sigma_L}^{}(x,x')\,d^2x d^2x'+\\
&\int_{\Sigma_L\times\Sigma_L}\left(\int_{Y\times Y}\dfrac{\partial G_{\Sigma_L}(x,y)}{\partial\nu(y)}f(x)g(x')\dfrac{\partial G_{\Sigma_L}(y',x')}{\partial\nu(y')}K(y,y')\,dy dy' \right) \,d^2x d^2x'\bigg\},
\end{eqnarray*} %deleted $D$
where $K$ is the integral kernel of the inverse of the Dirichlet-to-Neumann operator $D_{\Sigma_L,\Sigma_R}$. We have used the gluing formula for the Dirichlet-to-Neumann operators \cite{SK2014} which amounts gluing of solutions of Helmholtz equations %added a about what gluing of Dirichlet-to-Neumann operators mean
and the BFK gluing formula for the zeta-regularized determinants above. 
%\iffalse
%\underline{\textbf{Case (ii):}} Suppose both $f$ and $g$ are supported in $\Sigma_2.$ Following as above,
% \begin{eqnarray*}&\left<O_{f} O_{g}\right>(\eta_1,\eta_2)\\
% &=\frac{e^{-S(\phi_{(\eta_1,\eta_2)}}}{\det_{\zeta}(\Delta_{\Sigma}+m^2)^{\frac12}}\bigg\{\int_{\Sigma_2\times\Sigma_2}f\left(\Delta_{\Sigma_2}+m^2\right)^{-1}g+\\&
%\int_{\Sigma_2\times\Sigma_2}\int_{Y\times Y}\dfrac{\partial G_{\Sigma_2}(x,\sigma_1)}{\partial\nu(\sigma_1)}f(x)\dfrac{\partial G_{\Sigma_2}(\sigma_2,y)}{\partial\nu(\sigma_2)}g(y)D_{\Sigma_2,\Sigma_1}^{-1}(\sigma_2,\sigma_1)\bigg\}
%\end{eqnarray*}
%where 
%\begin{eqnarray*}\left<B_{f}, D_{\Sigma_1,\Sigma_2}^{-1}B_g\right>=\int_{\Sigma_2\times\Sigma_2}\int_{Y\times Y}\dfrac{\partial G_{\Sigma_2}(x,\sigma_1)}{\partial\nu(\sigma_1)}f(x)\dfrac{\partial G_{\Sigma_2}(\sigma_2,y)}{\partial\nu(\sigma_2)}g(y)D_{\Sigma_2,\Sigma_1}^{-1}(\sigma_2,\sigma_1)
% \end{eqnarray*}
%\fi
\\[1mm]
\underline{\textbf{Case (ii):}} Suppose $f$ is supported in $\Sigma_L$ and $g$ is supported in $\Sigma_R.$ Then as above,
\begin{eqnarray*}&\left<O_{f} O_{g}\right>(\eta_L,\eta_R)\\
&=\int_{\eta}\left(\int_{\pi^{-1}(\eta,\eta_L,\eta_R)} e^{-S(\phi_L)-S(\phi_R)}\int_{\Sigma_L}f\phi_L\,\dvol_{\Sigma_L}\int_{\Sigma_R}g\phi_R\,\dvol_{\Sigma_R}\,D\Phi\right)\,D\eta\\
&=\int_{\Sigma_L\times\Sigma_R}\left(\int_{Y\times Y}\dfrac{\partial G_{\Sigma_L}(x,y)}{\partial\nu(y)}f(x)g(x')\dfrac{\partial G_{\Sigma_R}(y',x')}{\partial\nu(y')}K(y,y')\,dydy'\right)d^2 xd^2x'
\end{eqnarray*}

If we take $f=\delta_x$ and $g=\delta_{x'}$ in the relations above, they suggest a gluing relation for the Green's function. This gluing relation can be proven mathematically which is the content of the following proposition.
%This is the content of the following proposition. 

\begin{prop} \label{Proposition:gluing_formula_green} The Green's functions satisfy the following gluing relation:
    \begin{enumerate}[(i)]
    \item For $i\in\{L,R\}$ and  $x,x'\in \Sigma_i$:
    \begin{align*}
    \!\!\!G_{\Sigma}^{}(x,x')-G_{\Sigma_i}^{}(x,x')=\int_{Y\times Y}\frac{\partial G_{\Sigma_i}^{}(x,y)}{\partial \nu(y)} K(y,y')\frac{\partial G_{\Sigma_i}^{}(y',x')}{\partial \nu(y')}\,dydy'
  \end{align*}
  \item For $x\in\Sigma_L$ and $x'\in\Sigma_R:$ 
  \begin{align*} 
  G_{\Sigma}^{}(x,x')=\int_{Y\times Y}\frac{\partial G_{\Sigma_L}^{}(x,y)}{\partial \nu(y)} K(y,y')\frac{\partial G_{\Sigma_R}^{}(y',x')}{\partial \nu(y')}\,dydy'
  \end{align*}
  \end{enumerate} 
  \begin{proof} This proposition follows from the proof of Theorem 2.1  \cite{Carron} and the Green's identity. Let us consider the case when $x,y\in\Sigma_L$, the other cases follow similarly.  In Theorem 2.1 \cite{Carron}, it is shown that $K(y,y')=G_{\Sigma}(y,y')$ on $Y.$ By the Green's identity, we have
\begin{eqnarray}\label{Gluing_Carron}
G_{\Sigma}^{}(x,x')-G_{\Sigma_L}^{}(x,x')= -
\int_{Y}G_{\Sigma}(x,y')\frac{\partial G_{\Sigma_L}^{}(y',x')}{\partial \nu(y')}\,dy'
\end{eqnarray}
and 
\begin{eqnarray}\label{Green's_identity}
-  \int_{Y}\frac{\partial G_{\Sigma_L}^{}(x,y)}{\partial \nu(y)}G_{\Sigma}(y,y')\, dy =G_{\Sigma}(x,y').
\end{eqnarray}
Now the proposition, when $x,x'\in \Sigma_L$, follows from combining (\ref{Gluing_Carron}) and (\ref{Green's_identity}). The other cases follow similarly.
\end{proof}
\end{prop}

\begin{remark}\label{rem:greensfunctionfigure}
The gluing relation for the Green's function can be pictorially represented as in Figure \ref{fig:gluingrelation}. We represent the kernel of the inverse Dirichlet-to-Neumann operator by a zig-zag: \begin{tikzpicture} \node[bdry] (b1) at (0,0){}; \node[bdry] at(1,0){} edge[bobo2] (b1); \end{tikzpicture}

\begin{figure}[h!]
\centering
\begin{tikzpicture}[scale=.6]
\begin{scope}[shift={(-1,0)}]
\coordinate[label=left:$G_{\Sigma}$] (K) at (1.75, 2.2);
\coordinate[label=left:$\Sigma_L$] (G) at (3.5, 2.2);
\coordinate[label=left:$Y$] (M) at (4.6, 0);
\draw[ultra thick] (0,0) .. controls (1,2) and (2,2) .. (3,0);
\node[bulk] at (0,0){};
\node[bulk] at (3,0){};
\draw[densely dotted](4,0)--(4,3);
\end{scope}
\coordinate[label=left:${=}$] (i) at (5,1.5);
\begin{scope}[shift={(1,0)}]
\coordinate[label=left:$G_{\Sigma_L}$] (K) at (6.75, 2.2);
\draw[thick] (5,0) .. controls (6,2) and (7,2) .. (8,0);
\node[bulk] at (5,0){};
\node[bulk] at (8,0){};
\end{scope}
\coordinate[label=left:${+}$] (i) at (10,1.5);
\begin{scope}[shift={(2,0)}]
\coordinate[label=left:$\Sigma_L$] (G) at (10,2.4);
\coordinate[label=left:$Y$] (M) at (12.8, -.5);
\draw[densely dotted](12,0)--(12,3);
\node[bdry] at (12,.5){};
\node[bdry] at (12,2){};
\node[bulk] at (9,.5){};
\node[bulk] at (9,2){};
\draw[bubo]
(9, .5)--(12, .5);
\draw[bubo]
(9,2)--(12, 2);
\draw[bobo2]
(12, .5)--(12,2);
\end{scope}
\end{tikzpicture}
\vskip2mm
\begin{tikzpicture}[scale=.6]
\begin{scope}[shift={(-1,0)}]
\coordinate[label=left:$G_{\Sigma}$] (K) at (1.75, 2.2);
\coordinate[label=left:$\Sigma_L$] (G) at (0, 2.4);
\coordinate[label=left:$\Sigma_R$] (G) at (3.5, 2.4);
\coordinate[label=left:$Y$] (M) at (2, -0.4);
\draw[ultra thick] (0,0) .. controls (1,2) and (2,2) .. (3,0);
\node[bulk] at (0,0){};
\node[bulk] at (3,0){};

\draw[densely dotted](1.5,0)--(1.5,3);
\end{scope}
\coordinate[label=left:${=}$] (j) at (4,1.25);
\begin{scope}[shift={(1,0)}]
\draw[densely dotted](6.5,0)--(6.5,3);
\node[bdry] at (6.5,.5){};
\node[bdry] at (6.5,2){};
\node[bulk] at (8.5,.5){};
\node[bulk] at (4.5,2){};
\draw[bubo]
(4.5, 2)--(6.5, 2);
\draw[bubo]
(6.5,.5)--(8.5, .5);
\draw[bobo2]
(6.5, .5)--(6.5,2);
\coordinate[label=left:$\Sigma_L$] (G) at (5.5, 2.4);
\coordinate[label=left:$\Sigma_R$] (G) at (9.5, 2.4);
\coordinate[label=left:$Y$] (M) at (7, 0);
\end{scope}
\end{tikzpicture}
\caption{Gluing relation for the Green's function.
Thick lines mean one should associate the function corresponding to $\Sigma$.}
\label{fig:gluingrelation}
\end{figure}
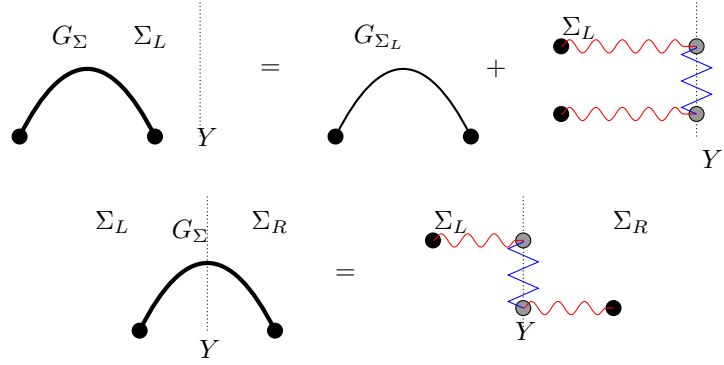
\end{remark}
\begin{remark}
The gluing formula implies similar formulae for the normal derivatives of the Green's function. These look schematically like the ones in Figure \ref{fig:gluingrelation2}.
\end{remark}
\begin{figure}[h!]
\centering
\begin{tikzpicture}[scale=.6]
\begin{scope}[shift={(-1,0)}]
\coordinate[label=left:$\partial_\nu G_{\Sigma}$] (K) at (1.75, 2.2);
\coordinate[label=left:$\Sigma$] (G) at (3.5, 2.2);
\coordinate[label=left:$Y$] (M) at (4.6, 0);
\node[bdry] (b1) at (0,1) {};
\node[bulk] at (2,2){} edge[bubo,ultra thick] (b1);
\draw[densely dotted](0,0)--(0,3);
\coordinate[label=below:$\partial_L\Sigma_L$] (a) at (0,0);
\draw[densely dotted](4,0)--(4,3);
\end{scope}
\coordinate[label=left:${=}$] (i) at (5,1.5);
\begin{scope}[shift={(1,0)}]
\coordinate[label=left:$\partial_\nu G_{\Sigma_L}$] (K) at (6.75, 2.2);
\node[bdry] at (5,1) {};
\node[bulk] at (7,2) {};
\draw[bubo] (5,1) -- (7,2);
\draw[densely dotted](5,0)--(5,3);
\draw[densely dotted] (9,0) --(9,3);
\coordinate[label=below:$\partial_L\Sigma_L$] (a) at (5,0);
\end{scope}
\coordinate[label=left:${+}$] (i) at (11,1.5);
\begin{scope}[shift={(4,0)}]
\coordinate[label=left:$\Sigma_L$] (G) at (10,2.4);
\coordinate[label=left:$Y$] (M) at (12.8, -.5);
\node[bdry] at (8,1) {};
\node[bulk] at (10,2) {};
\draw[densely dotted] (8,0) -- (8,3);
\coordinate[label=below:{$\partial_L\Sigma_L$}] (u) at (9,0);
\draw[densely dotted](12,0)--(12,3);
\node[bdry] at (12,.5) {} ;
\node[bdry] at (12,2) {} ;\draw[bubo]
(8, 1)--(12, .5);
\draw[bubo] (10,2)--(12, 2);
\draw[bobo2] (12, .5)--(12,2);
\end{scope}
\end{tikzpicture}
\vskip2mm
\begin{tikzpicture}[scale=.6]
\begin{scope}[shift={(-1,0)}]
\draw[densely dotted] (-1,0) -- (-1,3); 
\coordinate[label=below:$\partial_L\Sigma_L$] (i) at (-1,0);
\node[bdry] (b3) at (-1,.5) {};
\node[bdry] at (-1,2.5) {} edge [bobo, bend left, ultra thick] (b3);
\coordinate[label=right:$\Sigma_L$] (G) at (0, 2.4);
\coordinate[label=left:$\Sigma_R$] (G) at (3.5, 2.4);
\coordinate[label=left:$Y$] (M) at (2, -0.4);
\draw[densely dotted](1.5,0)--(1.5,3);
\end{scope}
\coordinate[label=left:${=}$] (j) at (4,1.25);
\begin{scope}[shift={(6,0)}]
\draw[densely dotted] (-1,0) -- (-1,3); 
\coordinate[label=below:$\partial_L\Sigma_L$] (i) at (-1,0);
\node[bdry] (b3) at (-1,.5) {};
\node[bdry] at (-1,2.5) {} edge [bobo, bend left] (b3);
\coordinate[label=right:$\Sigma_L$] (G) at (1, 3);
\coordinate[label=left:$\Sigma_R$] (G) at (3.5, 2.4);
\coordinate[label=left:$Y$] (M) at (2, -0.4);
\draw[densely dotted](1.5,0)--(1.5,3);
\end{scope}
\coordinate[label=left:${+}$] (j) at (10,1.25);

\begin{scope}[shift={(6,0)}]
\draw[densely dotted](6.5,0)--(6.5,3);
\node[bdry] at (6.5,.5){};
\node[bdry] at (6.5,2) {};
\node[bdry] at (4.5,.5){};
\node[bdry] at (4.5,2) {};
\draw[densely dotted] (4.5,0)--(4.5,3);
\draw[bobo] (4.5, 2)--(6.5, 2);
\draw[bobo]
(6.5,.5)--(4.5, .5);
\draw[bobo2]
(6.5, .5)--(6.5,2);
\coordinate[label=left:$\Sigma_L$] (G) at (5.5, 2.4);
\coordinate[label=left:$\Sigma_R$] (G) at (9.5, 2.4);
\coordinate[label=left:$Y$] (M) at (7, 0);
\end{scope}
\end{tikzpicture}
\caption{Gluing relation for normal derivatives of Green's functions. Thick lines mean one should associate the function corresponding to $\Sigma$.}
\label{fig:gluingrelation2}
\end{figure}
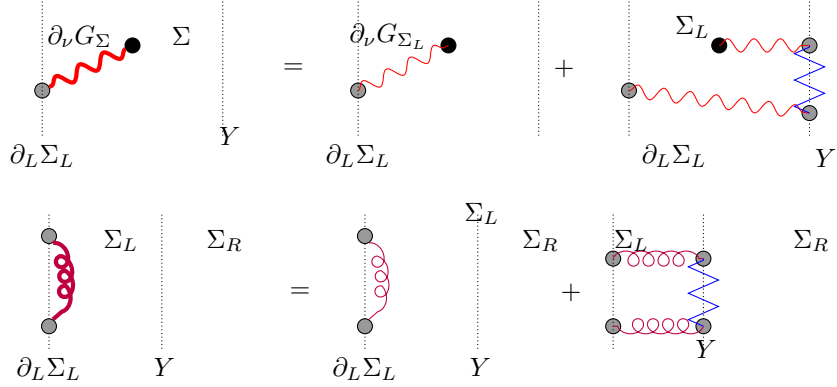
The goal of this paper is to show that the gluing formulae for determinants and Green's functions imply a gluing formula - a formal Fubini's theorem - for the perturbative partition functions. However, as it turns out, the gluing formula for the Green's function is not compatible with 
%\marginpar{\textcolor{blue}{Wick or normal ordering? (terminology) -PM}}
normal ordering, i.e. considering only Feynman diagrams without tadpoles (short loops) as in Section \ref{sec:pert_quant}. We discuss this issue in the next section.

\section{Regularization of tadpoles}\label{sec:tadpoles}
In principle, the formal application of Wick's theorem results in graphs with short loops. Under the usual Feynman rules, those would be assigned $G(x,x)$ - the (undefined) value of the Green's function on the diagonal. Normal ordering is tantamount to defining $G(x,x) = 0$, and with this assignment one obtains a well-defined perturbative partition function, as was shown in Section \ref{sec:pert_quant}. Below, we will explain why this definition is not quite satisfactory. We will then show how to overcome those problems by introducing more sophisticated regularizations $\tau(x)$ for $G(x,x)$. Finally, we discuss the relation between this approach and the ultra-violet cutoffs oftentimes used in quantum field theory.

\subsection{Why introduce tadpoles?} \label{sec: why tadpoles?}
First, the normal-ordered partition function does not satisfy the gluing formula: 
\begin{equation*}
%Z_\Sigma({\tilde{\eta}}_L,{\tilde{\eta}}_R) = \langle \widehat{Z_{\Sigma_L}}({\tilde{\eta}}_L,{\tilde{\eta}}),\widehat{Z_{\Sigma_R}}({\tilde{\eta}},{\tilde{\eta}}_R) \rangle_{\Sigma_L,Y,\Sigma_R}
Z_\Sigma \neq \langle \widehat{Z_{\Sigma_L}}({\tilde{\eta}}),\widehat{Z_{\Sigma_R}}({\tilde{\eta}}) \rangle_{\Sigma_L,Y,\Sigma_R}
\end{equation*}
%where the right hand side was defined in Remark \ref{rem:ext_pairing}. 
Indeed, from Equation \eqref{eq:exp_value_bdry} we see that the right hand side contains terms of the form in Figure \ref{fig:gluingviolation}
$$ \int_{Y\times Y}\frac{\partial G_{\Sigma_i}^{}(x,y)}{\partial \nu(y')}K(y,y')\frac{\partial G_{\Sigma_i}^{}(y',x)}{\partial \nu(y')}\,dydy',$$
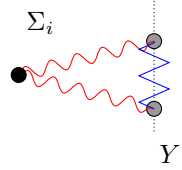
\begin{figure}[h!]
\begin{center}
\begin{tikzpicture}[scale=.6]
\coordinate[label=left:$\Sigma_i$] (G) at (10,2.4);
\coordinate[label=left:$Y$] (M) at (12.8, -.5);
\draw[densely dotted](12,0)--(12,3);
\node[bdry] at (12,.5){};
\node[bdry] at (12,2) {};
\draw[bubo]
(9, 1.25)--(12, .5);
\draw[bubo]
(9,1.25)--(12, 2);
\node[bulk] at (9,1.25){};
\draw[bobo2]
(12, .5)--(12,2);
\end{tikzpicture}
\end{center}
\caption{Diagrams which violate normal ordering when gluing}\label{fig:gluingviolation}
\end{figure}
which do not appear from the  gluing formula in Remark \ref{rem:greensfunctionfigure} for the Green's function if there are no tadpole diagrams. One can see in examples that they do not vanish.  \\
Second, defining $G(x,x) = 0$ is inconsistent with zeta-regularization of the determinant already at the level of the free theory, in the following sense. Namely, we can consider  a quadratic perturbation 
$$S = \frac{1}{2}\int_\Sigma d\phi \wedge *d\phi + m^2 \phi \wedge * \phi + \alpha  \phi \wedge *\phi.$$ 
If we include $\alpha$ in the free action, the corresponding partition function is 
$$Z = {\det (\Delta_\Sigma + m^2 + \alpha)^{-\frac12}}.$$
On the other hand, treating $\alpha$ as a perturbation, by the convention that $G(x,x)=0$ we obtain 
$$Z = {\det (\Delta_\Sigma + m^2)^{-\frac12}}.$$
We are thus led to look for another assignment $\tau(x) = G(x,x)$ which will resolve these issues. 
This motivates the definitions in the subsection below.

\subsection{Tadpole functions}
%\begin{definition}
%A function $T \in L^2(\Sigma))$ is called a \emph{tadpole function}.
%\end{definition}
Given $\tau\in L^p(\Sigma)$ for arbitrarily large $p$, we can define the corresponding Feynman rules $F^{\tau}(\Gamma)$ where we evaluate short loops using $\tau$. Since short loops are often called \emph{tadpoles}, we will refer to $\tau$ as a \emph{tadpole function.}
\begin{lemma}
Let $\Gamma$ be a Feynman diagram, possibly with short loops. Then $F^{\tau}(\Gamma) \in 
%L^2((\partial \Sigma)^n) \cap C^\infty(C_n^\circ(\partial \Sigma)))
\funlog(C^\circ_n(\dd \Sigma))
$. 
%\todo{S: I was not so sure, we want to leave it like this or change to $L^2((C_n^\circ(\partial \Sigma)))$}
\end{lemma}
{  
\begin{proof}%\marginpar{REWRITE!}
By inspection of the proof of Proposition \ref{prop: Feynman singularities}, we see that it adapts to this case.
%Again, square integrability follows from the fact that products of square integrable functions are square integrable. Smoothness follows because the Green's functions are smooth in the boundary arguments.
\end{proof}
%Having established the technicalities, 
We can now define the partition function with tadpole $\tau$. 
\begin{definition} \label{def: Z tadpole}
We define the partition function with respect to $\tau$: 
\begin{equation}\label{Z^tau}
Z^{\tau}_\Sigma({\tilde{\eta}}_L,{\tilde{\eta}}_R) = \frac{1}{\det(\Delta_\Sigma+m^2)^\frac12}\sum_{\Gamma}\frac{F^{\tau}(\Gamma)}{|\mathrm{Aut}(\Gamma)|}\hbar^{\ell(\Gamma)}
\end{equation}
\end{definition} 
}

Suppose we have two manifolds $\Sigma_L$, $\Sigma_R$ with a common boundary component $Y$ and tadpole functions $\tau_i$ on $\Sigma_i$ for $i\in\{L,R\}$.  We can then define a function $\tau_L * \tau_R$ on $\Sigma_L \cup_Y \Sigma_R$ by setting for $x \in \Sigma_i$
\begin{equation}\label{tadpole gluing formula}
(\tau_L * \tau_R)(x) = \tau_i(x) + \int_{(y,y') \in Y\times Y}\frac{\partial G_{\Sigma_i}^{}(x,y)}{\partial \nu(y)} K(y,y')\frac{\partial G_{\Sigma_i}^{}(y',x)}{\partial \nu(y')}dydy'.
\end{equation}
\begin{lemma} The following holds:
\[\tau_L * \tau_R\in L^p(\Sigma_L \cup_Y \Sigma_R)\]
for any $p>0$.
\end{lemma}
\begin{proof} %\marginpar{REWRITE}
 %  Square integrability follows immediately from the same arguments used in the proofs so far.  
 The second term on the r.h.s. in (\ref{tadpole gluing formula}) is smooth away from $Y$ and behaves as $\OO(\log d(x,Y))$ near $Y$ (as can be shown by considering a model integral on a half-space), which implies the statement, cf. Lemma \ref{lemma:log of distance to the boundary} below.
   \end{proof}
\begin{definition}
We call an assignment of tadpole functions $\tau_\Sigma$ to surfaces $\Sigma$  %$\tau\colon\Sigma \mapsto \tau_\Sigma \in L^2(\Sigma)$ 
a \emph{local} assignment if it satisfies the gluing formula  
%\todo{S:May be we just call this a local assignment?}
\begin{equation}
\tau_{\Sigma_L \cup_Y \Sigma_R} = \tau_L * \tau_R
\end{equation}
\end{definition}
Pictorially, this gluing formula can be represented as in Figure \ref{fig:gluingtadpole}.
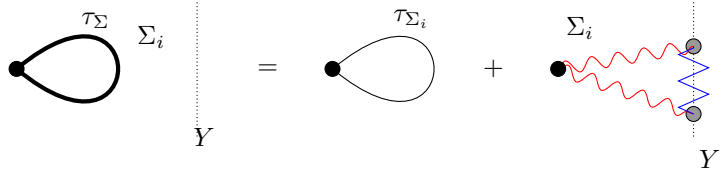
\begin{figure}[h!]
\begin{tikzpicture}[scale=.6]
\begin{scope}[shift={(-1,0)}]
\coordinate[label=above:$\tau_{\Sigma}$] (K) at (1.75, 2.2);
\coordinate[label=left:$\Sigma_i$] (G) at (3.5, 2.2);
\coordinate[label=left:$Y$] (M) at (4.6, 0);
\draw[ultra thick] (0,1.5) .. controls  (3,4) and (3,-1) .. (0,1.5);
\draw[densely dotted](4,0)--(4,3);
\node[bulk] at (0,1.5) {};
\end{scope}
\coordinate[label=left:${=}$] (i) at (5,1.5);
\begin{scope}[shift={(1,0)}]
\coordinate[label=above:$\tau_{\Sigma_i}$] (K) at (6.75, 2.2);
\draw (5,1.5) .. controls  (8,4) and (8,-1) .. (5,1.5);
\node[bulk] at (5,1.5) {};
\end{scope}
\coordinate[label=left:${+}$] (i) at (10,1.5);
\begin{scope}[shift={(2,0)}]
\coordinate[label=left:$\Sigma_i$] (G) at (10,2.4);
\coordinate[label=left:$Y$] (M) at (12.8, -.5);
\draw[densely dotted](12,0)--(12,3);
\node[bdry] at (12,.5){};
\node[bdry] at (12,2){};
\draw[bubo]
(9, 1.5)--(12, .5);
\draw[bubo]
(9,1.5)--(12, 2);
\draw[bobo2]
(12, .5)--(12,2);
\node[bulk] at (9,1.5) {};
\end{scope}
\end{tikzpicture}
\caption{Gluing relation for local tadpole functions}
\label{fig:gluingtadpole}
\end{figure}

This definition assures compatibility with the gluing formula. Another property we can ask for is the consistency with zeta-regularization:
\begin{definition}[Compatibility with zeta-regularization]\label{def:compatibility}
Let $\tau_\Sigma$ be a tadpole function.
\begin{enumerate}[i)]
\item We say that $\tau_\Sigma$ is \emph{weakly compatible} with zeta-regularization if 
\begin{equation*}
\int_\Sigma \tau_{\Sigma}\,\dvol_{\Sigma} = \frac{d}{dm^2}\log\det(\Delta_{\Sigma} + m^2). 
\end{equation*} %Negative sign removed
\item  Let $F\colon C_c^\infty(\Sigma \setminus \partial \Sigma) \supset U \to \R$ be given by $$F(\alpha) = \log\det(\Delta_{\Sigma}+m^2+\alpha).$$ Then $F$ is Fr\'echet differentiable in a neighbourhood of $0\in U$. We say that $\tau_\Sigma$ is \emph{strongly compatible} with zeta-regularization if $$DF(0)\alpha = \int_\Sigma \tau_\Sigma(x)\alpha(x)d^2x,$$ %Negative sign removed
i.e. $\tau_\Sigma$ is the distribution representing $DF(0)$.  
\end{enumerate}
\end{definition}
In the next two subsections, we will show that there exists a local assignment  which is consistent with the zeta-regularization. 

\subsection{Zeta-regularized tadpole}\label{tadpole:zeta-regularized}

Let $\Sigma$ be a closed and oriented two dimensional Riemannian manifold. Let $\theta_\A(x,x',t)$ be the integral kernel of $e^{-t\A}$ i.e. the heat kernel. Let 
$\theta_\A(x,t)$ denote $\theta_\A(x,x,t)$. Then, we define the \emph{local zeta function} associated to $\A$ as follows:

\begin{definition} The local zeta function associated to $\A$ is denoted by $\zeta_{\A}(s,x)$ and defined as
\begin{equation}\label{eq:zeta_function}
\zeta_{\A}(s,x):=\dfrac{1}{\Gamma(s)}\int_{0}^{\infty}t^{s-1}\theta_\A(x,t)\,\mathrm{d}t.
\end{equation}
\end{definition}
The relation between the zeta function of $A$ and local zeta function of $A$ is given by 
\[\zeta_\A(s)=\int_{\Sigma}\zeta_{\A}(s,x)\,d^2x
.\]

We can use the small time asymptotics of the heat kernel to investigate the local zeta function. We first  recall that \cite{Dewitt, Gilkey, Mckean}
\[\theta_\A(x,x',t)=e^{-m^2t}\frac{e^{-d(x,x')^2/4t}}{4\pi t}\left(a_0(x,x')+a_1(x,x')t + \mathrm{O}(t^2)\right)
\] for $t\to0.$
In particular, when $t\to0$,  
\begin{equation}\label{eq:local heat trace expansion}
\theta_\A(x,t)=\dfrac{e^{-m^2t}}{4\pi t}\left(1+a_1(x)t + \mathrm{O}(t^2)\right)
\end{equation}
where $a_1(x)=a_1(x,x)$. It is known that \cite{Mckean} that $$a_1(x)=\dfrac{1}{6}\mathrm{R}(x)$$ where $\mathrm{R}$ is the scalar curvature of $\Sigma$. We can use these properties of the heat kernel together with its large time behavior to show that $\zeta_\A(s,x)$ is holomorphic for $\mathrm{Re}(s)>1$, it has meromorphic extension to $\mathbb{C}$ and $\zeta_\A(s,x)$ is holomorphic at $s=0$. This is the content of the following lemma.

\begin{lemma} \label{lemma:local_zeta_meromorphic_extension}
For each $x\in \Sigma$, $\zeta_\A(s,x)$ is holomorphic for $\mathrm{Re}(s)>1$ and $\zeta_\A(s,x)$ has meromorphic extension to $\mathbb{C}$ and it is holomorphic at $s=0$. Moreover, 
\begin{align*}
&\dfrac{\d}{\d s}\Big|_{s=0}\zeta_\A(s,x)
=\dfrac{1}{4\pi}m^2(\log m^2-1)  -\dfrac{1}{4\pi} a_1(x) \,\log{m^2} +
\int_{0}^{\infty}\dfrac{e^{-m^2t}\left(\theta_{\Delta_{\Sigma}}(x,t)-g(x,t)\right)}{t}\,\dt,
\end{align*}
where $g(x,t)=\dfrac{1}{4\pi t}
+a_1(x)$. Furthermore, the assignment $x\mapsto \dfrac{\d}{\d s}\Big|_{s=0}\zeta_A(s,x)$ is smooth.
\begin{proof} 
For large $\mathrm{Re}(s)$, a simple computation shows
\begin{eqnarray}\label{eq:local_zeta_function_expression}
\begin{array}{lll}
\zeta_\A(s,x)&=\dfrac{1}{4\pi}\dfrac{1}{(s-1)m^{2s-2}}  + \dfrac{a_1(x)}{4\pi m^{2s}} 
\\
&\displaystyle +\dfrac{1}{\Gamma (s)}\int_{0}^{\infty}t^{s-1}e^{-tm^2}\left(\theta_{\Delta_{\Sigma}}(x,t)-g(x,t)\right)\,\dt
\end{array}
\end{eqnarray}
This representation of $\zeta_A(s,x)$ proves the first part of the lemma. Now, the expression for the derivative at $s=0$ follows from using the fact that $\Gamma$ has a pole of order one at $s=0$ while differentiating $\zeta_A(s,x)$. Finally, smoothness of coefficients of the local heat kernel expansion in the interior of $\Sigma$ implies the assignment $x\mapsto\dfrac{\d}{\d s}\Big|_{s=0}\zeta_\A(s,x)$ is smooth.
\end{proof}
\end{lemma}

From the proof of Lemma \ref{lemma:local_zeta_meromorphic_extension}, we see that $\zeta_\A(s,x)$ has a simple pole at $s=1$. However, we can consider the finite part of the local zeta function $s=1$. This motivates the following definition. 

\begin{definition}\label{def: tau zeta-reg} We define the tadpole function via zeta regularization by 
\begin{equation}\label{tau zeta-reg definition}
\tau^{\mathrm{reg}}_{\Sigma}(x)=\fp_{s=1}\zeta_\A(s,x):=\dfrac{\d}{\d s}\Big|_{s=1}(s-1)\zeta_\A(s,x).
\end{equation} 
Note that  $\tau^{\mathrm{reg}}_{\Sigma}(x)$ is the constant term in the Laurent series expansion of $\zeta_\A(s,x)$ at $s=1$.
\end{definition}
One can think of $\tau^\mathrm{reg}_\Sigma(x)$ as a regularization\footnote{
In fact, it would be more appropriate to call it \emph{renormalization}: we first regularize by shifting $s$ away from $1$ and then we subtract the singular part $\frac{\mathrm{const}}{s-1}$.
} of the value of the Green's function $\zeta_A(1,x)=G_\Sigma(x,x)$ on the diagonal.

In fact, we can write $\tau^{\mathrm{reg}}_{\Sigma}(x)$ more explicitly as shown in the following lemma.

\begin{lemma}\label{lem:tau zeta} The following holds:
\[\tau^{\mathrm{reg}}_{\Sigma}(x)=-\dfrac{1}{4\pi}\log m^2+\int_{0}^{\infty}e^{-m^2t}\left(\theta_{\Delta_{\Sigma}}(x,t)-\frac{1}{4\pi t}\right)\,\dt
.\]
Furthermore, 
\[\tau^{\mathrm{reg}}_{\Sigma}(x)=-\dfrac{d}{d m^2}\dfrac{\d}{\d s}\Big|_{s=0}\zeta_\A(s,x)
.\]        %Negative sign added to the equation, which gave rise to the discrepancies in the sign
\begin{proof}
Using the proof of Lemma \ref{lemma:local_zeta_meromorphic_extension}, we observe that
\[\zeta_\A(1+\varepsilon,x)=\dfrac{1}{4\pi\varepsilon}-\dfrac{1}{4\pi}\log m^2+\int_{0}^{\infty}e^{-m^2t}\left(\theta_{\Delta_{\Sigma}}(x,t)-\frac{1}{4\pi t}\right)\,\dt+\mathrm{O}(\varepsilon)
\] as $\varepsilon\to 0$. Now, the first part of the lemma follows from the definition of the finite part. The second part of the lemma follows from %essentially 
differentiating in $m^2$ the explicit representation of $\dfrac{\d}{\d s}\Big|_{s=0}\zeta(s,x)$ in Lemma \ref{lemma:local_zeta_meromorphic_extension}.  
\end{proof}
\end{lemma}

\begin{remark} The zeta-regularized tadpole function $\tau_{\Sigma}^{\mathrm{reg}}$ is an invariant for the action of the group of isometries of the Riemannian metric on $\Sigma$. In particular, if the group of isometries acts transitively on $\Sigma$ then $\tau_{\Sigma}^{\mathrm{reg}}$ is constant.  
\end{remark}

More generally, we can define the local zeta function on a compact manifold with boundary. Let $\theta_\A(x,x',t)$ be the integral kernel of $A$ with respect to the Dirichlet boundary condition. Then, it is well known \cite{Mckean} for small $t$ that
\[\theta_\A(x,x',t)=e^{-m^2t}\frac{e^{-d(x,x')^2/4t}}{4\pi t}\left(a_0(x,x')+b_0(x,x')\sqrt{t}+a_1(x,x')t + \mathrm{O}(t^{3/2})\right)
.\] 
Here $b_0$ appears as a contribution from the boundary and it is supported at the boundary\footnote{For any fixed $x,x'$ away the boundary, $b_0$ will not appear in the asymptotic expansion. However, upon restricting to the diagonal and integrating, one will have a contribution coming from $b_0$. See e.g. \cite{Mckean}, for a detailed statement see \cite[Theorem 3.12]{Grieser2014}. }. Hence, 
\begin{equation}\label{local heat trace expansion}
\theta_\A(x,t)=\dfrac{e^{-m^2t}}{4\pi t}\left(1+b_0(x)\sqrt{t}+a_1(x)t+ \mathrm{O}(t^{3/2})\right) 
\end{equation}
as $t\to0$ where $b_0(x)=b_0(x,x)$ and $a_1(x)=a_1(x,x)$.

The {local zeta} function and the tadpole function $\tau^{\mathrm{reg}}_{\Sigma}(x)$ via zeta regularization are defined as above. The discussion above regarding the meromorphic extension of $\zeta_\A(s,x)$ does not change. In particular, $\zeta_\A(s,x)$ is holomorphic at $s=0$. Furthermore: %we have the following.

\begin{lemma} \label{lemma:local_zeta_meromorphic_extension_boundary_condition}
Let $\Sigma$ be a two dimensional compact Riemannian manifold with smooth boundary. Then, the following holds:
\begin{enumerate}[(i)]
\item 
\begin{align*}
\dfrac{\d}{\d s}\Big|_{s=0}\zeta_\A(s,x)&=\dfrac{1}{4\pi}m^2(\log m^2-1)
%\dfrac{\Gamma(-1/2)b_0(x)}{2\pi m}
-\dfrac{mb_0(x)}{2\sqrt{\pi}}
- \frac{a_1(x) \log m^2}{4 \pi} \\
&+
\int_{0}^{\infty}\dfrac{e^{-m^2t}\left(\theta_{\Delta}(x,t)-\tilde{g}(x,t)\right)}{t}\,\dt
\end{align*}
where $\tilde{g}(x,t)=\dfrac{1}{4\pi t}\big(1+b_0(x)\sqrt{t}+a_1(x) \,t \big)$.
\item 
\[\tau^{\mathrm{reg}}_{\Sigma}(x)=-\dfrac{\d}{\d m^2}\dfrac{\d}{\d s}\Big|_{s=0}\zeta_\A(s,x)
\]
\end{enumerate}
\begin{proof} 
One can show that
\begin{eqnarray}
\begin{array}{lll}\label{eq:local_zeta_function_expression_boundarycondition}
\zeta_\A(s,x)&=\dfrac{1}{4\pi}\left(\dfrac{1}{(s-1)m^{2s-2}}+\dfrac{\Gamma(s-1/2)\, b_0(x)}{\Gamma(s)\, m^{2s-1}}  + \dfrac{a_1(x)}{m^{2s}}
\right)\\
&+\dfrac{1}{\Gamma(s)}\int_{0}^{\infty}t^{s-1}e^{-tm^2}\left(\theta_{\Delta_{\Sigma}}(x,t)-\tilde{g}(x,t)\right)\,\d t
\end{array}
\end{eqnarray} Differentiating at $s=0$ and using $\Gamma(-1/2) = -2\sqrt{\pi}$, we get the first part of the lemma. From this expression, it also follows that
\begin{align*}\zeta_A(1+\varepsilon,x)=\dfrac{1}{4\pi\varepsilon}-\dfrac{1}{4\pi}\log m^2
%+\dfrac{b_0(x)}{4\sqrt{\pi} m}
+\int_{0}^{\infty}e^{-m^2t}\left(\theta_{\Delta_{\Sigma}}(x,t)-\dfrac{1}{4\pi t} %\tilde{g}(t,x)
\right)\,\dt +\mathrm{O}(\varepsilon)
\end{align*}
as $\varepsilon\to0$. This shows that
\begin{align*}\tau^{\mathrm{reg}}_{\Sigma}(x)= -\dfrac{1}{4\pi}\log m^2
%+\dfrac{b_0(x)}{4\sqrt{\pi} m}
+\int_{0}^{\infty}e^{-m^2t}\left(\theta_{\Delta_{\Sigma}}(x,t)- \dfrac{1}{4\pi t} 
%\tilde{g}(t,x)
\right)\,\dt.
\end{align*}
The second part of the lemma follows from a simple computation.
\end{proof}
\end{lemma}

The behavior of $\tau^{\mathrm{reg}}_{\Sigma}$ in the interior of $\Sigma$ is similar to the case when there is no boundary. However, it is not clear how it behaves near the boundary. We will show that the behavior of $\tau^{\mathrm{reg}}_{\Sigma}$ is comparable to that of the function $x\mapsto\log(d(x,\partial\Sigma))$ up to a bounded function. We begin the analysis with the following lemma. 

\begin{lemma}\label{lemma:log of distance to the boundary}
The function $x\mapsto \log d(x,\partial\Sigma)$ is in $L^p(\Sigma)$ for any $p\geq1$.
\begin{proof} Let $f(x)=\log d(x,\partial\Sigma)$. Assume that the Riemannian metric is a product metric near the boundary. Then, the volume form on a collar neighborhood the boundary can be written as $dt\wedge\dvol_{\partial\Sigma}$ and consequently it is possible to find $C>0$ and $a>0$ such that 
\[\int_{\Sigma}|f|^p\leq C\int_{0}^{a}|\log t|^p\,dt.
\]  
For the general case, let $T_r(\partial\Sigma)=\{x\in\Sigma:d(x,\partial\Sigma)\leq r\}$ be the tube around $\partial\Sigma$ constructed  using the normal vector field \cite[Chapter 3]{Gray2004}. Then, for small $r>0$, the volume form on $T_r$ can be written, using Fermi coordinates, in the form $h(t) dt\wedge \dvol_{\partial\Sigma}$  \cite[Theorem 9.22]{Gray2004}, where $h(t)=1+\mathrm{O}(t)$. In these coordinates, we have $d(x,\partial\Sigma)=t$,  which implies $x \mapsto \log d(x,\partial\Sigma)$ is integrable in this coordinate neighborhood by the result for the product metric. 
%\todo{S: I changed the argument a bit to address the general case, please check it.}
\end{proof}
\end{lemma}
Next, we compare the behavior of $\tau^{\mathrm{reg}}_{\Sigma}$ to that of the function $\kappa(x)=\fp_{s=1}\xi(s,x)$ where
\[\xi(s,x)=\dfrac{1}{\Gamma(s)}\int_{0}^{1}t^{s-1}\Xi(x,t)\,\dt
\] and
\begin{align*}
\Xi(x,t)=\dfrac{1}{4\pi t}\left(1-\textrm{exp}\left({-\dfrac{d(x,\partial\Sigma)^2}{t}}\right)\right).
\end{align*}
First, we show:
\begin{lemma}\label{lemma:model behavior near boundary}
 The function $\kappa-\dfrac{1}{2\pi}\log(d(\cdot ,\partial\Sigma))$ is bounded in $\Sigma$. 
\begin{proof}Let $u=d(x,\partial\Sigma)^2$. Then, it follows that
\[\xi(s,x)=-\dfrac{u^{s-1}}{4\pi \Gamma(s)}\left(\Gamma(1-s)-\int_{0}^{u}(e^{-t}-1)t^{-s}\,dt\right).
\]
Hence, 
\begin{align*}\xi(1+\varepsilon,x)=-\dfrac{u^{\varepsilon}}{4\pi\varepsilon \Gamma(\varepsilon)}\left(\Gamma(-\varepsilon)-\int_{0}^{u}(e^{-t}-1)t^{-(1+\varepsilon)}\,dt\right)
\end{align*}
Now, the lemma follows from the definition of $\kappa$ using the fact that $u^{\varepsilon}=1+\varepsilon\log u+\mathrm{O}(\varepsilon^2)$ as $\varepsilon\to0$. 
\end{proof}
\end{lemma}
% We can easily check that $\tau^{\mathrm{reg}}_{\Sigma}-\kappa$  is bounded. Now, combinining the discussion above, we have:
%\begin{prop} We have $\tau^{reg}_{\Sigma}$ is in $L^p$ for all $p\geq1$.
%\end{prop}
%Also,  the consistency of $\tau^{\mathrm{reg}}_{\Sigma}$ with the zeta-regularization is immediate.
%\begin{corollary}
%The zeta-regularized tadpole is consistent with the zeta-regularization: We have $\int_\Sigma \tau^{\mathrm{reg}}_{\Sigma}\,\dvol_{\Sigma}= - \dfrac{\d}{\d m^2}\log \det (\de_{\Sigma} + m^2)$.
%\end{corollary}
Following a strategy similar to the study of the heat kernel expansion on a manifold with boundary \cite{Mckean}, we can check that $\tau^{\mathrm{reg}}_{\Sigma}-\kappa$ is bounded. Now, combining the discussion above, we have the following result. %concerning the behavior of the zeta-regularized tadpole. 
 \begin{prop} %There exists $c\in\mathbb{R}$ such that $\tau^{\mathrm{reg}}_{\Sigma}-c\log(.,\partial\Sigma)$ is a bounded function. 
 The zeta-regularized tadpole behaves near $\dd\Sigma$ as 
 $$ \tau^\mr{reg}_\Sigma= \frac{1}{2\pi}\log d(\cdot,\dd \Sigma)+ f $$
 with $f$ a bounded function.
 \end{prop}
 As a consequence of this proposition and Lemma \ref{lemma:log of distance to the boundary}, we have:
\begin{corollary}\label{corollary:tau reg integrable}
% We have 
 $\tau^{reg}_{\Sigma}$ is in $L^p$ for all $p\geq1$.
\end{corollary}
Also,  the consistency of $\tau^{\mathrm{reg}}_{\Sigma}$ with the zeta-regularization is immediate.
\begin{corollary}\label{corollary: tau via d/dm^2 log det}
The zeta-regularized tadpole is consistent with the zeta-regularization: we have
\[\int_\Sigma \tau^{\mathrm{reg}}_{\Sigma}\,\dvol_{\Sigma}=  \dfrac{\d}{\d m^2}\log \det (\de_{\Sigma} + m^2).
\] % Negative sign removed
\end{corollary}
Furthermore, we also have compatibilty with the zeta-regularization in the strong sense.
\begin{lemma}\label{lem: tau zeta-reg is variational derivative in mass}
$\tau^{\mathrm{reg}}_{\Sigma}$ is compatible with zeta-regularization in the strong sense.
\end{lemma}
\begin{proof}
As usual, we write $A=\de_{\Sigma}+m^2$. We consider the function $F(\alpha) = \log\det(A+\alpha)$, which is defined for $\alpha \in U \subset C_c^\infty(\Sigma \setminus \partial \Sigma)$, where $U$ is a neighborhood of $0$. Then we have to show that 
$$DF(0)\alpha = \int_\Sigma \tau^{\mathrm{reg}}_{\Sigma}(x)\alpha(x) d^2 x$$ % Negative sign removed
where $D$ is the Fr\'echet %Gateaux 
derivative.
Applying again the Mellin transform, we have that  
\begin{align*}F(\varepsilon \alpha) &= -\dfrac{\d}{\d s}\Big|_{s=0}\int_\Sigma \zeta_{A+\epsilon \alpha}(s,x) =  -\dfrac{\d}{\d s}\Big|_{s=0}\int_\Sigma \frac{1}{\Gamma(s)}\int_0^\infty t^{s-1}\theta_{A+\varepsilon\alpha}(x,t)dt d^2 x 
\end{align*}
The heat operator for $A + \epsilon\alpha$ can be represented by perturbation theory as 
$$e^{-t(A+\epsilon \alpha)} = \sum_{k=0}^\infty \varepsilon^kW_k(t)$$ where  $W_0 = e^{-tA}$. The relevant object for us is the trace of the first order perturbation $ W_1(t)$ which can be computed as 
$$\mathrm{tr}\, W_1(t) = -t \int_\Sigma d^2 x\, \alpha(x)\,\theta_A(x,t).$$ We therefore conclude that 
\begin{align*}
 F(\varepsilon\alpha) = F(0) + \varepsilon  \dfrac{\d}{\d s}\Big|_{s=0}\int_\Sigma \frac{1}{\Gamma(s)}\int_0^\infty t^{s}\alpha(x)\theta_\A(x,t)dt d^2 x + O(\varepsilon^2)
\end{align*}
Now, exploiting $s\Gamma(s) = \Gamma(s+1)$ we realize that the term of order $\varepsilon$ is 
\begin{align*}DF(0)\alpha &= \int_\Sigma\underbrace{\dfrac{\d}{\d s}\Big|_{s=0} \frac{s}{\Gamma(s+1)}\int_0^\infty t^{s}\theta_\A(x,t)dt }_{=\tau^{\mathrm{reg}}_{\Sigma}(x)}\alpha(x) d^2 x \\
&=\int_\Sigma \tau^{\mathrm{reg}}_{\Sigma}(x)\alpha(x) d^2 x. 
\end{align*}
\end{proof}

\begin{prop}\label{prop: tau reg locality}
The assignment $\Sigma \mapsto \tau^\Sigma_{reg}$ is a local assignment. 
\end{prop}
For the proof we investigate another tadpole function, given by the so-called point splitting.
\subsection{Point-splitting tadpole}\label{sec: point-splitting tadpole}
Let $\Sigma$ be a compact oriented Riemannian manifold possibly with boundary. Let us consider the operator $A=\Delta_{\Sigma}+m^2$. We will always impose the Dirichlet boundary condition if $\Sigma$ has boundary  and we assume boundary to be closed. We recall that the Green's function on $\Sigma$ associated to $A$ %has the following properties:
can be written as
\[G_{\Sigma}(x,x')=-\frac{1}{2\pi}\log(d(x,x'))+\mathrm{H}(x,x')
\] 
near the diagonal of $\Sigma\times\Sigma$, where 
$\mathrm{H}$ is in $C^1$ in a neighborhood of the diagonal in $\Sigma\setminus\partial\Sigma\times\Sigma\setminus\partial\Sigma$.  The singular support of the distribution $G_{\Sigma}$ is the diagonal. 
%Moreover, $\mathrm{H}$ is $C^1$ in a neighborhood of the diagonal in $\Sigma\setminus\partial\Sigma\times\Sigma\setminus\partial\Sigma$.  

\begin{definition}\label{tadpole:splitting} We define
$\tau^{\mathrm{split}}_{\Sigma}(x):=\lim_{x'\to x}\left[G_{\Sigma}(x,x')+\frac{1}{2\pi}\log(d(x,x'))\right].$ 
\end{definition}

We can think of $\tau^{\mathrm{split}}_{\Sigma}$ as a way to regularize $G_{\Sigma}$ on the diagonal via ``point-splitting". In the following lemma, we state some properties of $\tau^{\mathrm{split}}_{\Sigma}.$

\begin{lemma}\label{lemma:splitting_tadpole_behavior} 
The point-splitting tadpole $\tau^{\mathrm{\sp}}_{\Sigma}$ is $C^1$ on $\Sigma\setminus\partial\Sigma$ and $\tau^{\mathrm{\sp}}_{\Sigma}\in \mathrm{L}^p(\Sigma)$ for any $p\geq 1$. 
Moreover,
\[\tau^{\mathrm{\sp}}_{\Sigma}(x)=-\dfrac{\log m^2}{4\pi}+\dfrac{\log 2-\gamma}{2\pi}+\int_{0}^{\infty}e^{-tm^2}\left(\theta_{\Delta_{\Sigma}}(x,t)-\dfrac{1}{4\pi t}\right)\,\dt
\] where $\gamma$ is the Euler's constant.
\begin{proof}
The fact that the point-splitting tadpole $\tau^{\mathrm{\sp}}_{\Sigma}$ is $C^1$ on $\Sigma\setminus\partial\Sigma$ follows from the definition. To show $\tau^{\mathrm{\sp}}_{\Sigma}\in L^p(\Sigma)$, it suffices to show the function $f$ on $\Sigma$ defined by $x\mapsto \frac{1}{2\pi}\log(d(x,\partial\Sigma))$ is in $L^p(\Sigma)$ as $f-\tau^{\mathrm{\sp}}_{\Sigma}$ is locally bounded in a collar neighborhood of the boundary. We have already shown in Lemma \ref{lemma:log of distance to the boundary} that $f\in L^p(\Sigma)$.
This completes the proof of $\tau^{\mathrm{\sp}}_{\Sigma}\in L^p(\Sigma)$. Another proof of this assertion can be given by comparing $\tau^{\mathrm{\sp}}_{\Sigma}$ with $\tau^{\mathrm{reg}}_{\Sigma}$ using Corollary \ref{corollary:tadpole difference}, which uses only the second part of this lemma, and applying Corollary \ref{corollary:tau reg integrable}.

For the last part, let us recall that 
\[G_{\Sigma}(x,x')=\int_{0}^{\infty}e^{-tm^2}\theta_{\Delta_{\Sigma}}(x,x',t)\,\dt \] and 
\[\dfrac{1}{2\pi}K_0(md(x,x'))=\int_{0}^{\infty}\dfrac{1}{4\pi t}e^{-tm^2-d(x,x')^2/4t}\,\dt
\] where $K_0(z)$ is the modified Bessel's function. The $z\rightarrow 0$ asymptotics of $K_0(z)$ implies %(CITATION!) %\todo{K: Do we need to cite a reference for asymptotic behaviour of Bessel Function?} that
\[\dfrac{1}{2\pi}K_0(md(x,x'))=-\dfrac{\log(md(x,x'))}{2\pi}+\dfrac{\log2-\gamma}{2\pi}+\mathrm{O}(m^2d(x,x')^2)
\] as $d(x,x')\to 0$. 
We can rewrite this as:
\[\dfrac{\log(d(x,x'))}{2\pi}=-\dfrac{1}{2\pi}K_0(md(x,x'))-\dfrac{\log m^2}{4\pi}+
\dfrac{\log 2-\gamma}{2\pi}+\mathrm{O}(m^2d(x,x')^2)
\] as $d(x,x')\to 0$.
Using this, we see that
\begin{align*}
&\lim_{x'\to x}\left(G_{\Sigma}(x,x')+\frac{1}{2\pi}\log (d(x,x'))\right)\\
&=-\dfrac{\log m^2}{4\pi}+\dfrac{\log 2-\gamma}{2\pi}+\int_{0}^{\infty}e^{-tm^2}\left(\theta_{\Delta_{\Sigma}}(x,t)-\dfrac{1}{4\pi t}\right)\,\dt.
\end{align*}
\end{proof}
\end{lemma}

\begin{corollary} \label{corollary:tadpole difference}
For $x\in\Sigma\setminus\partial\Sigma$, we have 
$\tau^{\mathrm{reg}}_\Sigma(x) - \tau^{\sp}_{\Sigma}(x) = \dfrac{\gamma-\log 2}{2\pi}.$ 
\end{corollary}

We note that the point-splitting tadpole function is invariant under the action of the group of the isometries of the Riemannian metric. From this it follows that:

\begin{prop} Let $\Sigma$ be a closed oriented Riemannian manifold such that the group of isometries act transitively on $\Sigma$, then $\tau_{\Sigma}^{\sp}$ is constant on $\Sigma$. In particular, if $\Sigma$ is a closed Riemannian surface with constant scalar curvature, then $\tau_{\Sigma}^{\sp}$ is constant. 
\end{prop}

Finally, we can give two examples of local assignments.
\begin{prop}
The assignment $\Sigma \mapsto \tau^{\sp}_\Sigma$ is a local assignment of tadpole functions. In particular, $\tau^{\mathrm{reg}}_{\Sigma}(x)$ is also a local assignment of tadpole functions.
\end{prop}
\begin{proof}
 We can show $\tau^{\sp}_\Sigma$ is a local assignment by a direct application of the gluing formula for $G_{\Sigma}$ and we can complete the proof using Corollary \ref{corollary:tadpole difference}.
\end{proof}

We end this subsection with the following remark concerning the appearance of tadpole functions in other context in the literature. 
\begin{remark}
The zeta-regularized tadpole and the point-splitting tadpole appear in the study of conformally covariant elliptic operators such as Yamabe operator and Paneitz operator on a closed Riemannian manifold $\Sigma$ \cite{AH2005}. In this context, these functions are known as the mass function of the operators and they are used in the study of mass theorems, regularized traces, and conformal variation of regularized traces, we refer to \cite{SY1979,SY1981,AH2005,Kategafa2008,Kate2008,Ludewig2017} for details.  
\end{remark}

\subsection{Some comments on the relation to RG flow in 2D scalar field theory}\label{sec: RG flow}
%\marginpar{TO EDIT}
In this paper we regularize 2d partition functions by choosing a tadpole function. In the above we constructed two particular examples of tadpole functions, the ``zeta-regularized'' and the ``point-splitting'' one. The purpose of this subsection is to relate these to the 
setup of renormalization group and %more general notion of
 RG flow. 
\subsubsection{Renormalized action (action with counterterms)}\label{ss: renormalized action}
Consider the 2D scalar field theory with action
\begin{equation*}
S(\phi)= \underbrace{\int_\Sigma \frac12 d\phi\wedge *d\phi +\frac{m^2}{2}\phi^2 \,dvol}_{S_{free}} + p(\phi)\, dvol
\end{equation*}
with $p(\phi)=\sum_n \frac{p_n}{n!}\phi^n$ a polynomial interaction (or more generally a formal power series).

Say, we are interested in the normalized path integral 
\begin{equation}\label{Z_norm}
Z_{norm}=\frac{\int \DD\phi\; e^{-\frac{1}{\hbar}S(\phi)}}{\int\DD\phi\;e^{-\frac{1}{\hbar}S_{free}(\phi)}}
\end{equation}
and we define it by a perturbative expansion with Green's function regularized via proper time cut-off:
\begin{equation*}
G_\Lambda(x,x')=\int_{1/\Lambda^2}^\infty dt\; e^{-m^2 t} \theta_\Delta(x,x',t)
\end{equation*}
with $\theta_\Delta$ the heat kernel for the Laplacian, $x,x'\in\Sigma$ and $\Lambda$ a very large cut-off having the dimension of mass.
We note that, for $x'\neq x$, $\lim_{\Lambda\ra\infty}G_\Lambda(x,x')=G(x,x')$ exists, whereas on the diagonal we have the asymptotic behavior
\begin{equation}\label{G_Lambda}
G_\Lambda(x,x)\underset{\Lambda\ra\infty}{\sim} \frac{\log \Lambda}{2\pi} + \underbrace{\widetilde{\tau}(x)}_{\mr{finite}}
\end{equation}

\begin{lemma}
The finite part $\til{\tau}(x)$ appearing in the r.h.s. of (\ref{G_Lambda}) differs from the zeta-regularized tadpole by a universal constant:
\begin{equation}\label{T tilde = Treg + C}
\til{\tau}(x)=\tau^\mr{reg}(x)-\frac{\gamma}{4\pi}
\end{equation}
with $\gamma$ the Euler constant.
\end{lemma}
\begin{proof}
%\footnote{
%Proof of (\ref{T tilde = Treg + C}): 
Indeed, we find
\begin{multline*}
 G_\Lambda(x,x)=\int_{1/\Lambda^2}^\infty dt\, e^{-m^2 t} \theta_\Delta(x,t) = \int_{1/\Lambda^2}^\infty dt\, e^{-m^2 t} \left(\theta_\Delta(x,t)-\frac{1}{4\pi t} \right) +\\
 + \underbrace{\int_{1/\Lambda^2}^\infty dt\, e^{-m^2 t} \frac{1}{4\pi t} }_{\frac{1}{4\pi} E_1\big(\frac{m^2}{\Lambda^2}\big)} \; \underset{\Lambda\ra\infty}{\sim}\; \frac{\log\Lambda}{2\pi}\underbrace{-\frac{\gamma+\log m^2}{4\pi}+\int_0^\infty dt\, e^{-m^2 t} \left(\theta_\Delta(x,t)-\frac{1}{4\pi t} \right)}_{\til{\tau}(x)} +O\left(\frac{m^2}{\Lambda^2}\right).
 \end{multline*}
Here $E_1(u)=\int_u^\infty dt\, \frac{e^{-t}}{t}$ is the exponential integral and we used its asymptotic behavior $E_1(u)\sim -\log u-\gamma+O(u)$ at $u\ra 0$. Comparing this formula for $\til{\tau}$ with the result for $\tau^\mr{reg}$ (Lemma \ref{lem:tau zeta}), we obtain (\ref{T tilde = Treg + C}).
%} \marginpar{We can remove this long footnote: I don't know how helpful it is.}
\end{proof}

For the normalized path integral (\ref{Z_norm}) to be finite, we must assume that the coefficients of $p(\phi)=p_\Lambda(\phi)$ in the numerator depend on $\Lambda$ in such a way that the limit $\lim_{\Lambda\ra \infty}Z_{norm}$ exists.  For that to happen,  $p_\Lambda(\phi)$ must have the following form:
\begin{multline}\label{p_Lambda}
p_\Lambda(\phi)=\sum_{n\geq 0} \frac{p_n}{n!} \sum_{k=0}^{\left[\frac{n}{2}\right]}(-1)^k (2k-1)!! \left(\begin{array}{c}n \\ 2k\end{array} \right) \cdot \left(\frac{\hbar}{2\pi}\log\Lambda\right)^k \phi^{n-2k}= \\
= \sum_{n\geq 0} p_n \sum_{k=0}^{\left[\frac{n}{2}\right]}\frac{1}{k!} \left(-\frac{\hbar}{4\pi}\log\Lambda\right)^k\cdot \frac{\phi^{n-2k}}{(n-2k)!}
\end{multline}
-- here we are essentially subtracting from $p_\mr{naive}(\phi)=\sum \frac{p_n}{n!}\phi^n$ the ``counterterms'' compensating for the tadpole divergencies encountered when computing the path integral (\ref{Z_norm}) using $p_\mr{naive}$.

Note that (\ref{p_Lambda}) satisfies the differential equation
\begin{equation}\label{RG flow eq for p_Lambda}
\frac{\dd}{\dd\log\Lambda}\; p_\Lambda(\phi)=-\frac{\hbar}{4\pi}\;\frac{\dd^2}{\dd\phi^2}\, p_\Lambda(\phi)
\end{equation}
- one can see at as a heat equation with ``time'' coordinate $\log\Lambda$ and ``space'' coordinate $\phi$, and (\ref{p_Lambda}) is the general solution with initial condition given by $p_\mr{naive}(\phi)$ at ``time" $\log\Lambda=0$.\footnote{
One can call (\ref{RG flow eq for p_Lambda}) the RG flow equation ``at the ultraviolet end'': it tells how the local counterterms change when the cut-off is changed infinitesimally.
}

Some examples of solutions:%\marginpar{format better?}
\vspace{0.2cm}

\begin{tabular}{c|l}
$\frac{\phi^2}{2}-\frac{\hbar}{4\pi}\log\Lambda$
& shift of mass (gets additively renormalized), \\
$	 \Lambda^{-\frac{\hbar}{4\pi}\alpha^2} e^{\alpha \phi}$ 
&
potential of Liouville theory, %gets renormalized multiplicatively
% attains an anomalous dimension, 
 \\
 $	 \Lambda^{\frac{\hbar}{4\pi}\alpha^2} \cos(\alpha \phi)$ &
 potential of sine-Gordon theory.
\end{tabular}
\vspace{0.2cm}

In the last two examples the potential is multiplicatively renormalized (attains an anomalous dimension).\footnote{Cf. the fact that in free massless scalar field theory -- a prototypical CFT -- $:e^{\alpha\phi}:$ is a vertex operator of holomorphic/antiholomorphic dimension $h=\bar{h}=-\frac{\hbar}{8\pi}\alpha^2$ (in our normalization convention), and thus total scaling dimension $h+\bar{h}=-\frac{\hbar}{4\pi}\alpha^2$ and spin $h-\bar{h}=0$.
}

\subsubsection{Tadpoles vs. RG flow (``petal diagram resummation'')}
Let $Z^{\tau,p}$ be the partition function on a surface $\Sigma$ (possibly with boundary) for the massive scalar field with interaction potential 
\begin{equation} \label{p(phi)}
p(\phi)=\sum_n \frac{p_n}{n!}\phi^n, 
\end{equation} 
defined using the tadpole function $\tau=\tau(x)$. 

%In this section, we consider $\Sigma$ to be a closed surface. 
We denote $\Fun$ %\marginpar{Ok notation, or change it?} 
the space of smooth functions in the interior of $\Sigma$ which behave as $O(|\log d(x,\partial \Sigma)|^N)$ near the boundary, for some power $N$.

We will consider the setup where the coefficients $p_n$ of the interaction potential themselves are allowed to be functions on $\Sigma$ valued in power series in $\hbar$, i.e., $p=p(\phi,x,\hbar)\in \Fun[[\phi,\hbar]]$.\footnote{The reason for introducing this extended setup is that the transformation (\ref{phi tilde}) below generally (for a non-constant tadpole function $\tau$) transforms a potential with constant coefficients to one with non-constant coefficients.}
We have the following.
%\begin{figure}[ht]
%\begin{tikzpicture}
%\node[bulk] (x0) at (0,0) {};
%\node[shape=coordinate] at (20:1) {} edge[bubu] (x0);
%\node[shape=coordinate] at (90:1) {} edge[bubu] (x0);
%\node[shape=coordinate] at (160:1) {} edge[bubu] (x0);
%\node[shape=coordinate] at (230:1) {} edge[bubu] (x0);
%\node[shape=coordinate] at (300:1) {} edge[bubu] (x0);
%\node at (1.5,0) {$+$};
%\begin{scope}[shift={(3,0)}]
%\node at (-1.5,0) {$+$};
%\node[bulk] (x0) at (0,0) {};
%\node[shape=coordinate] at (20:1) {} edge[bubu] (x0);
%\node[shape=coordinate] at (90:1) {} edge[bubu] (x0);
%\node[shape=coordinate] at (160:1) {} edge[bubu] (x0);
%\draw[bubu] (x0) .. controls (230:1.5) and (300:1.5) .. (x0);
%\end{scope}
%\begin{scope}[shift={(6,0)}]
%\node at (-1.5,0) {$+$};
%\node[bulk] (x0) at (0,0) {};
%\draw[bubu] (x0) .. controls (20:1.5) and (90:1.5) .. (x0);
%\node[shape=coordinate] at (160:1) {} edge[bubu] (x0);
%\draw[bubu] (x0) .. controls (230:1.5) and (300:1.5) .. (x0);
%\end{scope}
%\end{tikzpicture}
%\caption{Petal diagram resummation}
%\end{figure}
\begin{prop}[Petal diagram resummation] \label{prop: petal resummation} \leavevmode
\begin{enumerate}[(i)]
\item \label{prop petal i} We have the equality of partition functions
\begin{equation}\label{Z with tadpole via Z without tadpole}
 Z^{\tau,p}=Z^{0,\til{p}} 
\end{equation}
Here the right hand side is defined with zero tadpole and 
\begin{equation}\label{phi tilde}
\til{p}(\phi,x,\hbar):=\sum_{n\geq 0} p_n(x,\hbar) \sum_{k=0}^{\left[\frac{n}{2}\right]} \frac{1}{k!}\left(\frac{\hbar\, \tau(x)}{2}\right)^k \cdot \frac{\phi^{n-2k}}{(n-2k)!}
\end{equation}

\item  \label{prop petal ii} Denote the r.h.s. of (\ref{phi tilde}) by $\RG_\tau(p)$. We have 
%\begin{itemize}
%\item 
$\RG_{\tau_1+\tau_2}(p)=\RG_{\tau_1}(\RG_{\tau_2}(p))$.\footnote{
Thus, %$\RG$ defines an action of the additive group $\RR$ at each point $x$ of $\Sigma$ on interaction potentials evaluated at $x$ -- the ``local RG flow.'' Put differently, 
$\RG$ defines an action of the additive group $\Fun$ on interaction potentials $p\in \Fun[[\phi,\hbar]]$ -- the ``local RG flow.'' 
%given by formal power series in $\phi$ (and $\hbar$) with coefficients in $C^\infty(\Sigma)$.
Note that, for a constant function $\tau$, $\RG_\tau$ transforms potentials with constant coefficients $p\in \RR[[\phi,\hbar]]$ to potentials with constant coefficients.
}

\item  \label{prop petal iii} If $q=\RG_{\tau_1-\tau_2}(p)$, then  $Z^{\tau_1,p}=Z^{\tau_2,q}$.
%\end{itemize}

\item \label{prop petal RG flow eq} The transformed potential (\ref{phi tilde}) satisfies the ``local RG flow equation'':
\begin{equation}\label{RG flow eq local}
\frac{\partial}{\partial \tau} \RG_\tau(p) = \frac{\hbar}{2}\, \frac{\partial^2}{\partial \phi^2} \RG_\tau(p)
\end{equation}
which holds pointwise on $\Sigma$.

\item  \label{prop petal iv} For a potential $p(\phi)=\sum_{j} c_j e^{\alpha_j \phi}$ given by a sum of exponents, with $\alpha_j,c_j$ independent on $\phi$ (but possibly depending on $x,\hbar$), the corresponding transformed potential (\ref{phi tilde}) is:
\begin{equation*}
%\til{p}(\phi)
\RG_\tau(p)
=\sum_j c_j e^{\frac{\hbar \tau}{2} \alpha_j^2} e^{\alpha_j \phi}
\end{equation*}
\end{enumerate}
\end{prop}
%Note that, generally, the potential $\til{p}=\til{p}(\phi,x)$ in the right hand side of (\ref{phi tilde}) depends explicitly on $x$ (not just via $\phi$), since the tadpole function $T$ appearing in (\ref{phi tilde}) depends on $x$.
\begin{proof}
For (\ref{prop petal i}), one shows (\ref{Z with tadpole via Z without tadpole}) as a resummation of perturbation theory: summation of the ``petal diagrams'' %\marginpar{picture?} 
for the theory with potential $p$ yields the vertices of the theory with potential 
\begin{equation}\label{til p in proof}
\til{p}=\sum_{n\geq 0}\sum_{k=0}^{\left[\frac{n}{2}\right]} \frac{p_n}{n!}\phi^{n-2k} \left(\begin{array}{c}
n \\ 2k
\end{array}\right) (2k-1)!!\; (\hbar \tau)^k
\end{equation}
where the combinatorial coefficient $\left(\begin{array}{c}
n \\ 2k
\end{array}\right) (2k-1)!!$ counts the number of ways to attach $k$ edges to a vertex with $n$ incident half-edges. Expression (\ref{til p in proof}) simplifies to (\ref{phi tilde}).
\begin{figure}[ht]
\begin{tikzpicture}
\node[bulk] (x0) at (0,0) {};
\node[shape=coordinate] at (20:1) {} edge[bubu] (x0);
\node[shape=coordinate] at (90:1) {} edge[bubu] (x0);
\node[shape=coordinate] at (160:1) {} edge[bubu] (x0);
\node[shape=coordinate] at (230:1) {} edge[bubu] (x0);
\node[shape=coordinate] at (300:1) {} edge[bubu] (x0);
\node at (1.5,0) {$+$};
\begin{scope}[shift={(3,0)}]
\node at (-1.5,0) {$+$};
\node[bulk] (x0) at (0,0) {};
\node[shape=coordinate] at (20:1) {} edge[bubu] (x0);
\node[shape=coordinate] at (90:1) {} edge[bubu] (x0);
\node[shape=coordinate] at (160:1) {} edge[bubu] (x0);
\draw[bubu] (x0) .. controls (230:1.5) and (300:1.5) .. (x0);
\end{scope}
\begin{scope}[shift={(6,0)}]
\node at (-1.5,0) {$+$};
\node[bulk] (x0) at (0,0) {};
\draw[bubu] (x0) .. controls (20:1.5) and (90:1.5) .. (x0);
\node[shape=coordinate] at (160:1) {} edge[bubu] (x0);
\draw[bubu] (x0) .. controls (230:1.5) and (300:1.5) .. (x0);
\end{scope}
\end{tikzpicture}
\caption{Petal diagram resummation}
\end{figure}
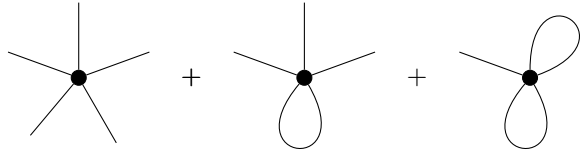

Item (\ref{prop petal ii}) is straightforward: denoting the $n$-th coefficient $p_n$ in $p(\phi)$ (normalized as in (\ref{p(phi)})) by $[p]_n$, we have from (\ref{phi tilde}) that
\begin{equation*}
[\RG_T(p)]_n  \sum_{k\geq 0} \left(\frac{\hbar\, \tau}{2}\right)^k p_{n+2k}
\end{equation*}
and therefore
\begin{multline*}
 [\RG_{\tau_1}(\RG_{\tau_2}(p))]_n=\sum_{k_1,k_2\geq 0} \frac{1}{k_1!}  \left(\frac{\hbar\, \tau_1}{2}\right)^{k_1} \frac{1}{k_2!}  \left(\frac{\hbar\, \tau_2}{2}\right)^{k_2} p_{n+2k_1+2k_2} \\
 \underset{l=k_1+k_2}{=}\quad
\sum_{l\geq 0}  \left(\frac{\hbar\, (\tau_1+\tau_2)}{2}\right)^l p_{n+2l} = [\RG_{\tau_1+\tau_2}(p)]_n  
\end{multline*}

Item (\ref{prop petal iii}) is a generalization of (\ref{prop petal i}) (since the case $q=0$ is (\ref{prop petal i})) and it follows from (\ref{prop petal i}) and (\ref{prop petal ii}):
$$
Z^{\tau_2,q} \underset{(\ref{prop petal i})}{=} Z^{0, \RG_{\tau_2}(q)} = Z^{0,\RG_{\tau_2}(\RG_{\tau_1-\tau_2}(p))}
\underset{(\ref{prop petal ii})}{=} Z^{0, \RG_{\tau_1}(p)} \underset{(\ref{prop petal i})}{=} Z^{\tau_1,p}
$$

The equation (\ref{prop petal RG flow eq}) follows immediately from (\ref{phi tilde}) by applying the relevant derivatives to the r.h.s.

Item (\ref{prop petal iv}) is the observation that when applied to an exponential $p(\phi)=e^{\alpha \phi}=\sum_{n\geq 0}\frac{\alpha^n}{n!}\phi^n$, the transformation (\ref{phi tilde}) yields
$$ \sum_{n\geq 0} \sum_{k=0}^{\left[\frac{n}{2}\right]} \frac{1}{k!}\left(\frac{\hbar \tau}{2}\right)^k \alpha^n \frac{\phi^{n-2k}}{(n-2k)!} = \sum_{n\geq 0} \sum_{k=0}^{\left[\frac{n}{2}\right]} \frac{1}{k!}\left(\frac{\hbar \tau}{2}\alpha^2\right)^k  \frac{(\alpha \phi)^{n-2k}}{(n-2k)!}
=e^{\frac{\hbar \tau}{2}\alpha^2} e^{\alpha \phi}
$$
Then (\ref{prop petal iv}) follows by $\Fun[[\hbar]]$-linearity of the transformation $\RG_T$. 
\end{proof}

Proposition \ref{prop: petal resummation} implies the following.
\begin{corollary} One has the equality
\begin{equation} \label{Z^(Treg,p) = Z^(T_Lambda,p_Lambda)}
Z^{\tau^\mr{reg},p}=Z^{\tau_\Lambda,p_\Lambda} 
\end{equation}
Here the l.h.s. is the partition function for an interaction potential $p(\phi)$, calculated using the zeta-regularized tadpole $\tau^\mr{reg}(x)$. The r.h.s. is the partition function with the tadpole $\tau_\Lambda(x):=G_\Lambda(x,x)$ (with $G_\Lambda$ defined via proper time cut-off, as in (\ref{G_Lambda})) and with the ``renormalized'' interaction potential (or ``potential with counterterms'') given by
\begin{equation*}
p_\Lambda=\mathcal{R}_{\tau^\mr{reg}-\tau_\Lambda}(p) \quad \underset{\Lambda\ra\infty }{\sim}\quad  
\RG_{-\frac{\log\Lambda}{2\pi}+\frac{\gamma}{4\pi}}(p)
\end{equation*}
\end{corollary}

Here in the last point we used the result (\ref{T tilde = Treg + C}).

The r.h.s. of (\ref{Z^(Treg,p) = Z^(T_Lambda,p_Lambda)}) is \emph{almost} the same as the computation of the perturbative path integral for the theory using the cut-off-regularized Green's function (\ref{G_Lambda}), with the action including counterterms, which are fine-tuned -- see (\ref{Z^(Treg,p) = Z^(T_Lambda,p_Lambda)}) -- so that the path integral is finite. ``Almost'' -- because in (\ref{Z^(Treg,p) = Z^(T_Lambda,p_Lambda)}) only the Green's functions in the tadpoles are regularized while the Green's functions between distinct vertices are the exact ones. However, the distinction between these two regularizations for Feynman diagrams becomes negligible as $\Lambda\ra\infty$.

Another way to present the result (\ref{Z^(Treg,p) = Z^(T_Lambda,p_Lambda)}) is:
\begin{equation} \label{Z^(Ttil,p') = Z^(T_Lambda, p_Lambda)}
Z^{\til{\tau},p'} \underset{\Lambda\ra\infty}{\leftarrow} Z^{\tau_\Lambda,p_\Lambda}
\end{equation}
where $\til{\tau}$ is as in (\ref{G_Lambda}) -- the cut-off-renormalized tadpole (i.e., with cut-off regularization imposed and with the singular term subtracted), with $p'=\RG_{\frac{\gamma}{4\pi}}(p)$ a finite transformation of the potential (arising from the difference in zeta vs cut-off renormalization schemes). Here $p_\Lambda=\RG_{-\frac{\log\Lambda}{2\pi}}(p')$. Note that this is the same as formula (\ref{p_Lambda}) from Section \ref{ss: renormalized action} if we identify $p'=p_\mr{naive}$.  Thus, in the asymptotic equality (\ref{Z^(Ttil,p') = Z^(T_Lambda, p_Lambda)}) we either subtract a singular part from the tadpole (in the l.h.s.), or we add counterterms to the action (in the r.h.s.).

\textbf{Cutting into tiny squares (a heuristic picture).}\footnote{We call it a heuristic picture because it relies on cutting with corners which is yet to be fully understood. %a work in progress by the authors.
} Consider a cellular subdivision $X_\epsilon$ of the surface $\Sigma$ into small squares $sq_i$, of linear size of order $\epsilon=\frac{1}{\Lambda}$, with $\epsilon \ra 0$. One can consider two different pictures (local assignments of tadpoles):
\begin{enumerate}[I]
\item Set the tadpole functions to zero for each small square, $\tau_{sq_i}=0$. By the gluing formula for tadpoles, this leads to a glued tadpole $\tau_\Sigma\sim -\frac{\log \epsilon}{2\pi}+\mr{(finite\; part)}$ on the surface; $\tau_\Sigma$ is a version of cut-off regularized tadpole $G_\Lambda(x,x)$, see (\ref{G_Lambda}).
\item Set the tadpole functions for the small squares and for $\Sigma$ to their zeta-regularized values. Then on a small square,
\begin{comment}
\footnote{
Explicit computation of the zeta-regularized tadpole on a rectangle of size $l_1\times l_2$ yields $\tau^\mr{reg}(x,y)=-\frac{\log m^2}{4\pi}+\int_0^\infty dt\, e^{-m^2 t} (\theta_\Delta((x,y),t)-\frac{1}{4\pi t})$ where $\theta_\Delta((x,y),t)=\frac{1}{4\pi t}(\vartheta(0,\frac{i l_1^2}{\pi t}) - e^{-\frac{x^2}{t}} \vartheta(-\frac{il_1 x}{\pi t},\frac{i l_1^2}{\pi t}))\cdot (\vartheta(0,\frac{i l_2^2}{\pi t}) - e^{-\frac{y^2}{t}} \vartheta(-\frac{il_2 y}{\pi t},\frac{i l_2^2}{\pi t})) $. In this answer, one can set $l_1=l_2=\epsilon$ and consider the asymptotics $\epsilon\ra 0$.
} 
\end{comment}
we have $\tau_{sq_i}\sim \frac{\log \epsilon}{2\pi}+\mr{(finite\; part)}$ (this is the $\epsilon\ra 0$ asymptotics of an explicit answer for a flat square) %(comes from an explicit computation of the zeta-regularized tadpole function on a flat square) 
and $\tau_\Sigma$ is finite ($\epsilon$-independent).
\end{enumerate}
In the first picture, we need to define the partition function using the renormalized potential $p_\Lambda$, in order to have a finite result; in the second picture, we are taking the non-renormalized potential $p$ and have a finite result.

\section{Formal Fubini Theorem and Atiyah-Segal gluing}\label{sec:Fubini}
In this section, we finally prove the gluing formula for the perturbative partition function. It comes in different flavors, according our choice of tadpole function. If $\tau$ is a local assignment of tadpole functions, we write $Z^\tau_\Sigma := Z^{\tau_\Sigma}_\Sigma$. 
%In particular, we denote the partition function defined using the zeta-regularized tadpole $\tau_{reg}$ by $Z^\zeta$.

\subsection{The gluing formula}
Let $\Sigma$ be a two-dimensional compact Riemannian cobordism with boundary $\partial \Sigma = Y_L \sqcup Y_R$. Let $Y$ be a collection of circles in $\Sigma$ such that $\Sigma=\Sigma_{{L}}\cup_{Y}\Sigma_{{R}}$, with $Y_L \subset \Sigma_L$ and $Y_R \subset \Sigma_R$. 
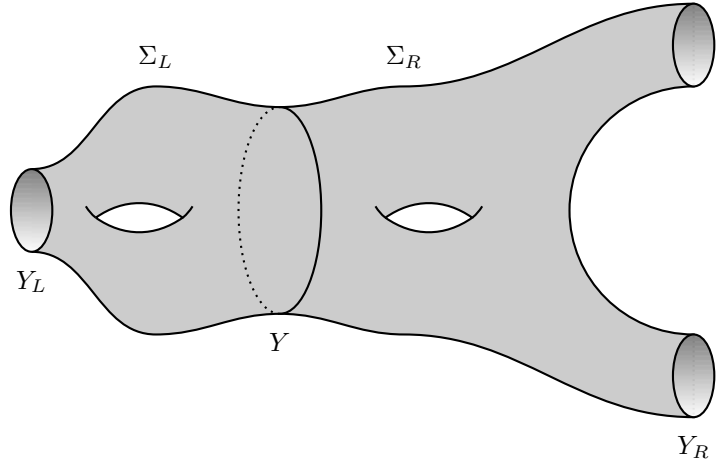
\begin{figure}
\centering
 \begin{tikzpicture}[scale=.55]
  \coordinate (x0) at (-6,0);
  \coordinate (x0p) at (-6,1);
    \coordinate (x0pp) at (-6,-1);

  \coordinate (x1) at (-3,3);
  \coordinate (x2) at (0,2.5);
  \coordinate (x3) at (3,3);
  \coordinate (x35)at (10,5);
  \coordinate (x4) at (7,0);
  \coordinate (x45)at (10,-5);
  \coordinate (x5) at (3,-3);
  \coordinate (x6) at (0,-2.5);
  \coordinate (x7) at (-3,-3);
  
  \coordinate (b1) at (-4,-.5);
  \coordinate (b2) at (3,-.5);
  
  \coordinate (y0) at (-2,1.5);
  \coordinate (y1) at (0,1.3);
  \coordinate (y2) at (2,1.1);
  \coordinate (y3) at (-1,-.5);
  \coordinate (y4) at (1,-.7);
  \coordinate (y5) at (0,-2);

%Labels
    \node[label=above:{$\Sigma_L$}] at (x1){};
    \node[label=above:{$\Sigma_R$}] at (x3){};

  \filldraw[fill=black!20, thick, draw=black] (x0p) to[in=180,out=0] 
   (x1)  to[in=180,out=0] (x2) to[in=180,out=0] (x3) to[in=180,out=0] (x35) to ($(x35)+(0,-2)$) to[out=180,in=90]  (x4) to[out=270, in=180] ($(x45) + (0,2)$) to (x45) to[in=0,out=180] (x5) to[in=0,out=180] (x6) to[in=0,out=180] (x7) to[in=0,out=180] (x0pp);
  %center
    \draw[thick] (x6) arc (270:450:1cm and 2.5cm);
    \draw[thick,dotted] (x6) arc (270:90:1cm and 2.5cm);
    \node[label=below:{$Y$}] at (x6){};

    %left boundary
    \shadedraw[thick] (x0) ellipse (0.5cm and 1cm);
    \node[label=below:{$Y_L$}] at ($(x0)+(0,-1)$){};

    %right boundary
    \shadedraw[thick] ($(x35)+(0,-1)$) ellipse (0.5cm and 1cm);
    \shadedraw[thick] ($(x45)+(0,1)$) ellipse (0.5cm and 1cm);
\node[label=below:{$Y_R$}] at (x45){};
%%GENUS
  \filldraw[fill=white, draw=black, thick]($(b1)+(-.45,.325)$) to[out=35, in=145] ($(b1)+(1.65,.325)$)to[in=-35, out=215]($(b1)+(-.45,.325)$);
  \draw[color=black, thick]($(b1)+(-.45,.325)$) to[out=141, in=145+180-22] ($(b1)+(-.69,.6)$);
  \draw[color=black, thick]($(b1)+(1.65,.325)$) to[out=39, in=35+180+22] ($(b1)+(1.89,.6)$);
  
  \filldraw[fill=white, draw=black, thick]($(b2)+(-.45,.325)$) to[out=35, in=145] ($(b2)+(1.65,.325)$)to[in=-35, out=215]($(b2)+(-.45,.325)$);
  \draw[color=black, thick]($(b2)+(-.45,.325)$) to[out=141, in=145+180-22] ($(b2)+(-.69,.6)$);
  \draw[color=black, thick]($(b2)+(1.65,.325)$) to[out=39, in=35+180+22] ($(b2)+(1.89,.6)$);

\end{tikzpicture}
\caption{A cobordism $\Sigma$ from $Y_L$ to $Y_R$ with a decomposition $\Sigma = \Sigma_L \cup_Y \Sigma_R$. }
\end{figure}
The main goal of this section is to prove the following theorem.

\begin{theorem}\label{thm:ASgluing}
Let $\Sigma = \Sigma_{{L}} \cup_Y \Sigma_{{R}}$ as above  and $\tau_L,\tau_R$ be tadpole functions on $\Sigma_L$,$\Sigma_R$ respectively.
Then\footnote{The slightly unwieldy notation is due to the fact that when gluing only over a part of the boundary, there are in general three different ``blocks'' of the Dirichlet-to-Neumann operator that play different roles in the gluing: The block supported on the gluing interface is formally absorbed in the Gaussian measure, and therefore deleted from the partition function, the off-diagonal block participates in the gluing as in Figure \ref{fig:gluingrelation2}, and the block supported on the remaining boundary component does not interact with gluing  gets corrected into the Dirichlet-to-Neumann operator of the glued bulk (as in the second part of Figure \ref{fig:gluingrelation2}). }
\begin{equation}
e^{-S_0(\phi^{\Sigma_L}_{{\tilde{\eta}}_L})} 
e^{-S_0(\phi^{\Sigma_R}_{{\tilde{\eta}}_R})}
 \langle \widehat{Z^{\tau_L}_{\Sigma_L}},\widehat{Z^{\tau_R}_{\Sigma_R}}\rangle_{\Sigma_L,Y,\Sigma_R} = Z^{\tau_L * \tau_R}_{\Sigma}.
\end{equation}
\end{theorem}
  
The following is an immediate corollary.
\begin{corollary}
If $\tau$ is an assignment of local tadpole functions, we have 
\begin{equation}
e^{-S_0(\phi^{\Sigma_L}_{{\tilde{\eta}}_L})}
e^{-S_0(\phi^{\Sigma_R}_{{\tilde{\eta}}_R})}
\langle \widehat{Z^{\tau}_{\Sigma_L}},\widehat{Z^{\tau}_{\Sigma_R}}\rangle_{\Sigma_L,Y,\Sigma_R} = Z^{\tau}_{\Sigma}.
\end{equation}
In particular, this holds for the %zeta-regularized 
perturbative partition function defined using the zeta-regularized tadpole. %$Z^\zeta$.
\end{corollary}

To prove this theorem we introduce some auxiliary structures on Feynman diagrams.

\subsection{More on Feynman diagrams}
{  This strategy of the proof was inspired by Johnson-Freyd's paper \cite{JF}.}
\begin{definition}[Decorated Feynman graphs]\leavevmode
\begin{enumerate}[i)]
\item A \emph{decoration} of a Feynman graph is a pair of functions\footnote{``u'' is for ``uncut'', ``c'' is for ``cut''.}  $$f=(dec_V\colon V(\Gamma) \to \{L,R\},dec_E \colon E(\Gamma) \to \{u,c\}).$$
\item A decoration is \emph{admissible} if $f(V_X) \subset \{X\}, X=L,R$, and all edges between a vertex decorated $L$ and a vertex decorated $R$ are decorated by $c$. 
\item A \emph{decorated Feynman graph} is a pair $(\Gamma,f)$ of a Feynman graph $\Gamma$ and a decoration $f$ of $\Gamma$. 
\end{enumerate} 
\end{definition}
The automorphism group $\mathrm{Aut}(\Gamma)$ acts on the set of decorations. Two decorations of $\Gamma$ that are related by an automorphism of $\Gamma$ are called \emph{isomorphic.} The set of isomorphism classes of decorations is denoted $dec(\Gamma)$. 
\begin{definition}
An \emph{automorphism} of a decorated Feynman graph is an automorphism $\varphi \in \mathrm{Aut}(\Gamma)$ that fixes the decoration: $dec_V(V_I(\varphi)(v)) = dec_V(v)$, $dec_E(E(\varphi)(e)) = dec_E(e)$, i.e. the decorated automorphism group is the stabilizer of the decoration under the action of $\Aut(\Gamma)$ on the set of decorations. We denote the set of automorphisms of a decorated graph by $\Aut^{dec}(\Gamma)$.
\end{definition}
Notice that all edge types - $E_0,E_1,E_2$ - can be cut. 
We introduce a set of Feynman rules for decorated graphs. 
{  
\begin{definition}
Let $(\Gamma,f)$ be a decorated Feynman graph and let $\Sigma = \Sigma_L \cup_Y \Sigma_R$. Also, let $\tau_L,\tau_R$ be tadpole functions on $\Sigma_L,\Sigma_R$. Then we define the weight $F^{dec,\tau_L,\tau_R}_\Sigma$ of the decorated Feynman graph as follows: 
For a bulk vertex labeled $X \in \{L,R\}$, we integrate over $\Sigma_X$. Edges decorated by $u$ between different $X$ vertices are assigned $G_{\Sigma_{X}}$ or its appropriate derivatives, or the tadpole function $\tau_X$. Edges labeled by $c$ are assigned the second term in the gluing formula for the appropriate derivative of the Green's function. Tadpoles labeled by $c$ are assigned the second term in gluing formula for tadpoles. 
\end{definition}
}
\begin{lemma} 
For all graphs $\Gamma$, we have 
\begin{equation}
\frac{F^{\tau_L*\tau_R}_\Sigma(\Gamma)}{|\Aut(\Gamma)|} = \sum_{f \in dec(\Gamma)} \frac{F_\Sigma^{dec,\tau_L,\tau_R}(\Gamma^{f})}{|\Aut^{dec}(\Gamma^{f})|}.\label{eq:decor}
\end{equation}
\end{lemma}

\begin{proof}
Decompose $F_\Sigma(\Gamma)$ using $\int_\Sigma = \int_{\Sigma_L} + \int_{\Sigma_R}$ and the gluing formula for the propagator $G_\Sigma = G_u + G_c$ between $L$ and $L$ (resp. $R$ and $R$) vertices and $G_\Sigma = G_c$ between $L$ and $R$ vertices. Every term in the resulting sum is labeled by a decoration $f$ of $\Gamma$.  Isomorphic decorations will evaluate to the same weight.  Thus, we obtain  
$$\frac{F^{\tau_L*\tau_R}_\Sigma(\Gamma)}{|\Aut(\Gamma)|} = \sum_{f \in dec(\Gamma)} F_\Sigma^{dec,\tau_L,\tau_R}(\Gamma^{f})\frac{|\Aut(\Gamma)\cdot f|}{|\Aut(\Gamma)|}.$$
By the orbit-stabilizer theorem, we obtain 
$$\frac{|\Aut(\Gamma)\cdot f|}{|\Aut(\Gamma)|} = \frac{1}{|\Aut^{dec}(\Gamma^f)|}$$ 
and the claim follows. 
\end{proof}

We define a gluing operation $*$ on Feynman graphs: Denote Feynman graphs with no $R - R$ edges by $\Gr_R$ and Feynman graphs with no $L - L$ edges by $\Gr_L$. For $\Gamma_L \in \Gr_R, \Gamma_R \in \Gr_L$ we define
$$\Gamma_L * \Gamma_R := \sum_{\sigma \text{ perfect matching of } V_{R}(\Gamma_L) \sqcup V_{L}(\Gamma_R)} \Gamma^{dec}(\sigma,\Gamma_L,\Gamma_
R),$$
where $\Gamma^{dec}(\sigma,\Gamma_L,\Gamma_R)$ is the decorated graph obtained by decorating vertices in $\Gamma_{X}$ with $X$, edges in $\Gamma_{X}$ by $u$ (for ``uncut'') and connecting the boundary vertices specified by $\sigma$ to an edge, decorated $c$ (for ``cut''), between the bulk vertices attached to these boundary vertices. In the language of Definition \ref{def:FeynmanGraph} we set $V_L(\Gamma^{dec}) = V_L(\Gamma_L)$,$V_R(\Gamma^{dec}) = V_R(\Gamma_R)$, $V_b(\Gamma^{dec}) = V_b(\Gamma_L) \sqcup V_b(\Gamma_R)$. The set of half-edges is the union of all half-edges incident to these vertices. The map $\tau$ specifying the edges is extended by the perfect matching $\sigma$. These new edges are decorated $c$, all other edges are decorated $u$, the vertices carry the obvious decorations. See Figure \ref{fig:graphgluing}. 
\begin{figure}[h]
\centering
\begin{tikzpicture}
\coordinate[label=above:$\Gamma_L$] (a) at (0,1.5);
\node[bulk] (b1) at (0,0){};
\node[bdry] at (1,1) {} edge[bubo] (b1);
\node[bdry] at (1,-1){} edge[bubo] (b1);
\draw[densely dotted] (1,-2) -- (1,2);
\coordinate[label=right:{\huge $*$}] (c) at (1.2,0);
\draw[densely dotted] (2,-2) -- (2,2);
\coordinate[label=above:$\Gamma_R$] (a) at (3,1.5);
\node[bulk] (b2) at (3,1){};
\node[bulk] (b3) at (3,-1) {} edge[bubu, bend left] (b2) edge[bubu,bend right] (b2);
\node[bdry] at (2,1){} edge[bubo] (b2);
\node[bdry] at (2,-1){} edge[bubo] (b3);
\coordinate[label=right:{\large $=$}] (c) at (3.5,0);
\begin{scope}[shift={(5,0)}]
\coordinate[label=above:$\Gamma_1$] (a) at (1.5,1.5);
\node[bulk,label=above:$L$] (b1) at (0,0){} edge[tadpole,"c"] (b1);
\node[bulk,label=above:$R$] (b2) at (3,1){};
\node[bulk,label=below:$R$] (b3) at (3,-1) {}edge[bubu, bend left,"u"] (b2) edge[bubu,bend right,"u"] (b2) edge[bubu,bend left=120,looseness=1.5,"c"] (b2);
\coordinate[label=right:{\large $+2$}] (c) at (3.5,0);
\end{scope}
\begin{scope}[shift={(10,0)}]
\coordinate[label=above:$\Gamma_2$] (a) at (1.5,1.5);
\node[bulk,label=above:$L$] (b1) at (0,0){};
\node[bulk,label=above:$R$] (b2) at (3,1){} edge[bubu,"c"] (b1);
\node[bulk,label=below:$R$] (b3) at (3,-1) {} edge[bubu,"c"] (b1) edge[bubu, bend left, "u"] (b2) edge[bubu,bend right,"u"] (b2);
\end{scope}
\end{tikzpicture}

\caption{The gluing operation on graphs. The first term corresponds to the perfect matching which matches vertices on either side, the second term to the two matchings identifying the vertices on different sides. }\label{fig:graphgluing}
\end{figure}
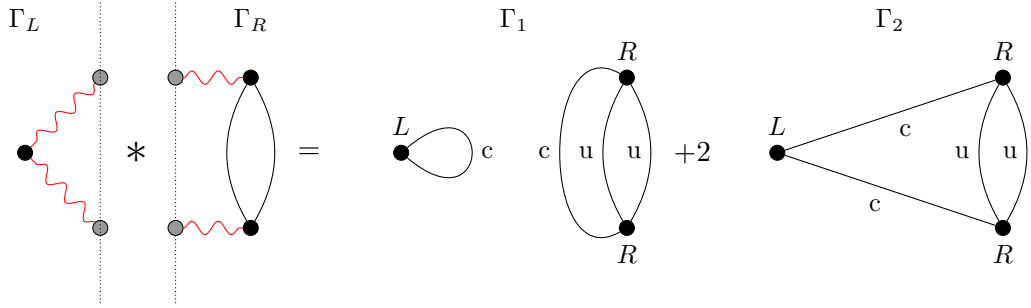 
\begin{remark} We can glue graphs with different amounts of boundary vertices. For instance, the graph $\Gamma_L$ in Figure \ref{fig:graphgluing} could be glued to the empty graph on the right hand side, and these terms are important for the gluing formula for the Green's (or tadpole) function. 
\end{remark}
The gluing operation lands in formal linear combinations of decorated graphs.
$$ \Gamma_L * \Gamma_R = \sum_{\Gamma^{dec} \in \mathrm{supp}(\Gamma_L * \Gamma_R)}m^{\Gamma^{dec}}_{\Gamma_L,\Gamma_R}\Gamma^{dec}.$$ For instance, in the example of Figure \ref{fig:graphgluing}, we have $m^{\Gamma_{1}}_{\Gamma_L,\Gamma_R} = 1$ and $m^{\Gamma_{2}}_{\Gamma_L,\Gamma_R}=2$. Then:
\begin{lemma}\label{lem:unique}
Let $\Gamma^{dec}$ be a decorated Feynman graph. Then there is a unique $  \Gamma_L\in \Gr_R,\Gamma_R\in \Gr_L$ such that $\Gamma^{dec} \in \mathrm{supp}(\Gamma_L * \Gamma_R)$.
\end{lemma}

\begin{proof}
Delete every  $c$-decorated edge $e=\{v_1,v_2\}$ in $\Gamma^{dec}$ and replace it with  vertices $v_1',v_2'$ and edges $e_1={v_1,v_1'},e_2=\{v_2,v_2'\}$. This results in two disconnected graphs $\Gamma_L$ and $\Gamma_R$ containing vertices decorated $L$ and $R$ respectively (they are disconnected since edges between $L$ and $R$ vertices are decorated $c$). The newly added vertices are declared right resp. left boundary vertices if connected to a vertex decorated $L$ resp. $R$. Forgetting the decorations, this is the unique combination of graphs that will contain $\Gamma^{dec}$ in its support after gluing. 
\end{proof}
Notice that the pairing $\left\langle  F^{\tau_{L}}_{\Sigma_L}(\Gamma_L),F^{\tau_R}_{\Sigma_R}(\Gamma_R)\right\rangle_{\Sigma_L,Y,\Sigma_R}$ makes sense also if the graphs $\Gamma_L$ and $\Gamma_R$ have $L - L$ (resp. $R - R$ edges).
\begin{lemma} Using notation as above, for $\Gamma_L \in \Gr_R$ and $\Gamma_R \in \Gr_L$ we have
\begin{equation}\left\langle  F^{\tau_{L}}_{\Sigma_L}(\Gamma_L),F^{\tau_R}_{\Sigma_R}(\Gamma_R)\right\rangle_{\Sigma_L,Y,\Sigma_R}= \frac{F_{\Sigma}^{dec,\tau_L,\tau_R}(\Gamma_L * \Gamma_R)}{\det(D_{\Sigma_L,\Sigma_R})^\frac12}\label{eq:int_pert_vs_gluing}
\end{equation}
Here we extend $F_\Sigma^{dec}$ linearly to formal linear combinations of decorated graphs. 
\end{lemma}
\begin{proof}
The only nontrivial point here is that the integral kernel we used to define the pairing is the same kernel as the one appearing in the gluing formula for the Green's function. Apart from this fact the proof is a matter of plugging in the definitions. 
\end{proof}
Finally we require the following combinatorial lemma.
\begin{lemma}
Using notation as above
\begin{equation}\frac{\Gamma_L}{|\Aut(\Gamma_L)|}*\frac{\Gamma_R}{|\Aut(\Gamma_R)|} = \sum_{\Gamma^{dec} \in \mathrm{supp}\Gamma_L * \Gamma_R} \frac{\Gamma^{dec}}{|\Aut^{dec}(\Gamma^{dec})|}.\label{eq:auto_gluing}\end{equation}
\end{lemma}
\begin{proof}
This is essentially a consequence of the orbit-stabilizer theorem. Notice we can rewrite \eqref{eq:auto_gluing} as 
\begin{equation}
m^{\Gamma^{dec}}_{\Gamma_L,\Gamma_R} \stackrel{!}{=} \frac{|\Aut(\Gamma_L)||\Aut(\Gamma_R)|}{|\Aut^{dec}(\Gamma^{dec})|}\label{eq:orbitstabilizer}
\end{equation}
which is already suggestive of the group action we want to consider. Namely, the group $\Aut \Gamma_L \times \Aut \Gamma_R$ acts on the set $\mathfrak{m}$ of perfect matchings of $V_R(\Gamma_L) \sqcup V_L(\Gamma_R)$. Two perfect matchings related by this group action will define the same decorated graph, so that for any $m^{\Gamma^{dec}_{\Gamma_L,\Gamma_R}} = |(\Aut \Gamma_L \times \Aut \Gamma_R)\cdot \sigma |$ for any $\sigma \in \mathfrak{m}$ that defines $\Gamma^{dec}$. The main realization is that the stabilizer group of $\sigma$ is isomorphic to the automorphism group of the decorated graph $\Gamma^{dec}$: The condition that a pair $(\varphi_L,\varphi_R)$ stabilizes $\sigma$ is equivalent to asking that $\varphi_L \sqcup \varphi_R$ preserves incidence of cut edges. Equation \eqref{eq:orbitstabilizer} is then just the orbit-stabilizer theorem. 
\end{proof}
\subsection{Proof of the gluing formula}
We now have all the necessary ingredients to prove Theorem \ref{thm:ASgluing} (we suppress dependence on arguments after the first line). 
\begin{proof}[Proof of Theorem \ref{thm:ASgluing}]
Recall that we denote by $\Gr_L$ (resp. $\Gr_R$) the set of all Feynman graphs containing \emph{no} edges between left (resp. right) boundary vertices. Then, as noted in Remark \ref{rem:Zfactors}, we have %\marginpar{RECTIFY}
\begin{equation*}
\sum_{\Gamma} \frac{F_\Sigma^\tau(\Gamma)}{|\mathrm{Aut(\Gamma)}|} = \sum_{\Gamma \in \Gr_R} e^{-S_0(\phi^\Sigma_{{\tilde{\eta}}_R})}\frac{F_\Sigma^\tau(\Gamma)}{|\mathrm{Aut(\Gamma)}|}
\end{equation*}
In particular, for a surface $\Sigma$ with $\partial\Sigma = \partial_L\Sigma \cup \partial_R\Sigma$, we have 
\begin{equation}
e^{-S_0(\phi^\Sigma_{{\tilde{\eta}}_L})}
\widehat{Z^\tau_\Sigma} = \frac{1}{\det(\Delta_{\Sigma_L}+m^2)^{\frac12}}\sum_{\Gamma \in  \Gr_R}\frac{F^\tau_\Sigma(\Gamma)}{|\mathrm{Aut}(\Gamma)|}\label{eq:lasteq}.
\end{equation}
The proof of the gluing formula is now a simple consequence of our previous work. 
Consider again a cobordism $(\Sigma,\partial_L\Sigma,\partial_R\Sigma)$ with a decomposition $\Sigma = \Sigma_L \cup_Y \Sigma_R$, 
where $(\Sigma_L,\partial_L\Sigma,Y), (\Sigma_R,Y,\partial_R\Sigma_R)$ are two cobordisms with common boundary component $Y$. Then, we have 
%\marginpar{changed $1/\mr{Aut}(\cdots)$ to $1/|\mr{Aut}(\cdots)|$ everywhere in this computation. PM}
\begin{align*}
&e^{-S_0(\phi^{\Sigma_L}_{{\tilde{\eta}}_L})}
e^{-S_0(\phi^{\Sigma_R}_{{\tilde{\eta}}_R})}
\left\langle  \widehat{Z^{\tau_L}_{\Sigma_L}}({\tilde{\eta}}_L,{\tilde{\eta}}),\widehat{Z^{\tau_R}_{\Sigma_R}}({\tilde{\eta}},{\tilde{\eta}}_R)
\right\rangle_{\Sigma_L,Y,\Sigma_R} = \\
 %\textsf{(By def.)} &=\frac{1}{\det(\de_{\Sigma_L}+m^2)^\frac12\det(\Sigma_R + m^2)^\frac12}\sum_{\Gamma_L,\Gamma_R}\left\langle\frac{F^{\tau_L}(\Gamma_L)}{|\Aut(\Gamma_L)|},\frac{F^{\tau_R}(\Gamma_R)}{|\Aut(\Gamma_R)|}\right\rangle_{\Sigma_L,Y,\Sigma_R} \\
  \textsf{(Eq. \eqref{eq:lasteq})} &= \frac{1}{\det(\Sigma_L+m^2)^\frac12\det(\Sigma_R + m^2)^\frac12}\sum_{\substack{\Gamma_L \in \Gr_R \\ \Gamma_R \in \Gr_L}}\left\langle \frac{F^{\tau_L}_{\Sigma_L}(\Gamma_L)}{|\Aut(\Gamma_L)|},\frac{F^{\tau_R}_{\Sigma_R}(\Gamma_R)}{|\Aut(\Gamma_R)|}\right\rangle_{\Sigma_L,Y,\Sigma_R} \\ 
\textsf{(Eq. \eqref{eq:int_pert_vs_gluing})} &= \frac{1}{\det(\Sigma_L+m^2)^{\frac{1}{2}}\det(\Sigma_R+m^2)^\frac12\det(D_{\Sigma_L,\Sigma_R})^\frac12}\sum_{\substack{\Gamma_L \in \Gr_R \\ \Gamma_R \in \Gr_L}}\frac{F_\Sigma^{dec,\tau_L,\tau_R}(\Gamma_L*\Gamma_R)}{|\Aut(\Gamma_L)|\cdot | \Aut(\Gamma_R)|} \\
\textsf{(Eqs. \eqref{eq:gluing_zeta_determinants},\eqref{eq:auto_gluing})}&= \frac{1}{\det(\Delta_{\Sigma}+m^2)^\frac12}\sum_{\substack{\Gamma_L \in \Gr_R \\ \Gamma_R \in \Gr_L}}\sum_{\Gamma^{dec} \in \Gamma_L*\Gamma_R}\frac{ F_\Sigma^{dec,\tau_L,\tau_R}(\Gamma^{dec})}{|\Aut^{dec}(\Gamma^{dec})|} \\
\textsf{(Lemma \ref{lem:unique})}&= \frac{1}{\det(\Delta_{\Sigma}+m^2)^\frac12}\sum_{\Gamma^{dec}} \frac{F_\Sigma^{dec,\tau_L,\tau_R}(\Gamma^{dec})}{|\Aut^{dec}(\Gamma^{dec})|} \\ 
\textsf{(Eq. \eqref{eq:decor})}&= \frac{1}{\det(\Delta_{\Sigma}+m^2)^\frac12}\sum_{\Gamma}\frac{F^{\tau_L*\tau_R}_\Sigma(\Gamma)}{|\Aut(\Gamma)|} = Z^{\tau_L*\tau_R}_{\Sigma}.
\end{align*}
\end{proof}

\section{Functoriality}\label{sec:functor} 
Returning to the discussion of the introduction, it is a natural question whether the assignment $Y\mapsto H_Y, \Sigma \mapsto Z_\Sigma$ has an interpretation as a functor. It turns out that the answer to this question is positive, provided source and target category are adequately defined, and one introduces the correct mathematical setup. The main problem is that in the treatment of Section \ref{sec:pert_quant} the pairing on the space of boundary states depends on the bulk. We'll briefly describe the idea how to remedy this. Remember that heuristically we want the pairing to be integration against the ``Lebesgue measure'' on the space of boundary fields: 
$$\langle\Psi_1,\Psi_2\rangle = \int_{{\tilde{\eta}}_Y}\Psi_1({\tilde{\eta}}_L,{\tilde{\eta}}_Y)\Psi_2({\tilde{\eta}}_Y,{\tilde{\eta}}_R).$$ 
Of course this formal Lebesgue measure does not depend on the bulk, but 
%there is no nice adequate 
%reasonable
%way for us to define it, 
it is not mathematically well-defined, 
so another idea is needed. From the point of view of perturbation theory, it was natural to use the factor $e^{-S_0(\phi_{{\tilde{\eta}}_Y})}$ to define a formal measure with respect to which we were defining the pairing, but this factor depends on the bulk, since $S_0(\phi_{{\tilde{\eta}}_Y}) = \frac12\int_Y {\tilde{\eta}} D_\Sigma {\tilde{\eta}} \,\dvol_Y$. The trick to obtain a functor is to realize that dependence on the bulk is ``small'' in the sense that $D_\Sigma = \sqrt{\Delta_Y + m^2} + S$, where $S$ is a compact
 operator\footnote{In fact, $\delta=(\Delta_Y+m^2)^{-\frac12} S$ is a 
 trace-class 
%Hilbert-Schmidt
 operator, see Proposition \ref{prop:delta regularity}. Moreover, if the metric $\Sigma$ in the neighborhood of $Y$ is the product metric, then $S$ is  a smoothing operator. We refer to Appendix \ref{App: examples delta} for examples.
 } 
 that contains all the bulk dependence. We can thus use the operator $\sqrt{\Delta_Y + m^2}$ to define an \emph{actual Gaussian measure} on a completion of $C^\infty(Y)$. This Gaussian measure will then be corrected to the one induced by the Dirichlet-to-Neumann operator by a ``small'' contribution from the bulk. Thus, we will multiply the partition function by a factor to obtain its correct normalization. But on a heuristic level, nothing happens at all: we are merely splitting the Gaussian factor in the partition function in a different way. In this section we will spell out the details of this idea and prove that this is enough to make the partition function functorial.
\subsection{The source category}  The source category is the semicategory  (i.e. category without identity morphisms) $\mathbf{Riem}^2$ of 2-dimensional Riemannian cobordisms defined as follows:
\begin{itemize}
\item  Objects are closed Riemannian 1-manifolds with two-sided collars $Y \times (-\epsilon,\epsilon)$ with an arbitrary metric restricting to the metric on $Y$ on $Y\times \{0\}$.
\item A morphism from $Y_L\times (-\epsilon_L,\epsilon_L)$ to $Y_R \times (-\epsilon_R,\epsilon_R)$ is a Riemannian cobordism $\Sigma$ with $\partial_X\Sigma = Y_X$ such that $\partial_L\Sigma$ has a collar (tubular neighborhood) in $\Sigma$ isometric to $Y_L \times [0,\epsilon_L)$ and $\partial_R\Sigma$ has a collar isometric to $Y_R \times (-\epsilon_R,0]$.
\end{itemize}
Composition of morphisms is well-defined since the 2-sided collars ensure that metrics can be glued smoothly. 

We refer the reader to \cite{HST,Stolz} for a detailed discussion of the Riemannian cobordism category.

\subsection{The space of boundary states revisited and the target category}
 
We turn to describing the space of states associated to a closed $1$-manifold.  
%We will describe two pictures for the space of states. 
%The first one is simply
It is constructed as
as in Section \ref{sec:pert_quant}, but with Dirichlet-to-Neumann operator replaced by the square root of the Helmholtz operator on the boundary. 

In order to prove functoriality of the appropriately adjusted partition function, we also introduce the measure-theoretic formulation of the space of states and compare it with the conventional Fock space formulation.
%The second one  is the Fock space construction usual in constructive quantum field theory (see e.g.\cite{GJ, SimonB}). In fact, these two pictures agree, see Proposition \ref{prop:Normal ordering} below. 
\subsubsection{Space of boundary states - perturbative picture}
Consider again the vector spaces $H^{(n)}_Y$ of Definition \ref{def:bdry_space}: functionals $\Psi$ on $C^{\infty}(Y)$ of the form 
$$\Psi(\til{\eta}) = \int_{C^\circ_n(Y)} \psi(y_1,\ldots,y_n)\tilde{\eta}(y_1)\cdots\til{\eta}(y_n)\, dy_1 \ldots dy_n$$ 
where $\psi \in \funlog(C^\circ_n(Y))[[\hbar^{1/2}]]$ is symmetric and has admissible singularities (Definition \ref{def: adm singularities}) on diagonals, and set $H^\mr{pre}_Y = \bigoplus_{n \geq  0}H^{(n)}_Y$. Then we define a new pairing on $H^\mr{pre}$, similar in form to the one of Definition \ref{def:pairing}, but with the Dirichlet-to-Neumann operator replaced by (twice) the square root of the Helmholtz operator on $Y$. Notice that this pairing will be intrinsic to the metric on $Y$, in particular, it can be defined without reference to any bulk manifold. 
%To apply the methods of constructive quantum field theory (see e.g.\cite{GJ, SimonB}), 
We define the space of states as the completion of $H^\mr{pre}$ with respect to that pairing. 
\begin{definition}\label{def:pairing 2 kappa}
Let $g$ be a Riemannian metric on $Y$ and $m>0$ and consider the Helmholtz operator $\Delta_g + m^2$ on $Y$. Denote $\varkappa$ the square root of this operator. 
%\begin{enumerate} 
%\item 
We define the pairing 
\begin{align}\label{<>_kappa def}
\langle\cdot,\cdot\rangle_{2\varkappa}\colon H^{(n)}_Y \times H^{(m)}_Y &\to \R[[\hbar^{1/2}]] \notag \\
\langle \Psi, \Psi'\rangle_{2\varkappa} &= \langle \Psi \odot \Psi' \rangle_{2\varkappa} 
\end{align}
where $\odot$ is the symmetric tensor product defined in \eqref{eq:def symm tens} and $\langle\cdot\rangle_{2\varkappa}\colon H^{(k)} \to \R[[\hbar^{1/2}]]$ the $2\varkappa$-expectation value map 
\begin{equation}
\left\langle \Psi\right\rangle_{2\varkappa}  = \dfrac{1}{(\det2\varkappa)^{\frac12}}\sum_{\mathfrak{m}\in \mathfrak{M}_n}\int_{C_n^\circ(Y)}\psi(y_1,\ldots,y_n)\prod_{\{v_1,v_2\}\in\mathfrak{m}}
(2\varkappa)^{-1}(y_{v_1},y_{v_2})\; dy_1\cdots dy_n.\label{eq:exp_value_bdry_kappa}
\end{equation}
We extend the pairing bilinearly to $H^\pre_Y = \bigoplus_{n\geq 0} H^{(n)}_Y$. 
%\item We denote $\overline{H}(S^1)$ the Hilbert space completion of $H(S^1)$ with respect to the pairing $\langle\cdot,\cdot\rangle_{2\varkappa}$. 
%\item We set $\overline{H}(S^1) := \overline{H}(S^1) \otimes \R[[\hbar^{1/2}]]$. 
%\end{enumerate}
\end{definition}
The completion of $H^\pre_Y$ with respect to the pairing $\langle\cdot,\cdot \rangle_{2\varkappa}$ coincides (as a topological vector space) with $H_Y$, see Definition \ref{def: H completion DN}, since operators $2\varkappa$ and $D_{\Sigma_L,\Sigma_R}$ are sufficiently close.

%Completely analogously, one defines $\overline{H}(Y)$ and $\overline{H}(Y)$ for any closed 1-dimensional Riemannian manifold $Y = S^1 \sqcup \ldots \sqcup S^1$ and one has 
%\begin{equation}
%\overline{H}(Y) = \overline{H}(S^1) \otimes \cdots \otimes \overline{H}(S^1)
%\end{equation}
%where the tensor product is the Hilbert-Schmidt tensor product.
 Notice that $H_Y$ has the decomposition 
\begin{equation*}
H_Y = \bigoplus \overline{H}^{(n)}_Y
\end{equation*} 
where $\overline{H}^{(n)}_Y$ is understood as $H^{(n)}_Y$ 
completed with respect to the restriction of $\langle\cdot,\cdot \rangle_{2\kappa}$ degree-wise in $\hbar$. %and tensored with $\RR[[\hbar^{1/2}]]$. 
However, this is \emph{not} a decomposition into orthogonal subspaces.
%$$\colon\psi\colon = \psi(y_1 ,\ldots, y_n) +C_1\int \psi(y_1,\ldots,y_n)(2\kappa)^{-1}(y_1,y_2) dy_1dy_2 +  C_2\int \psi(y_1,\ldots,y_n) $$
%\textcolor{magenta}{ \marginpar{Redundant?}It will be given by %formal power series in 
%a Fock space tensored with $\RR[\hbar^{1/2}]$.  The space associated to a disjoint union of circles will be given by an appropriate tensor product of such spaces. Our construction is inspired by the constructive quantum field theory \cite{GJ, SimonB}, where it is shown that the Hilbert space associated to any two dimensional scalar field theory with polynomial potential or higher dimensional free scalar field theory can be constructed by using a pseudo-differential operator on the time zero hypersurface. For our consideration, a closed one-dimensional Riemannian manifold will play the role of time zero hypersurface. First, we have the following definition for the single-particle Hilbert space associated to a circle.}  \\

\subsubsection{Space of boundary states - Fock space picture}
We now want to compare our model of the space of states to the more classical notion of Fock space. This is a space naturally associated to a pseudodifferential operator on $C^\infty(Y)$ defined as follows.
%\begin{definition}
%Let $g$ be a Riemannian metric on $S^1$ and $m>0$ and consider the Helmholtz operator $\Delta_g + m^2$ on $S^1$. Denote $\varkappa$ the square root of this operator. Then, we define the \emph{single-particle Hilbert space} {\color{cyan}$\overline{H^{(1)}}(S^1)$  \marginpar{I changed notation here, is it ok?}} as the completion of $C^{\infty}(S^1)$ with respect to the pairing
%\begin{equation}\label{eqon:pairing single-particle}
%\langle f,g\rangle_{2\varkappa} = \frac12\int_{S^1\times S^1}f(x)\, \varkappa^{-1}(x,y)\, g(y)dxdy.
%\end{equation}
%% We suppress the dependence on the metric from the definition.
%\end{definition}
%%\begin{remark} Note that an operator $O\colon H_{S^1}\ra H_{S^1}$ is Hilbert-Schmidt with respect to the pairing (\ref{eqon:pairing single-particle}) if and only if it is trace class.
%%\end{remark}
%\textcolor{cyan}{Notice that this space is the completion of $H^{(1)}(S^1)$ with respect to the pairing of definition \ref{def:pairing 2 kappa} above, since $C^\infty(S^1) = H^{(1)}(S^1)$. }
%\begin{remark}
%Similarly, we can define the single-particle Hilbert space $H_Y$ of any 1-dimensional closed Riemannian manifold and we have 
%\begin{equation}
%H_{S^1 \sqcup \ldots \sqcup S^1} \cong H_{S^1} \otimes \cdots \otimes H_{S^1},
%\end{equation}
%where the tensor product is the Hilbert-Schmidt tensor product of Hilbert spaces.
%\end{remark}
%
\begin{definition}
Denote $V_{2\varkappa}$ the completion of $C^\infty(Y)$ with respect to the pairing\footnote{Notice this is almost the same pairing as $\langle f,g\rangle_{2\varkappa}$ on $H^{(1)}_Y = C^\infty(Y)[[\hbar^{1/2}]]$ but without the prefactor $\dfrac{1}{(\det2\varkappa)^{\frac12}}$.} 
\begin{equation}\label{eqon:pairing single-particle}
\langle f,g\rangle'_{2\varkappa} = \frac12\int_{Y\times Y}f(x)\, \varkappa^{-1}(x,y)\, g(y)dxdy.
\end{equation}
The \emph{Fock space model} of the space of boundary states  associated with $Y$  is 
\begin{equation*} 
F_+({Y}) := \widehat{\bigoplus_{k=0}^\infty} S^kV_{2\varkappa} \;\otimes \RR[[\hbar^{1/2}]]
%F_+(S^1) = \widehat{\bigoplus_{k=0}^\infty} S^k 
%\overline{H^{(1)}}({S^1})
\end{equation*}
where $S^kV$ is the $k$-th symmetric power of the Hilbert space $V$ and $\widehat{\bigoplus}$ denotes the completed orthogonal sum of Hilbert spaces. 
\end{definition}
%Again, for $Y = S^1 \sqcup \ldots \sqcup S^1$ we have $F_+(Y) \cong F_+({S^1}) \otimes \cdots \otimes F_+({S^1})$.
\begin{remark}\label{remark:L2representation of Fock space}
The pairing (\ref{eqon:pairing single-particle}) on $C^{\infty}(Y)$ is the covariance of a Gaussian probability measure $\mu_{2\varkappa}$ on $D'(Y)$--the space of distributions on $Y$. This means that $\widehat{\bigoplus_{k=0}^\infty} S^kV_{2\varkappa}$ is  isomorphic to $L^2(D'(Y),\mu_{2\varkappa})$ via a canonical isomorphism  \cite{V, AG, SimonB}. The isomorphism is constructed by considering the Wiener chaos decomposition \cite{janson} of $L^2(D'(Y),\mu_{2\varkappa})$ obtained from the so-called normal odering procedure. From this, it follows that $F_+({Y})$ is isomorphic to $L^2(D'(Y),\mu_{2\varkappa})[[\hbar^{1/2}]]$. 
\end{remark}
%{
%\color{cyan}
%Notice that in the decomposition of $H(Y) = \bigoplus H^{(n)}(Y)$, the individual components $H^{(n)}(Y)$(``$n$-particle sectors'') are not orthogonal to each other. On the other hand, in the Fock space we have $S^kV \perp S^l V$ for $l \neq k$. Nevertheless, %as Hilbert spaces 
%the completion $H(Y)$ of $H_\pre(Y)$ and $F_+(Y)$ are isomorphic degree-wise in $\hbar$ as Hilbert spaces. 
Notice that in the decomposition of $H_Y = \bigoplus H^{(n)}_Y$, the individual components $H^{(n)}_Y$ (``$n$-particle sectors'') are not orthogonal to each other. On the other hand, in the Fock space we have $S^kV_{2\varkappa} \perp S^l V_{2\varkappa}$ for $l \neq k$. Nevertheless, %as Hilbert spaces 
the completion $H_Y$ of $H^\pre_Y$ can also be decomposed into an orthogonal direct sum. 
%and $F_+(Y)$ are isomorphic degree-wise in $\hbar$. 
This can be done by using the normal ordering, which for example, for\footnote{Here superscript $S_n$ denotes functions invariant under the natural action of $S_n$ on $C_n^\circ(Y)$, i.e. symmetric functions.} $\psi \in \funlog(C^\circ_n(Y))^{S_n}[[\hbar^{1/2}]]$ is defined by: %\marginpar{not sure about $1/2^k$ factors - PM}
\begin{align*}
:\psi: = \psi(y_1,\ldots,y_n) &- {n \choose 2}\int_{Y \times Y}\psi(y_1,\ldots,y_n)(2\kappa)^{-1}(y_1,y_2)\, dy_1dy_2 \\
&+{n \choose 4}\int_{Y^4}\psi(y_1,\ldots,y_n)(2\kappa)^{-1}(y_1,y_2)(2\kappa)^{-1}(y_3,y_4)\, dy_1dy_2dy_3dy_4  \notag\\
& - \ldots \notag
\end{align*}
We then have the following. 
\begin{prop}\label{prop:Normal ordering}
Denote the pairing on the Fock space by $\langle\cdot,\cdot\rangle_{F}$. 
%Then, the map $\psi \mapsto \Phi(\psi)$ extends to a degree-wise (in $\hbar$) isomorphism of Hilbert spaces $$\big(H(Y),\langle\cdot,\cdot\rangle_{2\varkappa}\big) \cong \left(F_+(Y), \frac{1}{\det(2\varkappa)^\frac12}\langle\cdot,\cdot\rangle_F \right).$$
Then, there is a canonical isomorphism %between 
$$\big(H_Y,\langle\cdot,\cdot\rangle_{2\varkappa}\big) \cong \left(F_+(Y), \frac{1}{\det(2\varkappa)^\frac12}\langle\cdot,\cdot\rangle_F \right).$$
\end{prop}
\begin{proof}
%Denote the pairing on Fock space by $\langle\cdot,\cdot\rangle_{F}$.
%The space $\bigoplus\funlog(C^\circ_n(Y))^{S_n}[[\hbar^{1/2}]]$ is dense in both $H(Y)$ and $F_+(Y)$. The claim then follows from the fact that for $\psi \in \funlog(C^\circ_n(Y))^{S_n}$, $\psi' \in \funlog(C^\circ_m(Y))^{S_n}$ we have \marginpar{convert to normal un-ordering}
%\begin{equation}
%\langle \psi, \psi'\rangle_{2\kappa} = \frac{1}{\det(2\varkappa)^\frac12} \langle \Phi(\psi)\;, \Phi(\psi')\rangle_F.
%\end{equation}
First, we observe that, if we ignore the determinant factor in the pairing in Definition \ref{def:pairing 2 kappa}, then there is an obvious isomorphism from $H_Y$ onto  $L^2(D'(S^1),\mu_{2\varkappa})[[\hbar^{1/2}]]$  that respects orthogonal decomposition, where the orthogonal decomposition of $H_Y$ comes from the normal ordering as discussed above and the one on $L^2(D'(S^1),\mu_{2\varkappa})[[\hbar^{1/2}]]$ comes from Wiener chaos decomposition. Now, the proposition follows from  Remark \ref{remark:L2representation of Fock space}.
%\marginpar{This is a very minimalistic attempt at a proof and needs to be checked/extended.}
\end{proof}
%\begin{remark}\label{remark:L2representation of Fock space}
%The pairing \ref{eqon:pairing single-particle} on $C^{\infty}(S^1)$ is the covariance of a Gaussian probability measure $\mu_{2\varkappa}$ on $D'(S^1)$--the space of distributions on $S^1$. This means the Fock space $F_+({S^1})$ is canonically isomorphic as a Hilbert space to $L^2(D'(S^1),\mu_{2\varkappa})$,  see e.g. \cite{AG, SimonB}. 
%%This isomorphism is also an isomorphism of algebras (the algebra structure on Fock space is that of the completed symmetric algebra). 
%In fact the single-particle Hilbert space $H_{S^1}$, which is also called the Cameron-Martin space of $\mu_{2\varkappa}$, contains ``all'' the relevant information about the measure $\mu_{2\varkappa}$. We refer to \cite{V} for details. %This isomorphism is the reason to use the pairing induced by $\kappa$ in the single-particle Hilbert space. The naive choice would have been to use the standard pairing in the single-particle Hilbert space, but for this choice we do not have a corresponding Gaussian measure. 
%\end{remark}
For us it will be convenient to use the following normalization of Gaussian measures. 
\begin{definition}
We define the ``unnormalized'' Gaussian measure corresponding to a symmetric positive operator $A$ to be 
\begin{equation*}
\mu'_A = \frac{\mu_A}{\det(A)^{1/2}}.
\end{equation*}
Here we assume $\det A$ exists in the zeta-regularized sense.
\end{definition}
\begin{remark}\label{rem:HSoperators}
Given a cobordism $(\Sigma,\partial_L\Sigma,\partial_R\Sigma)$ let $Y = \partial \Sigma = \partial_L \Sigma \sqcup \partial_R\Sigma$. Then the associated space of boundary states $H_Y$ satisfies $$H_Y \cong \mathrm{HS}(H_{\partial_L\Sigma},H_{\partial_R\Sigma}),$$ where $\mathrm{HS}$ denotes Hilbert-Schmidt operators (this is a standard property of the tensor product of Hilbert spaces). Given the three compact Riemannian 1-manifolds $Y_L,Y,Y_R$, the pairing
 $H_Y \otimes H_Y \to \R[[\hbar^{1/2}]]$ extends to the composition map $H_{Y_L \sqcup Y} \otimes H_{Y \sqcup Y_R} \to H_{Y_L \sqcup Y_R}$. 
\end{remark}
We now define the target category as follows.
\begin{definition} 
The category $\mathbf{Hilb}^{\mathrm{form}}$ is the category where 
\begin{itemize}
\item objects are %algebras of formal power series over Hilbert spaces with an algebra structure, 
real Hilbert spaces tensored with $\RR[[\hbar^{1/2}]]$,
\item morphisms are %power series of linear maps which are degree-wise Hilbert-Schmidt. 
$\RR[[\hbar^{1/2}]]$-linear Hilbert-Schmidt operators.
\end{itemize}
\end{definition}
\subsection{Proof of functoriality} \label{sec: proof of functoriality, measure-theoretic}

As was explained in the discussion above, we need to slightly adjust the partition function to account for the pairing on the space of boundary states: 
\begin{definition}
Let $\Sigma \equiv(\Sigma,\partial_L\Sigma,\partial_R\Sigma)$ be a cobordism and $\tau$ be a tadpole function on $\Sigma$. For a 1-dimensional manifold $Y$, denote $\varkappa:= \sqrt{\Delta_Y +m^2}$. The \emph{functorial partition function} of $\Sigma$ is 
\begin{equation}
\overline{Z^\tau_\Sigma}[{\tilde{\eta}}] = e^{\frac{1}{2}\int_{\partial\Sigma}{\tilde{\eta}} \varkappa{\tilde{\eta}}\, \dvol_{\dd\Sigma}}Z^\tau_\Sigma[\tilde{\eta}].
\end{equation}
\end{definition} 
\begin{prop}
We have $\overline{Z^\tau_\Sigma} \in \mathbf{Hilb}^{\mathrm{form}}(H_{\dd_L \Sigma},H_{\dd_R\Sigma})$.
\end{prop}
\begin{proof}
By Remark \ref{rem:HSoperators}, it is enough to show that $\overline{Z^\tau_\Sigma} \in H_{\partial \Sigma}$. Notice that 
\begin{align*}
\overline{Z^\tau_\Sigma}[{\tilde{\eta}}] &=
% %e^{\frac{1}{2}\int_{\partial\Sigma}{\tilde{\eta}} \varkappa{\tilde{\eta}}\, \dvol_{\dd\Sigma}}e^{-\frac{1}{2}\int_{\partial\Sigma}{\tilde{\eta}}D_\Sigma {\tilde{\eta}}\, \dvol_{\dd\Sigma}}\widehat{Z^\tau_\Sigma}[\tilde{\eta}]\\
  \frac{1}{\det(\Delta_\Sigma + m^2)^\frac12}e^{-\frac{1}{2}\int_{\partial\Sigma}{\tilde{\eta}} S{\tilde{\eta}}\, \dvol_{\dd\Sigma}}Z^{\mr{pert},\tau}_\Sigma[\tilde{\eta}]
\end{align*}
where $S = D_\Sigma - \varkappa$ and $Z^{\mr{pert},\tau}_\Sigma$ is given by summing over all diagrams with no boundary-boundary edges. Let $\Gamma$ be a Feynman diagram with $n$ boundary vertices and no boundary-boundary edges. By Proposition \ref{prop: Feynman singularities}, we have $F(\Gamma) \in H^{(n)}_{\partial \Sigma}$. Denote $(H_{\partial\Sigma})_{k/2}$ the order $k/2$ part in $\hbar$ of $H_{\partial\Sigma}$. Then we have that $%\overline{F}(\Gamma) = 
e^{-\frac{1}{2}\int_{\partial\Sigma}{\tilde{\eta}} S{\tilde{\eta}}\, \dvol_{\dd\Sigma}}F(\Gamma) \in (H_{\partial\Sigma})_{\ell(\Gamma)}$, because  the operator $\delta = \varkappa^{-1}S$ is trace-class (see Proposition \ref{prop:delta regularity}). Since at any order in $\hbar$ there are only finitely many diagrams with no boundary-boundary edges, we conclude that $\overline{Z^\tau_\Sigma} \in H_{\partial \Sigma}$. 
%Let $\Gr_{(k/2)}$ denote the set of Feynman diagrams contributing to $Z^{\mr{pert},\tau}_\Sigma$ of order $\hbar^{k/2}$, by remark \ref{rem:Zfactors}, this is a finite set. In particular, $(\overline{Z^\tau_\Sigma})_{k/2} = \sum_{\Gamma \in Gr_{(k/2)}}|\Aut(\Gamma)|^{-1}\overline{F}(\Gamma) \in H(\partial \Sigma)$. It follows that 
%$$\overline{Z^\tau_\Sigma} = \frac{1}{\det(\Delta_{\Sigma})^\frac12}\sum_{k \in \mathbb{Z}_{\geq 0}}(\overline{Z^\tau_\Sigma})_{k/2}\hbar^{k/2} \in H(\partial \Sigma)[[\hbar^{1/2}]] = H(\partial\Sigma).$$
%\marginpar{Probably needs to be rewritten again, now}
%This is a direct consequence of the results of Section \ref{sec:Feynmanrules} which imply that the weights of the Feynman graphs are square-integrable with respect to the induced measure on the configuration space.\marginpar{This still works with the new model, right? Should we say something about it?}
%In particular, they are square integrable with respect to measure induced by $2\varkappa$.  Since for a fixed order of $\hbar$ and a given number of boundary points there are only finitely many diagrams, the sum over the contributions of these diagrams is an element of the corresponding symmetric power of the single-particle Hilbert space.  
\end{proof}
The gluing formula can then be re-interpreted as the fact that, if $\tau$ is a local assignment of tadpole functions, then %functorial partition function is a functor.  
partition functions $\overline{Z^\tau}$ assemble into a functor.
\begin{theorem}\label{thm: Zbar functoriality}
Let $\tau_\Sigma$ be a local assignment of tadpole functions. 
The assignment 
$$\overline{Z^\tau} \colon \mathbf{Riem}^2 \to \mathbf{Hilb}^{\mathrm{form}}$$ given on objects by 
$$\overline{Z^\tau}(Y \times (-\epsilon,\epsilon)) = H_Y$$
and on morphisms by 
$$\overline{Z^\tau}(\Sigma) = \overline{Z^\tau_\Sigma}$$ is a functor. 
\end{theorem}
\begin{proof}
Let $\Sigma = \Sigma_L \cup_Y \Sigma_R$ be a Riemannian cobordism. Since $ \mathbf{Riem}^2$ is a semicategory, we actually only have to check the composition rule
$$ \overline{Z^\tau}(\Sigma_L \cup_Y \Sigma_R)= \overline{Z^\tau}(\Sigma_L) \circ \overline{Z^\tau}(\Sigma_R).$$ We denote $Y_L = \dd_L\Sigma = \dd_L\Sigma_L, Y=\dd_R\Sigma_L = \dd_L\Sigma_R, Y_R = \dd_R\Sigma = \dd_R\Sigma_R$. 
To see this, recall that (this is Remark \ref{rem:HSoperators}) composition of morphisms $F_1\colon H_1 \to H_2$, $F_2\colon H_2 \to H_3$ in the category $\mathbf{Hilb}^\mathrm{form}$ is given by extension of the pairing on $H_2$. On the other hand, this pairing is given by integrating against the Gaussian measure $\mu'_{2\varkappa}$ in the $L^2$ space corresponding to $H_2$:
\begin{align}  \overline{Z^\tau}(\Sigma_L) &\circ \overline{Z^\tau}(\Sigma_R) = \int_{D'(Y)}  \overline{Z^\tau}(\Sigma_L) \overline{Z^\tau}(\Sigma_R)d\mu'_{2\varkappa} \notag \\
&= e^{\frac12 \int_{Y_L}\tilde{\eta}\varkappa\tilde{\eta}\,\dvol_{Y_L}}e^{\frac12\int_{Y_R}\tilde{\eta}\varkappa\tilde{\eta}\,\dvol_{Y_R}}\int_{D'(Y)} e^{\int_{Y}\tilde{\eta}\varkappa\tilde{\eta}\,\dvol_Y}Z^\tau_{\Sigma_L}Z^\tau_{\Sigma_R}d\mu'_{2\varkappa} \notag\\
&= e^{\frac12\int_{Y_L}\tilde{\eta}\varkappa\tilde{\eta}\,\dvol_{Y_L}}e^{\frac12\int_{Y_R}\tilde{\eta}\varkappa\tilde{\eta}\,\dvol_{Y_R}}
e^{-S_0(\phi^{\Sigma_L}_{{\tilde{\eta}}_L})}
e^{-S_0(\phi^{\Sigma_R}_{{\tilde{\eta}}_R})} 
\notag \\
&\qquad\int_{D'(Y)} \widehat{Z^\tau_{\Sigma_L}}\widehat{Z^\tau_{\Sigma_R}}e^{-\frac12 \int_{Y}\tilde{\eta}(S_L +S_R)\tilde{\eta}\,\dvol_Y}d\mu'_{2\varkappa} \label{eq:proof functor}
\end{align}
Here $S_X = D_{\Sigma_X} - \varkappa$ for $X \in \{L,R\}$.
 Since $\varkappa^{-1}S_X$ is a 
 %Hilbert-Schmidt 
 trace-class operator by Proposition \ref{prop:delta regularity}, by known results on Gaussian measures (see e.g. \cite[Chapter 7]{AG} or \cite{V}),we have $e^{-\frac12 \int_{Y}\tilde{\eta}(S_L +S_R)\tilde{\eta}\,\dvol_Y}\mu'_{2\varkappa} = \mu'_{D_{\Sigma_L,\Sigma_R}}$ as measures on $D'(S^1)$.\footnote{With our convention for the Gaussian measure, this can be proven similarly to Theorem 7.3 in \cite{AG} by showing the Fourier transforms of the two measures are equal.}
However, integration against the latter Gaussian measure is just the pairing defined in Definition \ref{def:pairing}. So, expression \eqref{eq:proof functor} can be rewritten as 
$$ e^{\frac12\int_{\partial\Sigma}\til{\eta}\varkappa\til{\eta}}
e^{-S_0(\phi^{\Sigma_L}_{{\tilde{\eta}}_L})}
e^{-S_0(\phi^{\Sigma_R}_{{\tilde{\eta}}_R})}
\langle\widehat{Z^\tau_{\Sigma_L}},\widehat{Z^\tau_{\Sigma_R}}\rangle_{\Sigma_L,Y,\Sigma_R}.$$ Hence the gluing formula \ref{thm:ASgluing} implies the composition law for $\overline{Z^\tau}$. 
\end{proof}

\begin{remark}
The Hilbert-Schmidt norm %\marginpar{This Hilbert-Schmidt norm is a formal power series, right? Should we mention that?} 
of a partition function of a Riemannian cobordism $\Sigma$  admits the following interpretation: its square is the partition function of the closed ``doubled surface'' $\widetilde{\Sigma}=\Sigma\cup_{\partial \Sigma}\overline\Sigma$ (assuming the glued metric on $\til\Sigma$ is smooth):
$$ ||\overline{Z^\tau}(\Sigma)||_{HS}^2=Z^{\til\tau}(\til\Sigma)  $$
Here we endow $\til\Sigma$ with the tadpole function $\til\tau=\tau*\tau$ -- the gluing of $\tau$ (some a priory fixed tadpole function) on $\Sigma$ and its reflection on the second copy, $\overline{\Sigma}$.
\end{remark}

\subsection{Another proof of functoriality}
For the reader's convenience, here we give another proof of Theorem \ref{thm: Zbar functoriality}, under an extra assumption on admissible cobordisms. 
It is a direct proof using perturbation theory and not relying on infinite-dimensional measure theory.

\begin{assumption}\label{assump: ||delta|| <1}
For each boundary component $Y$ of the cobordism $\Sigma$, the operator $\delta=\varkappa^{-1} D_\Sigma-1$ on $Y$ has the operator norm 
\begin{equation}\label{||delta|| < 1}
||\delta|| <1
\end{equation} 
(i.e. eigenvalues of $\delta$ are in the interval $(-1,1)$).
\end{assumption}

\begin{remark}
\begin{enumerate}[(a)]
\item \label{assump remark (a)} Cobordisms satisfying Assumption \ref{assump: ||delta|| <1} form a \emph{subcategory} of $\mathbf{Riem}^2$.\footnote{
The reason is that if $\til\Sigma=\Sigma\cup_{Y'} \Sigma'$ and $Y$ is a boundary component of $\Sigma$ disjoint from $Y'$, then we have $D_{\Sigma}>D_{\til\Sigma}>0$ (we mean the $YY$ block of both Dirichlet-to-Neumann operators; an inequality of operators $A>B$ means that $A-B$ is a positive operator). This inequality follows from the gluing formula for Green's functions (Proposition \ref{Proposition:gluing_formula_green}) upon taking the second normal derivative. Therefore, $\delta_{\Sigma}>\delta_{\til\Sigma}>-1$. Thus, if $||\delta_{\Sigma}||<1$, then also $||\delta_{\til\Sigma}||<1$.
}
We denote this subcategory $\mathbf{Riem}^2_{||\delta||<1}$.
\item \label{assump remark (b)} Assumption \ref{assump: ||delta|| <1} is not vacuous. E.g., a cylinder of height $H$, cf. (\ref{lambda/omega examples}), satisfies it iff $H> \frac{c}{m}$ where $c=\mr{arccoth}\,2 \approx 0.5493$. Thus, \emph{short cylinders} fail the assumption. \\
Another example: a spherical sector satisfies the assumption if the cone angle satisfies $\phi<c'$, with $c'\approx 0.9023 \pi$. For $c'<\phi<\pi $, the assumption fails for $mR$ in a certain interval ($R$ is the sphere radius).
\item As implied by (\ref{assump remark (a)}) and (\ref{assump remark (b)}), if we attach to the boundaries of any surface $\Sigma$ sufficiently long cylinders, the resulting surface will satisfy Assumption \ref{assump: ||delta|| <1}.
\end{enumerate}
\end{remark}

Let $\Sigma=\Sigma_L\cup_Y \Sigma_R$ be a Riemannian cobordism cut into two by $Y$. Let $S_i=D_{\Sigma_i}-\varkappa$, with $i\in\{L,R\}$. 
We assume that $\Sigma_L, \Sigma_R
%\in \mathbf{Riem}^2_{||\delta||<1}
$ satisfy Assumption \ref{assump: ||delta|| <1}.

\begin{lemma} \label{lemma: <>_kappa vs <>_D}
\begin{enumerate}[(i)]
\item \label{lemma: <>_kappa vs <>_D  (i)}
For $\Psi\in H_Y$ any state on $Y$, one has the following comparison of expectation values (\ref{eq:exp_value_bdry_kappa}) and (\ref{eq:exp_value_bdry}):
\begin{equation}\label{eq: <>_kappa vs <>_D}
 \big\langle \Psi e^{-\frac12 \int_Y \til\eta (S_L+S_R)\til\eta\,\dvol_Y }\big\rangle_{2\varkappa} = \big\langle \Psi \big\rangle_{\Sigma_L,Y,\Sigma_R} 
\end{equation}
\item \label{lemma: <>_kappa vs <>_D (ii)}
For $\Psi_1,\Psi_2 \in H_Y$ any pair of states on $Y$, one has the following comparison of pairings (\ref{<>_kappa def}) and (\ref{<>_D def}):
\begin{equation}
\langle \Psi_1 e^{-\frac12 \int_Y \til\eta S_L(\til\eta)\,\dvol_Y }\, , \, \Psi_2 e^{-\frac12 \int_Y \til\eta S_R(\til\eta)\,\dvol_Y } \rangle_{2\varkappa}=
\big\langle \Psi_1, \Psi_2 \big\rangle_{\Sigma_L,Y,\Sigma_R}
\end{equation}
\end{enumerate}
\end{lemma}

\begin{proof} For (\ref{lemma: <>_kappa vs <>_D  (i)}), assume that $\Psi$ is given by a wave function $\psi(y_1,\ldots,y_n)$. The l.h.s. of (\ref{eq: <>_kappa vs <>_D}) evaluates to
\begin{multline*}
\frac{1}{\det(2\varkappa)^{1/2}}\sum_{\mathfrak{m}\in \mathfrak{M}_n} \int_{C^\circ_n(Y)}dy_1\cdots dy_n\;\psi(y_1,\ldots,y_n)\cdot \\
\cdot  \prod_{\{i,j\}\in\mathfrak{m}} \Big(
%(2\varkappa)^{-1}-(2\varkappa)^{-1}(S_L+S_R)(2\varkappa)^{-1}+ (2\varkappa)^{-1}(S_L+S_R)(2\varkappa)^{-1}(S_L+S_R)(2\varkappa)^{-1}-\cdots 
\sum_{k=0}^\infty  \big(-(2\varkappa)^{-1}(S_L+S_R)\big)^k (2\varkappa)^{-1}
\Big)(y_i,y_j)
\cdot \exp\Big(\sum_{p=1}^\infty \frac{1}{2p}\tr \big(-(2\varkappa)^{-1}(S_L+S_R)\big)^p \Big)
\end{multline*}
The sum over $k$ -- the ``dressed boundary propagator'' -- evaluates to 
$(2\varkappa+S_L+S_R)^{-1}=K$, the Green's function of $D_{\Sigma_L,\Sigma_R}$. The sum over $p$ evaluates to 
$$-\frac12 \mr{tr}\log (1+(2\varkappa)^{-1}(S_L+S_R))=-\frac12 \log\det (1+(2\varkappa)^{-1}(S_L+S_R))= -\frac12 \log \frac{\det (D_{\Sigma_L,\Sigma_R})}{\det (2\varkappa)}$$ 
Here  $\det(1+\cdots )$ is understood as a Fredholm determinant. 
%Thus, the l.h.s. of (\ref{eq: <>_kappa vs <>_D}) is
%$$  \frac{1}{\Big(\det(2\varkappa)\cdot \det(1+(2\varkappa)^{-1}(S_L+S_R))\Big)^{1/2}}\sum_{\mathfrak{m}\in \mathfrak{M}_n} \int_{C^\circ_n(Y)}dy_1\cdots dy_n\;\psi(y_1,\ldots,y_n)\cdot $$
Therefore, the l.h.s. of (\ref{eq: <>_kappa vs <>_D}) coincides with the r.h.s.

Here the convergence of the sum over $k$ and the sum over $p$, relies on two ``smallness'' properties of the operator  $\delta^\mr{tot}=(2\varkappa)^{-1}(S_L+S_R) = \frac12(\delta_L+\delta_R)$:
\begin{itemize}
\item the trace-class property $\mr{tr}\, \delta^\mr{tot}<\infty$ (we have it by Proposition \ref{prop:delta regularity}) which is needed for the individual terms in the sum over $p$ to be well-defined and
\item the property $||\delta^\mr{tot}||<1$ (implied by Assumption \ref{assump: ||delta|| <1} for $\Sigma_L,\Sigma_R$) needed for the convergence of the sums over $k$ and $p$.
\end{itemize}

Part (\ref{lemma: <>_kappa vs <>_D (ii)}) follows trivially from (\ref{lemma: <>_kappa vs <>_D  (i)}) by setting $\Psi=\Psi_1\odot\Psi_2$.
\end{proof}

\begin{proof}[Proof of Theorem \ref{thm: Zbar functoriality} 
under Assumption \ref{assump: ||delta|| <1}
]
Functoriality of $\overline{Z}$ 
restricted to $\mathbf{Riem}^2_{||\delta||<1}$ 
follows from the gluing formula have already proven (Theorem \ref{thm:ASgluing}) and from Lemma \ref{lemma: <>_kappa vs <>_D}. Indeed, for any $\Sigma$ we have 
$$\overline{Z^\tau_\Sigma} = 
%e^{-\frac12 \int_{\dd_L \Sigma} \til\eta_L S(\til\eta_L)\,\dvol_{\dd_L \Sigma}} e^{-\frac12 \int_{\dd_R \Sigma} \til\eta_R S(\til\eta_R)\,\dvol_{\dd_R \Sigma}} 
e^{-\frac12 \int_{\dd \Sigma} \til\eta S(\til\eta)\,\dvol_{\dd \Sigma}}
\widehat{Z^\tau_\Sigma}$$
Therefore,
\begin{align*}
\big\langle  \overline{Z}_{\Sigma_L} &, \overline{Z}_{\Sigma_R} \big\rangle_{2\varkappa}= e^{-\frac12 \int_{Y_L} \til\eta_L S_L(\til\eta_L)\,\dvol_{Y_L}}e^{-\frac12 \int_{Y_R} \til\eta_R S_R(\til\eta_R)\,\dvol_{Y_R}} \cdot \\
&\cdot \big\langle e^{-\frac12 \int_Y \til\eta S_L(\til\eta)\,\dvol_Y } \widehat{Z}_{\Sigma_L} \, ,\,  e^{-\frac12 \int_Y \til\eta S_R(\til\eta)\,\dvol_Y } \widehat{Z}_{\Sigma_R} \big\rangle_{2\varkappa}  \\
&\underset{\mr{Lemma}\; \ref{lemma: <>_kappa vs <>_D}}{=} 
e^{-\frac12 \int_{Y_L} \til\eta_L S_L(\til\eta_L)\,\dvol_{Y_L}}e^{-\frac12 \int_{Y_R} \til\eta_R S_R(\til\eta_R)\,\dvol_{Y_R}} \big\langle \widehat{Z}_{\Sigma_L} \, ,\, \widehat{Z}_{\Sigma_R}  \big\rangle_{\Sigma_L,Y,\Sigma_R}\\
&\underset{\mr{Theorem}\; \ref{thm:ASgluing}}{=} \overline{Z}_\Sigma
\end{align*}
Here we are suppressing the tadpoles in the notations. Thus, $\overline{Z}$ satisfies the gluing formula with respect to $\langle\cdot ,\cdot\rangle_{2\varkappa}$, which proves functoriality.
\end{proof}

\begin{remark} 
The fact that in the perturbative approach we needed an additional assumption (\ref{||delta|| < 1}) on cobordisms while this assumption was not needed in the measure-theoretic approach of Section \ref{sec: proof of functoriality, measure-theoretic} can be modeled on on the following toy example. Consider a 1-dimensional Gaussian integral
\begin{align*}
\int_{-\infty}^\infty dx\, e^{-(1+C) x^2/2} &= \int dx\, e^{-x^2/2} e^{-C x^2/2} 
=  \int dx\, e^{-x^2/2}\sum_{k=0}^\infty \frac{(-C)^k}{2^k k!} x^{2k} \\
& ``=" \sum_{k=0}^\infty   \frac{(-C)^k}{2^k k!} \int dx\, e^{-x^2/2} x^{2k} 
\underset{\mr{Wick's\ lemma}}{=} \sqrt{2\pi} \sum_{k=0}^\infty   \frac{(-C)^k}{2^k k!} (2k-1)!!\\
& = \sqrt{2\pi} (1-\frac{1}{2} C+ \frac{1}{2!}\frac{1}{2}\frac{3}{2} C^2-\cdots) 
\quad =
%``{ `` = "}"
\quad  \sqrt{2\pi} (1+C)^{-\frac12}
\end{align*}
Here the l.h.s. defined measure-theoretically makes sense for any $C>-1$ whereas the sum after $``="$ is only convergent for $-1<C<1$, i.e. an additional restriction on $C$ arises. The point here is that in the equality $``="$ we are interchanging an integral and a sum which is only valid under this additional restriction on $C$.
\end{remark}

\begin{remark}
One can remove the restrictive Assumption \ref{assump: ||delta|| <1} in the perturbative proof of functoriality by considering the following ``mixed'' picture for the pairing $\langle \cdot ,\cdot \rangle_{2\varkappa}$. One can\footnote{
This strategy is inspired by the proof of Theorem 7.3 in  \cite{AG}.
} 
split functions on $Y$ into the span of eigenfunctions of $\delta$ wth eigenvalues $\geq 1$ and the span of eigenfunctions with eigenvalues in the interval $(-1,1)$:
$$C^\infty(Y)=[C^\infty(Y)]_{||\delta||\geq 1} \oplus [C^\infty(Y)]_{||\delta||< 1}$$
Here the first term on the right is a \emph{finite-dimensional} vector space. Then, one can define the pairing  $\langle \cdot ,\cdot \rangle_{2\varkappa}$ as a combination of a finite-dimensional measure-theoretic Gaussian integral over $[C^\infty(Y)]_{||\delta||\geq 1}$ and a perturbatively defined, via Wick contractions, Gaussian integral over $[C^\infty(Y)]_{||\delta||< 1}$ where one does not have a convergence problem.
\end{remark}

%\subsubsection{Gluing two cylinders into a torus}
%We can also try to glue two cylinders into a torus. 
%\subsection{Squares}
%Finally, we analyze the theory on squares.
\section{Discussion and outlook}
In this paper we have defined the perturbative partition function of two-dimensional scalar field theory as a formal power series, and shown that it satisfies an Atiyah-Segal type gluing relation. In particular, this shows that the perturbatively defined path integral in our model
%retains some properties that one would expect from the path integral.  
satisfies a crucial property expected from the path integral -- a Fubini-type theorem.

To obtain this result we used gluing formulae for the zeta-regularized determinants and the Green's function of the Helmholtz operator, together with some combinatorics of Feynman diagrams. Naturally, one is led to the expectation that similar techniques will allow to prove gluing formulae for other theories. %In particular, it will be interesting to see the interplay with renormalization in theories where it is needed. \\

As explained above, the gluing pairing can be thought of as a mathematical definition of a functional integral over boundary fields. One can think of this functional integral as an expectation value with respect to a non-local boundary theory. %as advocated in
%\textcolor{cyan}{See also \cite{D1},\cite{D2} for a similar perspective.}
A similar perspective was advocated in \cite{D1},\cite{D2}.

Similar results for %topological 
first-order gauge theories in BV formalism have been obtained by Cattaneo, Reshetikhin and the second author
in \cite{CMRQ}. %We have been inspired by their work, but 
We choose a slightly different way to define the partition function on a manifold with boundary.\footnote{We use the (unique) harmonic extension of the boundary field, while \cite{CMRQ} uses a discontinuous extension of the boundary fields,
% (forced by compatibility with the BV symplectic structure), 
i.e. it drops to zero immediately outside the boundary.} However, we expect the two approaches to be ultimately equivalent. We plan to explore this relation in the future. %This will be another direction for future work. 
%
%\\
%\textcolor{cyan}{Finally, let us remark on another interesting interpretation of the gluing formula. Consider for simplicity a closed 2-manifold $\Sigma$ and a decomposition $\Sigma = \Sigma_L \cup_Y \Sigma_R$. Then, as discussed above, there is an induced action functional on $C^\infty(Y)$ given by $S_Y(\eta) = \int_Y \eta D_{\Sigma_L,\Sigma_R} \eta \dvol_g$. This is the action functional of a free classical field theory on $Y$, and the perturbative partition functions $Z_{\Sigma_L},Z_{\Sigma_R}$ can be regarded as observables in this theory. The partition function $Z_\Sigma$ is then the expectation value of the product of theses observables. It would be interesting to investigate whether this phenomenon is related to wider concepts of ``holography'' that are discussed both in the mathematics and physics literature these days. }

We have also proven that the perturbative quantization in our model gives rise to a functor from the category of Riemannian 2-cobordisms to the category of Hilbert spaces and Hilbert-Schmidt operators.

The cutting-gluing formula for partition functions underlying the functoriality result relies on the careful treatment of tadpole diagrams and their interaction with locality.

The following questions naturally arise from our treatment of scalar theory.
\begin{enumerate}[I.]
\item Compare with the treatment in the first order formalism. In particular: is the non-local ``gluing theory'' on the boundary the effective theory for some local gluing theory that arises in the first order formalism? 
%\textcolor{cyan}{Second: 
%how to compare the harmonic extension of the boundary fields into the bulk in the second order formalism with the discontinuous extensions used in
%\cite{CMRQ} in the first order formalism?}
%what is the analogue of the harmonic extension of the boundary fields into the bulk in the first order formalism and is it equivalent to the discontinuous extensions used in \cite{CMRQ}.}
\item The adjusted partition function $\overline{Z}$ entering in the functorial formulation modifies the standard partition function $Z$ (corresponding to quantization with Dirichlet polarization on the boundary) by a factor $e^{\frac{1}{2\hbar}\int_{\partial \Sigma}\dvol_{\partial \Sigma}\eta \varkappa(\eta)}$. It begs an interpretation in terms of a new ``Helmholtz'' polarization imposed on the boundary, where $\partial_n\phi-\varkappa(\phi)$ is fixed on $\partial \Sigma$. This new polarization can be seen as a complex polarization on the boundary phase space, whereas Dirichlet condition gives a real polarization; the two are connected by a Segal-Bargmann transform which can be represented by a partition function of a short cylinder with Dirichlet polarization on one side and Helmholtz condition on the other side.\footnote{In the context of Chern-Simons theory, (generalized) Segal-Bargmann transform via attaching a cylinder with appropriate boundary polarizations is studied in the paper in preparation \cite{CMW}.}
\item It would be very interesting to extend our treatment of 2-dimensional scalar theory to allow cutting and gluing with corners. Correspondingly, we expect the functorial picture to generalize to a fully extended FQFT out of an appropriate Riemannian cobordism 2-category. In topological case, this formalism is known from Baez-Dolan-Lurie \cite{Baez1995},\cite{Lurie2009}. A related question is enrichment of the theory by defects supported on strata.
\item A big open problem is the compatibility of renormalization with locality in more general setting and in higher-dimensional theories. In particular, it would be natural to try to extend our treatment of scalar theory to higher dimension (but restricting the potential $p$ to be renormalizable, e.g. $p(\phi)=\phi^4$ in dimension $\leq 4$ or $p(\phi)=\phi^3$ in dimension $\leq 6$) and studying renormalization and RG flow there.
\item This paper and \cite{Iraso2019},\cite{CMRQ} suggest that there is certain algebraic structure on Feynman graphs which is responsible for Atiyah-Segal type gluing formulae for perurbative partition functions. One can consider graph-valued partition function (in the spirit of LMO invariant or Kontsevich-Kuperberg-Thurston-Lescop construction) and one expects this version of partition function to be an idempotent (or, dually, a group-like element) w.r.t. the gluing operation on graphs.\footnote{A similar problem is currently being investigated in \cite{KY}.}
\end{enumerate}

\appendix

\section{Examples}\label{App: Examples}
In this section we provide some explicit examples of determinants, tadpole functions, and gluing formulae. Even though the main focus of this paper is two-dimensional scalar field theory, we consider also 1-dimensional examples, where answers are simpler and more explicit. All the constructions in this paper are valid, with minor adjustments, also for 1-dimensional scalar field theory. 
\subsection{One-dimensional examples}
\subsubsection{Interval}
%Consider  an interval of length $l$. 
%\begin{equation}
%S[\phi] = \frac12\int_0^l \left(\frac{\dd\phi}{\dd x}\right)^2 + m^2\phi^2
%\end{equation} which has equation of motion 
%\begin{equation}
%\left(-\frac{\d^2}{\d x^2}+m^2\right)\phi = 0.
%\end{equation}
Denote $A^{DD}_{m,l} = (-\frac{\d^2}{\d x^2}+m^2)$ with Dirichlet boundary conditions.
%we need to put boundary conditions on $\phi$ at $\partial_LI=\{0\}$ and $\partial_RI=\{l\}$. Natural candidates are Dirichlet and Neumann boundary conditions. 
%Let us start with Dirichlet boundary conditions at both ends and denote the restriction of $A_{m,l}$ to functions satisfying these boundary conditions by $A_{m,l}^{DD}$. 
The Green's function of the operator $A^{DD}_{m,l}$ can be explicitly computed and yields 
\begin{equation*}
G(x,y) = \frac{1}{m}\frac{\sinh mx \sinh m(l-y)\theta(y-x)+\sinh m (l-x)\sinh my \theta (x-y)}{\sinh ml },
\end{equation*}
where $\theta$ is the Heaviside function, which leads to the tadpole function\footnote{Here $\theta(0) = 1/2$. Notice that the Green's function is actually continuous across the diagonal.  }
\begin{equation*}
\tau(x) = G(x,x) = \frac{\sinh mx \sinh m (l-x)}{m\sinh ml}. 
\end{equation*}
The zeta-regularized determinant of $A$ can be computed\footnote{See e.g. \cite[Example 3, p.220]{Q1993} for a derivation.} as
\begin{equation*}
\det A^{DD}_{m,l} = \frac{2 \sinh ml }{m}.
\end{equation*}
Notice that in the limit as $m \to 0$ we obtain $2l$ - for Dirichlet boundary conditions, the operator $A^{DD}_{0,L}$ has no kernel and we obtain its nonzero determinant.  In particular, we see that the tadpole is consistent (in the weak sense, Definition \ref{def:compatibility}) with zeta-regularization: We have 
\begin{equation*}
\frac{\d}{\d m^2} \log \det (A_{m,l}^{DD}) = 
%\frac{1}{2m}\frac{\d}{\d m}\log\frac{\sinh ml}{ml} = 
\frac{1}{4m}\left(\frac{l\cosh ml}{\sinh ml} - \frac{1}{m}\right) = \int_0^l \tau(x)dx.
\end{equation*}
%and on the other hand 
%\begin{equation}
%\int_0^l T(x)dx = \int_0^l \frac{\sinh mx \sinh m(l-x)}{m\sinh ml} dx = \frac{1}{4m}\left(\frac{l\cosh ml}{\sinh ml} - \frac{1}{m}\right). 
%\end{equation}
%Next, let us compute the Dirichlet-to-Neumann operator. Let ${\eta_L,\eta_R}$ be a pair of boundary values. Then, the unique solution of the equations of motion $\phi_\eta$ satisfying $\phi_\eta(0) = \eta_L$ and $\phi_\eta(l) = \eta_R$ is given by \begin{equation}
%\phi_{\eta} = \eta_L \frac{\sinh m(l-x)}{\sinh ml} + \eta_R \frac{\sinh mx}{\sinh ml}.
%\end{equation}
%The Dirichlet-to-Neumann operator can be read off by taking derivative with respect to the outward normal (which is $\pm\d /\d x$) and evaluating at $x=0$ or $x=l$ respectively. One thus obtains 
%\begin{equation}
%D_{\partial_I}\colon(\eta_L,\eta_R) \mapsto m\left(\frac{\eta_L \cosh ml - \eta_R}{\sinh ml}, \frac{-\eta_L + \eta_R\cosh ml}{\sinh ml}\right).
%\end{equation}
Now consider the gluing of the two intervals $I_{l_1} = [0,1_1]$ and $I_{l_2} = [l_1,l_1+l_2]$ over the the point $Y = \{l\}$. Then the Dirichlet-to-Neumann operator along $Y$ is given by 
$$D_{I_{l_1},Y,I_{l_2}}\colon \eta \mapsto m (\coth ml_1 + \coth ml_2)\eta =m\frac{\sinh m(l_1 +l_2)}{\sinh ml_1 \sinh ml_2}\eta.$$ Then we compute 
\begin{align*}&\det(A^{DD}_{ml_1})\det(A^{DD}_{ml_2})\det\frac12 D_{I_{l_1},Y,I_{l_2}} \\
&= \left(\frac{2}{m}\sinh ml_1\right)\left(\frac{2}{m}\sinh ml_2 \right)\frac12\left(m \coth ml_1 + \coth ml_2\right)  \\
&= \frac{2}{m}\sinh m(l_1 + l_2) = \det A^{DD}_{m,l_1+l_2}.
\end{align*}
The factor $\frac12$ which appears here - in contrast to the gluing formula \eqref{eq:gluing_zeta_determinants} -  arises because we are gluing 1-dimensional determinants. It is a correctional factor in the gluing formula for determinants that is present in \emph{odd} dimensions,\footnote{This factor hints at the fact that in odd dimensions one should normalize the ``measure'' of the path integral accordingly. The correct normalization of the path integral on spacetimes with boundaries should be such that it is compatible with gluing. A similar discussion for factors in gluing of partition functions in abelian BF theory is given in \cite{CMRQ}.} see \cite{Lee2}. Finally, let us consider the gluing of the tadpole function. We will check that for $x < l_1$, 
$ \tau_{l_1 +l_2}(x) = \tau_{l_1}(x) * \tau_{l_2}(x)$ (the other cases are similar). 
Indeed, the left hand side is 
$$\tau_{l_1 +l_2}(x) = \frac{\sinh mx \sinh m(l_1 +l_2 - x)}{m\sinh m(l_1+l_2)}$$ while the right hand side is 
\begin{align*}
\tau_{l_1}(x) &+ \left(\left.\frac{\d}{\d y}\right|_{y=l_1}G(x,y)\right)^2 D_{I_{l_1},Y,I_{l_2}}^{-1} \\ 
&= \frac{\sinh mx \sinh m (l_1 -x)}{m\sinh m l_1} + \left(\frac{\sinh mx}{\sinh ml_1}\right)^2\frac{\sinh ml_1\sinh ml_2}{m\sinh m(l_1 + l_2)} \\
&  =\left( \frac{\sinh mx}{m}\right)\frac{\sinh m(l_1-x)\sinh m(l_1+l_2)+ \sinh mx \sinh m l_2}{\sinh ml_1\sinh m(l_1 +l_2)} = \tau_{l_1+l_2}(x).
\end{align*}
\subsubsection{Circle}
Consider a circle of length $l$, and let $A_{m,l} = -\d^2/\d x^2 + m^2$. The spectrum of this operator is $\lambda_k = (2\pi k/l)^2 + m^2, k \in \mathbb{Z}.$ Then one can compute the determinant as 
\begin{equation}
\det A_{m,l} = 4 \sinh^2 \frac{ml}{2}.\label{eq:detAcircle}
\end{equation} 
For $x,y \in \R$ denote $d(x,y) =   l\left\lbrace\frac{x-y}{l}\right\rbrace$, then the Green's function can be expressed as  
 \begin{equation*}
G(x,y) = \frac{1}{2m}\frac{\cosh m (d(x,y)-l/2)}{\sinh \frac{ml}{2}}
\end{equation*}
and the tadpole function is 
\begin{equation*}
\tau(x) = G(x,x) = \frac{1}{2m}\coth \frac{ml}{2}.
\end{equation*}
Again, one immediately verifies that the tadpole function is (weakly) consistent with zeta-regularization, namely 
$$\frac{1}{2m}\frac{\d}{\d m}\log \det A_{m,l} = \frac{2}{2m}\frac{\d}{\d m}\log \sinh \frac{ml}{2} = \frac{l}{2m}\coth \frac{ml}{2}=\int_{S^1}\tau(x)\,dx$$
%and 
%$$\int_{S^1}*T(x) = \int_{0}^{l}T(x) 
%dt = \frac{l}{2m}\coth\frac{ml}{2}.$$
Next, we want to glue a circle out of two arcs $I_1,I_2$ of length $l_1,l_2$ along the interface $
Y = \{p,q\}$ ( see Figure \ref{fig:gluingcircle}). 
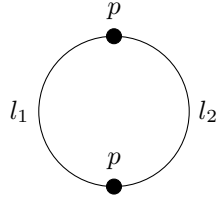
\begin{figure}[h]
\centering
\begin{tikzpicture}
\node[bulk,label=above:$p$] (p) at (0,1) {}; 
\node[bulk,label=above:$p$] (q) at (0,-1) {}; 
\draw (0,0) circle (1cm);
\node[coordinate,label=left:{$l_1$}] at (-1,0) {};
\node[coordinate,label=right:{$l_2$}] at (1,0) {};
\end{tikzpicture}
\caption{Gluing a circle from two intervals of length $l_1,l_2$. }\label{fig:gluingcircle}
\end{figure} The corresponding Dirichlet-to-Neumann operator is the sum of the two operators:
\begin{equation*}
D_N (\eta_p,\eta_q) = m\begin{pmatrix}\coth ml_1 + \coth ml_2 & - \left(\frac{1}{\sinh ml_1} + \frac{1}{\sinh ml_2}\right) \\
- \left(\frac{1}{\sinh ml_1} + \frac{1}{\sinh ml_2}\right) & \coth ml_1 + \coth ml_2\end{pmatrix}\begin{pmatrix}
\eta_p \\
\eta_q
\end{pmatrix}
\end{equation*}
and straightforward computation shows that its determinant is 
\begin{equation*}
\det D_N = \frac{4m^2}{\sinh ml_1 \sinh ml_2}\sinh^2\frac{m(l_1 + l_2)}{2}
\end{equation*}
Therefore, we obtain that the product of determinants is 
\begin{align*}
&\det A^{DD}_{m,l_1}\det A^{DD}_{m,l_1}\det\frac{1}{2} D_N \\
&= \frac{2\sinh ml_1}{m}\frac{2\sinh ml_2}{m}\frac{m^2}{\sinh ml_1 \sinh ml_2}\sinh^2\frac{m(l_1 + l_2)}{2} \\
&= 4\sinh^2 \frac{m(l_1+l_2)}{2} =  \det A_{m,l_1+l_2}
\end{align*}
where again, the factor $\frac12$ in the Dirichlet-to-Neumann operator turns out to be correct\footnote{The general formula says that the prefactor is $e^{\log 2(\zeta_{\Delta_Y}(0) + \dim\ker \Delta_Y)}$, where $\Delta_Y$ is the Laplacian on the interface. If $\dim Y = 0 $ we have $\Delta_Y=0$, and hence $\zeta_{\Delta_Y} = 0$ and $\dim \ker \Delta_Y = \dim C^{\infty}(Y) = |Y|$. In this case $|Y| =2$ - hence the prefactor $4$. } according to the gluing formula \cite{Lee2}. \\
Let us also check the gluing of tadpoles. Again, we will check the case where $x \in I_1$. Then, the gluing formula for the tadpole reads 
\begin{equation*}
\tau_{I_1} * \tau_{I_2}(x) = \tau_{I_1}(x) + \begin{pmatrix}-\frac{\d}{\d\nu}G_{I_1}(x,p) & \frac{\d}{\d\nu}G_{I_1}(x,q)\end{pmatrix}D_N^{-1}\begin{pmatrix}
-\frac{\d}{\d\nu}G_{I_1}(x,p) \\
\frac{\d}{\d\nu}G_{I_1}(x,q))
\end{pmatrix}
\end{equation*}
where matrix multiplication replaces integration.
% We have ($p$ is the left  and $q$ is the right endpoint of $I_1$)
%\begin{align*}
%\frac{\d}{\d\nu}G_{I_1}(x,p) &= \frac{\sinh m(l_1 -x)}{\sinh ml_1} \\
%\frac{\d}{\d\nu}G_{I_1}(x,q) &= -\frac{\sinh mx}{\sinh ml_1}. \\
% \end{align*}
The inverse of the Dirichlet-to-Neumann operator can be explicitly computed and yields 
$$D_N^{-1} = \begin{pmatrix}\frac{1}{2} \coth m(l_1 + l_2)/2 &  \frac{\sinh ml_1 + \sinh ml_2}{4\sinh^2m(l_1 + l_2)/2}  \\
\frac{\sinh ml_1 + \sinh ml_2}{4\sinh^2 m(l_1 + l_2)/2} & \frac{1}{2}\coth m(l_1 + l_2)/2
\end{pmatrix} $$
A straightforward computation then shows 
%\begin{multline*}\frac{1}{\sinh^2 ml_1}\begin{pmatrix}
%\sinh m(l_1 - x) & \sinh mx 
%\end{pmatrix} D_N^{-1} \begin{pmatrix}
%\sinh m(l_1 - x) \\ \sinh mx  \end{pmatrix} \\ = \frac{1}{4\sinh ml_1}\left(2 \cosh m(l_1 -2x) + (\cosh ml_1 - \cosh ml_2 )\sinh^{-2}\frac{m(l_1+l_2)}{2})\right), 
%\end{multline*}
%and finally, after more manipulations of hyperbolic functions, 
\begin{align*}
\tau_{I_1} * \tau_{I_2}(x) &= \frac{\sinh mx \sinh m(l_1 -x) }{m \sinh ml_1} \\
&+ \frac{1}{4\sinh ml_1}\left(2 \cosh m(l_1 -2x) + (\cosh ml_1 - \cosh ml_2 )\sinh^{-2}\frac{m(l_1+l_2)}{2}\right)  \\
&= \frac{1}{2m}\coth\frac{m(l_1 + l_2)}{2} = \tau(x)
\end{align*}

\subsection{Two-dimensional examples}
Now let us turn to two-dimensional examples. The main tool that we will use is the heat kernel of the Laplacian $K_\Delta(t,x,y)$ and the heat kernel for the corresponding Helmholtz operator, $K_{\Delta + m^2} = e^{-m^2t}K_{\Delta}(t,x,y)$. We recall some formulae for heat kernels of standard metrics. On the real line, the heat kernel is 
\begin{equation*}
K_\Delta^{\R}(t,x,y) = \frac{1}{\sqrt{4\pi t}}\, e^\frac{-(x-y)^2}{4t}.
\end{equation*}
From this, one can infer the heat kernel on the circle of length $L$ through periodic summation: 
\begin{equation}
K_\Delta^{S^1}(t,x,y) = \frac{1}{\sqrt{4\pi t}}\sum_{k=-\infty}^\infty e^{\frac{-(x-y-kL)^2}{4t}} \label{eq:heat_circle}
\end{equation}
and the heat kernel on an interval of length $L$ with Dirichlet boundary conditions (through image charges): 
\begin{equation}
K^{DD}_\Delta(t,x,y) = \frac{1}{\sqrt{4\pi t}}\sum_{k=-\infty}^\infty e^{\frac{-(x-y-2L)^2}{4t}}  -e^{\frac{-(x+y-2kL)^2}{4t}} \label{eq:heat_interval}
\end{equation}
In addition, we recall the fact that the heat kernel of the Laplacian of a product metric is the product of the heat kernels of the Laplacians associated to the two metrics. 
\subsubsection{Torus}
First, we consider a torus $T\equiv T_{L_1,L_2}$ of circumferences $L_1$ and $L_2$. Then the heat kernel is given by 
\begin{equation*}
K^T_{\Delta}(t,(x,x'),(y,y')) = \frac{1}{4\pi t}\sum_{k,l=-\infty}^\infty e^{\frac{-(x-x'-kL_1)^2}{4t}}e^{\frac{-(y-y'-lL_2)^2}{4t}}
\end{equation*}
Its restriction to the diagonal reads
\begin{equation*}
\theta^T_{\Delta}(t,(x,y)) \equiv \theta^T_\Delta(t) = \frac{1}{4\pi t}\sum_{k,l=-\infty}^\infty e^{\frac{-(kL_1)^2}{4t}}e^{\frac{-(lL_2)^2}{4t}}
\end{equation*}
which is conveniently expressed in terms of the Jacobi theta function
\begin{equation}
\vartheta(z,\tau) = \sum_{k=-\infty}^{\infty}\exp(\pi i k^2 \tau + 2\pi i k z) \label{eq:deftheta}
\end{equation}
as 
\begin{equation*}
\theta^T_{\Delta}(t,(x,y)) \equiv \theta^T_\Delta(t) = \frac{1}{4\pi t}\vartheta\left(0,\frac{iL_1^2}{4\pi t}\right)\vartheta\left(0,\frac{iL_2^2}{4\pi t}\right).
\end{equation*}
The heat kernel of $A = \Delta + m^2$ is then given by 
$$\theta^T_A(t) = \frac{e^{-m^2t}}{4\pi t}\vartheta\left(0,\frac{iL_1^2}{4\pi t}\right)\vartheta\left(0,\frac{iL_2^2}{4\pi t}\right). $$
Extracting the divergence at $t=0$, we write 
\begin{equation*}
\theta^T_A(t) = \frac{e^{-m^2t}}{4\pi t} + e^{-m^2t}h(t) 
\end{equation*}
where $h(t) = \frac{1}{4\pi t}\left(\vartheta\left(0,\frac{iL_1^2}{4\pi t}\right)\vartheta\left(0,\frac{iL_2^2}{4\pi t}\right)-1\right)$ falls off like $e^{-C/t}$ as $t \to 0$.
The local zeta function of $A$ is the Mellin transform of this object: 
\begin{equation*}
\zeta_A^{loc}(s) = \frac{1}{\Gamma(s)}\int_0^\infty t^{s-1}\frac{e^{-m^2t}}{4\pi t}(1 + h(t))dt
\end{equation*}
For $\operatorname{Re} s >1$, the first term is given by 
$$\frac{1}{\Gamma(s)}\int_0^\infty t^{s-1}\frac{e^{-m^2t}}{4\pi t} = \frac{\Gamma(s-1)}{4\pi\Gamma(s)m^{2(s-1)}} = \frac{1}{4\pi (s-1) m^{2(s-1)}}.$$
The second integral can be explicitly given in terms ofa modified Bessel function of the second kind $K_s(x)$, thus the analytically continued zeta function is 
$$\zeta_A(s)^{loc} = \frac{1}{4\pi (s-1) m^{2(s-1)}} + \frac{m^{1-s}}{2\pi\Gamma(s)}\sum_{k,l \neq 0}b_{k,l}^{s-1}K_{1-s}(2 m b_{k,l})$$ 
where $b_{k,l} = \sqrt{k^2L_1^2 + l^2L_2^2}$. The first term has a pole at $s=1$, while the second term is an entire function of $s$. The tadpole function is the finite part of $\zeta_A(s)^{loc}$ at $s=1$: 
\begin{equation*}\tau^{reg} = -\frac{\log m^2}{4\pi} + \frac{1}{2\pi}\sum_{k,l\neq 0}K_0(2mb_{k,l})
\end{equation*}
and the logarithm of the zeta-regularized determinant is 
\begin{equation*}
\log \det A = \zeta'_A(0) = \frac{L_1L_2}{4 \pi}m^2(\log m^2 -1) + \frac{mL_1L_2}{2\pi}\sum_{k,l\neq 0}\frac{K_1(2mb_{k,l})}{b_{k,l}}.
\end{equation*}
One immediately verifies $\frac{\d}{\d m^2}\log \det A = -\int_T\tau^{reg}\,\dvol_T$ from the relation $\frac{\d}{\d x}x^s K_s(x) = -x^sK_{s-1}(x)$. 
\subsubsection{Cylinder}\label{sec:example cylinder}
Let us consider a cylinder $C$ of circumference $L$ and height $H$. Then, the heat kernel is given by 
\begin{equation*}
K^C_\Delta(t,(x,y),(x',y')) = K_{\Delta}^{S^1}(t,x,x')K_{\Delta}^{DD}(t,y,y'). 
\end{equation*}
The restriction to the diagonal is given by, restricting to the diagonal in \eqref{eq:heat_circle},\eqref{eq:heat_interval},  
\begin{equation} 
\theta_{\Delta}^C(t,(x,y)) = \frac{1}{4\pi t}\sum_{k=-\infty}^{\infty}e^{-\frac{(2kL)^2}{4t}}\sum_{l=-\infty}^\infty e^{\frac{-(lH)^2}{4t}}  -e^{\frac{-(2y-2lH)^2}{4t}}\label{eq:heat_cyl1}
\end{equation}
Using the Jacobi theta function \eqref{eq:deftheta} we can express \eqref{eq:heat_cyl1} as 
\begin{equation}
\theta_{\Delta}^C(t,(x,y)) = \frac{1}{4\pi t}\vartheta\left(0,\frac{iL^2}{4\pi t}\right)\left(\vartheta\left(0,\frac{iH^2}{\pi t}\right) - e^{\frac{-y^2}{t}}\vartheta\left(\frac{yH}{\pi i t },\frac{iH^2}{\pi t}\right)\right)\label{eq:heat_cyl_theta1}
\end{equation}
Recall that the theta function has the modular transform $\vartheta(z/\tau,-1/\tau) = \alpha \vartheta (z,\tau)$, where $\alpha = (-i\tau)^\frac12 \exp(\pi i z^2 /\tau)$.  
Setting $\tau = \frac{4\pi i t}{H^2}, z= \frac{y}{H}$, we can rewrite \eqref{eq:heat_cyl_theta1} as 
\begin{equation}
\theta_{\Delta}^C(t,(x,y)) = \frac{1}{2LH}\vartheta\left(0,\frac{4\pi i t}{L^2}\right)\left(\vartheta\left(0,\frac{\pi i t}{H^2}\right) - \vartheta\left(\frac{y}{H},\frac{\pi i t}{H^2}\right)\right)\label{eq:heat_cyl_theta2}
\end{equation}
Integrating over $(x,y) \in C$ we obtain
$$\Theta^C_\Delta(t):=\int_C*\theta_\Delta^C = \frac{1}{2}\vartheta\left(0,\frac{4\pi i t}{L^2}\right)\left(\vartheta\left(0,\frac{\pi i t}{H^2}\right) - 1\right)$$
which we rewrite using the modular transform as 
\begin{equation}\Theta^C_\Delta(t) = \frac{LH}{4\pi t}\vartheta\left(0,\frac{i L^2}{4\pi t}\right)\left(\vartheta\left(0,\frac{iH^2}{\pi t}\right) - \frac{\sqrt{\pi t}}{H}\right)= \frac{LH}{4\pi t} - \frac{L}{4\sqrt{\pi t}} + h(t)\label{eq:Theta_decomp}\end{equation}
where 
$$h(t) =  \frac{LH}{4\pi t}\left(\vartheta\left(0,\frac{i L^2}{4\pi t}\right)\vartheta\left(0,\frac{iH^2}{\pi t}\right)-1\right) - \frac{L}{4\sqrt{\pi t}}\left(\vartheta\left(0,\frac{i L^2}{4\pi t}\right)-1 \right) $$
satisfies $h(t) \simeq e^{-C/t}/t$ as $t \to 0$. 
Now consider the operator $\A := \Delta_C + m^2$, it heat kernel is given by $\theta_\A^C = e^{-m^2t}\theta_\Delta$. The local zeta function of $A$ is its Mellin transform, 
$$\zeta_\A(s,(x,y)) = \frac{1}{\Gamma(s)}\int_0^\infty t^{s-1}e^{-m^2t}\theta_\Delta dt.$$ 
To investigate its behavior at $s=0,1$, we define $g(t)= e^{-m^2t}/(4\pi t)$ and then write 
$$\zeta(s,(x,y)) = \frac{1}{\Gamma(s)}\int_0^\infty t^{s-1}\frac{e^{-m^2t}}{4\pi t}dt +   \frac{1}{\Gamma(s)}\int_0^\infty t^{s-1}e^{-m^2t}\left(\theta_\Delta-\frac{1}{4\pi t}\right)dt.$$ 
The second integral converges absolutely at $s=1$. The first integral can be explicitly computed (for $\operatorname{Re} s >1$) and yields 
\begin{equation*}\frac{1}{\Gamma(s)}\int_0^\infty t^{s-1}\frac{e^{-m^2t}}{4\pi t}dt = \frac{\Gamma(s-1)}{4\pi\Gamma(s)m^{2(s-1)}} = \frac{1}{(s-1)}\cdot\frac{1}{4\pi m^{2(s-1)}}
\end{equation*} 
The zeta-regularized tadpole function is given by 
\begin{align}
\tau_\A^{reg}(x,y) & = \lim_{s\rightarrow 1} \frac{\d}{\d s}(s-1)\zeta_\A(s,(x,y)) \notag \\
&= - \frac{\log m^2 }{4\pi }+  \frac{1}{\Gamma(s)}\int_0^\infty t^{s-1}e^{-m^2t}\left(\theta^C_\Delta-\frac{1}{4\pi t}\right)dt \label{eq:cyl_tadpole}
\end{align}
which $\theta^C_\Delta$ given by \eqref{eq:heat_cyl_theta1}. The zeta function of $\A$ is $\zeta_\A(s) = \frac{1}{\Gamma(s)}\int_0^\infty t^{s-1}e^{-m^2t}\Theta_\Delta dt $ and using the decompositon \eqref{eq:Theta_decomp} we write 
\begin{equation*}
\zeta_\A(s) = \frac{LH}{(s-1)}\cdot\frac{1}{4\pi m^{2(s-1)}} - \frac{L}{4\sqrt{\pi}}\cdot\frac{\Gamma(s-1/2)}{\Gamma(s)\,m^{2s-1}} + \frac{1}{\Gamma(s)}\int_0^\infty t^{s-1}e^{-m^2t}h(t)dt
\end{equation*}
Here, the last integral converges absolutely for any $s\in \C$. The logarithm of the zeta-regularized determinant is given by 
\begin{equation*}
\log \det\A = \zeta_\A'(0) = \frac{LH}{4\pi}m^2(\log m^2 -1) + \frac{L m}{2 } + \int_0^\infty t^{-1}e^{-m^2t}h(t)dt
\end{equation*}
%\subsubsection{Gluing two cylinders into a cylinder}
\textbf{Gluing two cylinders into a cylinder.} 
Next, we will investigate gluing of two cylinders into a longer cylinder. For this, consider the Dirichlet-to-Neumann operator $D_H$, say, on the lower end of a cylinder of circumference $L$ and height $H$, with Helmholtz operator $A_H$ (Circumference $L$ and mass $m$ are fixed). It is easiest to determine the Dirichlet-to-Neumann operator by its action on the basis of $C^\infty(S^1)$ given by ${\eta}_n(x) = \exp(2\pi i n x/L)$. The unique function  $\phi_{\eta_n}(x,y)$ satisfying $A_{H}\phi_{\eta_n}=0$, $\phi_{\eta_n}(x,0) = \eta_n(x)$ and $\phi_{\eta_n}(x,H) = 0$ is 
\begin{equation*}
\phi_{\eta_n}(x,y) = \eta_n(x)\frac{\sinh (H-y)\omega_n}{\sinh H\omega_n},
\end{equation*} 
where $\omega_n = \sqrt{m^2 + (2\pi n/L)^2}$. Taking the $y$ derivative at $y=0$ we find 
\begin{equation*}
D_H(\varphi_{\eta_n}) = \varphi_{\eta_n}\omega_n\coth H\omega_n. 
\end{equation*}
When gluing two cylinders, the relevant Dirichlet-to-Neumann operator is $D_{H_1} + D_{H_2}$, which has eigenvalues $\lambda_n = \omega_n(\coth H_1 \omega_n + \coth H_2 \omega_n)$. To compute its zeta-regularized determinant, recall that $\det_\zeta(K_1K_2) = \det_\zeta K_1\det(K_2)$ if $K_2$ 
is the identity plus a trace class operator (and hence has a well-defined Fredholm determinant). %Fredholm. 
In this example, we let $K_2\colon\eta_n \mapsto \frac{1}{2}(\coth H_1 \omega_n + \coth H_2 \omega_n)\eta_n$, which 
is identity plus trace class %defines a Fredholm operator 
and $K_1\colon\eta_n \mapsto 2\omega_n \eta_n$. The zeta determinant of $K_1$ is the square root of the zeta-determinant of the Helmholtz operator\footnote{Taking squares commutes with zeta-regularized products. On the other hand, multiplying all terms by a constant $a$ multiplies the zeta-regularized product $a^{\zeta(0)}$, where $\zeta$ is the zeta function of the sequence. One can check that in this case $\zeta(0) = 0$, using e.g. the results of \cite{Q1993}.} on $S^1$ given in \eqref{eq:detAcircle}: $\det K_1 = 2\sinh mL/2$. Thus, 
\begin{equation*}
\det 
(D_{H_1} + D_{H_2}) = 2\sinh \frac{mL}{2}\; \prod_{n \in \mathbb{Z}}\frac{ \coth H_1\omega_n + \coth H_2\omega_n}{2}. 
\end{equation*}
The gluing formula for zeta-regularized determinants thus implies the interesting identity 
\begin{align*}
\log \det A_{H_1+ H_2}& - \log \det A_{H_1} - \log \det A_{H_2} \\
&= - \frac{Lm}{2} + \int_0^\infty t^{-1}e^{-m^2t}(h_{H_1+H_2}(t) - h_{H_1}(t) - h_{H_2}(t))dt \\
&= \log \det  (D_{H_1} + D_{H_2}) \\
&= %\log \sinh mL/2 + \log 2 +
\log \left(2\sinh \frac{mL}{2}\right) +
 \sum_{n \in \mathbb{Z}}\log\left(\frac{\coth H_1 \omega_n + \coth H_2 \omega_n}{2}\right)
\end{align*}
and one can check numerically that his formula holds. \\
We can extend this numerical check to tadpoles. The value of the glued tadpole $\tau_{H_1} * \tau_{H_2}$ at some point $(x,y)\in C_{H_1}$ is $$\tau_{H_1}(x,y) + \int_{S^1 \times S^1}  \partial_\nu G((x,y),(x',0))(D_{H_1}+D_{H_2})^{-1}(x',x'')\partial_\nu G((x,y),(x'',0))dx'dx''. $$
Here the tadpole function on the cylinder is given by \eqref{eq:cyl_tadpole}; $\partial_\nu G$ is the normal derivative of the Green's function in the second argument.
To compute the second term, we expand the normal derivative in Fourier modes 
$$\partial_\nu G ((x,y),(x',0)) = \int_0^\infty dt\frac{e^{-m^2t}}{\sqrt{4\pi t^3}L}\left(\sum_{n=-\infty}^\infty e^{-\frac{4t\pi^2n^2}{L^2}}\varphi_n(x')\right)g(t,y),$$
where $$g(t,y) = e^{-\frac{y^2}{4t}}\left(y\,\vartheta\left(\frac{yH}{\pi i t },\frac{iH^2}{\pi t}\right)-\frac{2H_1}{2\pi i}\vartheta'\left(\frac{yH}{\pi i t },\frac{iH^2}{\pi t}\right)\right).$$
The inverse of the Dirichlet-to-Neumann operator is given by $$\varphi_n(x)\mapsto (\omega_n(\coth(H_1\omega_n)+\coth(H_2\omega_n)))^{-1}\varphi_n(x)$$
 so that we obtain 
\begin{multline*}
\tau_{H_1}*\tau_{H_2} = \tau_{H_1} + \\
 + \int_{[0,\infty]^2}dtdu\frac{e^{-m^2(t+u)}}{\sqrt{4\pi u^3}\sqrt{4\pi t^3} L^2}\sum_{n=-\infty}^\infty\frac{e^{-\frac{4(t+u)\pi^2n^2}{L^2}}}{\omega_n(\coth (H_1\omega_n) + \coth(H_2\omega_n))}g(t,y)g(u,y).
\end{multline*}
Again one can check numerically that this equals $\tau_{H_1 + H_2}(y)$.

\subsubsection{Sphere}

Consider a sphere of radius $R$.

\textbf{Green's function.}
The Green's function for the Helmholtz operator on a sphere is:
\begin{equation}\label{G_sphere}
G_{S^2}(x,y) = \frac{1}{4\cos \pi \big(\frac14-(mR)^2\big)^{\frac12}}\cdot   {}_2F_1\Big(\alpha_1,\alpha_2;1;\cos^2\frac{d(x,y)}{2R}\Big)
\end{equation}
with $\alpha_{1,2}$ the roots of the quadratic equation $\alpha^2-\alpha+(mR)^2=0$ %as in (\ref{lambda DN hemisphere}) 
and $d(x,y)$ is the geodesic distance in the sphere metric; ${}_2F_1$ is the hypergeometric function. 

Note that the Green's function on a hemisphere can be obtained from (\ref{G_sphere}) by the image charge method: 
$$G_{H^+}(x,y)=G_{S^2}(x,y)-G_{S^2}(x,\hat{y})$$ 
where  $\hat{y}$ is the reflection of $y$ through the equatorial plane.

\begin{remark}\label{rem: sphere propagator m->0 limit}
%We also remark that 
The asymptotics $m\ra 0$  of the Green's function (\ref{G_sphere}) reads
$$
G_{S^2}(x,y)\quad \underset{m\ra 0}{\sim}\quad  \frac{1}{\mr{Area}(S^2)\cdot 
%4\pi R^2 
m^2}\underbrace{-\frac{1}{2\pi}\left(\log\sin\frac{d(x,y)}{2R}+\frac12\right)}_{G^\mr{fin}}+\mr{O}(m^2R^2)
$$
The finite part $G^\mr{fin}$ here is the propagator of the massless scalar theory on $S^2$. It satisfies 
\begin{equation}\label{Delta G^fin}
\Delta G^\mr{fin}=\delta(x,y)-c
\end{equation} 
with the constant $c=\frac{1}{\mr{Area}(S^2)}$.\footnote{
The constant shift by $-c$ is related to the zero-mode of $\Delta$. More precisely, the integral operator defined by $G^\mr{fin}$ inverts $\Delta$ only on the orthogonal complement of constant functions and vanishes on constant functions.
} 

 In fact, as one can show from examining the $t\ra\infty$ asymptotics of the heat kernel, this behavior is universal: for any surface $\Sigma$, the asymptotics $m\ra 0$ of the Green's function is 
$ G(x,y)\;\underset{m\ra 0}{\sim} \frac{1}{\mr{Area}(\Sigma)\cdot m^2}+G^\mr{fin}(x,y)+\mr{O}(m^2)$ with $G^\mr{fin}(x,y)$ some function satisfying the  (\ref{Delta G^fin}) with $c=\frac{1}{\mr{Area}(\Sigma)}$.
\end{remark}
%\marginpar{\textcolor{cyan}{I would add a remark saying that this is the propagator in the massless theory. Also: I wonder whether subtracting the singular part corresponds to a shift of normalization similar to remark below.}}

\textbf{Tadpole.}
Next, note that the zeta-regularized tadpole on the sphere can be computed from $y\ra x$ asymptotics of the Green's function (\ref{G_sphere}) (by first calculating the point-splitting tadpole and then using Corollary \ref{corollary:tadpole difference}) and yields
\begin{equation}\label{sphere tadpole}
\tau^\mr{reg}=\frac{\log R^2-\psi(\alpha_1)-\psi(\alpha_2)}{4\pi}
\end{equation}
where $\psi(z)=\frac{d}{dz}\log \Gamma(z)$ is the digamma function.

\textbf{Determinant.}
%The partition function of the free massive scalar theory is $Z=\det(\Delta+m^2)^{-\frac12}$ with 
Helmholtz operator on the sphere of radius $R$ has eigenvalues $\lambda_l=\frac{l(l+1)}{R^2}+m^2$ with multiplicities $2l+1$, for $l=0,1,2,\ldots$. The corresponding zeta-regularized determinant is:
%The determinant of the Helmholtz operator on the sphere is:
\begin{equation}\label{det sphere}
\begin{aligned}
\det(\Delta+m^2) &= \Big(
\prod_{l\geq 0} \left(
\frac{l(l+1)}{R^2}+m^2
\right)^{2l+1}\Big)_\mr{reg} \\
&= R^{-2 \zeta_{\bar{A}}(0)} \Big(\prod_{l\geq 0} (l(l+1)+m^2 R^2)^{2l+1}\Big)_\mr{reg} \\
&= R^{-2 (\frac13-m^2R^2)} \mathbb{F}(m^2R^2)
\end{aligned}
\end{equation}
Here we denoted $\bar{A}=R^2(\Delta+m^2)$ the rescaled Helmholtz operator; we denoted its regularized determinant by  $\mathbb{F}(m^2R^2)$. We are using the property of zeta-regularized determinants $\det(c \bar{A})=c^{\zeta_{\bar{A}}(0)}\det(\bar{A})$. The value $\zeta_{\bar{A}}(0)=\frac13-m^2R^2$ is calculated straightforwardly.\footnote{
Indeed, for $\mr{Re}(s)>1$ one has $\zeta_{\bar{A}}(s)=\frac{1}{\Gamma(s)}\int_0^\infty dt\, t^{s-1}e^{-m^2 R^2 t} f(t)$, where $f(t)=\sum_{l=0}^\infty {(2 l+1)e^{-l(l+1)t} }$. Using Euler-Maclaurin formula, one obtains the asymptotics $f(t)\sim \frac{1}{t}+\frac13+\mathrm{O}(t)$ at $t\ra 0$. Thus, the analytic continuation of $\zeta_{\bar{A}}(s)$ to $s=0$ is:
$\zeta_{\bar{A}}(0)=\lim_{s\ra 0} \frac{1}{\Gamma(s)}\int_0^\infty dt\,t^{s-1} e^{-m^2 R^2 t}(f(t)-\frac{1}{t}-\frac13)+\frac{(m^2 R^2)^{1-s}}{s-1}+\frac13 (m^2 R^2)^{-s} =\frac13-m^2 R^2 $.
}

In fact, we can determine the function $\mathbb{F}(z)$ in (\ref{det sphere}) explicitly from knowing the zeta-regularized tadpole on the sphere (\ref{sphere tadpole}), by integrating the tadpole against mass (cf. Corollary \ref{corollary: tau via d/dm^2 log det}). Indeed, setting $R=1$, we obtain
\begin{equation}\label{FF}
\begin{aligned}
\log\mathbb{F}(z)&=\int^z dm^2\, 4\pi\, \tau(m^2)\big|_{R=1} \\
& =C+\log z- \int_0^{z} dm^2\, \left(\psi(\alpha_1)+\psi(\alpha_2)+\frac{1}{m^2}\right)
\end{aligned}
\end{equation}
where $\alpha_{1,2}=\frac12\pm \sqrt{\frac14-m^2}$ are the roots of the equation $\alpha^2-\alpha+m^2=0$. The constant in (\ref{FF}) is 
\begin{equation*}
C=\log {\det}'\Delta\big|_{R=1} = \frac12-4\zeta'(-1)
\end{equation*}
-- the logarithm of the determinant of the Laplacian on the unit sphere (with zero-mode excluded), for which we quote the result from \cite{kumagai1999determinant}.

Function (\ref{FF}) has the following asymptotics at $z\ra 0$ and at $z\ra \infty$:
$$
\log\mathbb{F}(z)\underset{z\ra 0}{\sim} \log z+C+\mathrm{O}(z),\qquad
\log\mathbb{F}(z)\underset{z\ra \infty}{\sim} -z \log z+z+\frac13 \log z 
%-\log \Gamma(z)-\frac16\log z
+ o(1)
$$

The integral (\ref{FF}) can be evaluated in terms of Barnes $G$-function (a.k.a. ``double Gamma function''), yielding
\begin{equation*}
\log \mathbb{F}(z)=C-2z-\log\left(\frac{1}{\pi}\cos\pi\Big(\frac14-z\Big)^{\frac12}\right) + 2\log\left(G(\alpha_1)G(\alpha_2)\right)\Big|_{\alpha_{1,2}=\frac12\pm\big(\frac14-z\big)^{\frac12}}
\end{equation*}

Putting together the result for the determinant, we have:
\begin{lemma} The zeta-regularized determinant of the Helmholtz operator $\Delta+m^2$ on the sphere of radius $R$ is:
\begin{equation*}
\begin{aligned}
\det(\Delta+m^2)&= e^{\frac12-4\zeta'(-1)} R^{-2(\frac13-m^2R^2)}\cdot \\ 
&\cdot\frac{\pi e^{-2 m^2R^2}}{ \cos\pi\Big(\frac14-m^2 R^2\Big)^{\frac12} } \;\; 
%\cdot \exp\Big(-\int_0^{m^2R^2}d\mu\, \left(\psi\big(\frac12-\sqrt{\frac14-\mu}\big)+ \psi\big(\frac12+\sqrt{\frac14-\mu}\big)+\frac{1}{\mu} \right)%\Big)
G(\alpha_1)^2G(\alpha_2)^2 \Big|_{\alpha_{1,2}=\frac12\pm\big(\frac14-m^2R^2\big)^{\frac12}}
\end{aligned}
\end{equation*}
\end{lemma}

In particular, in the limit $m\ra 0$, the determinant behaves as 
$$\det(\Delta+m^2) \underset{m\ra 0}{\sim} e^C R^{-2\cdot \frac13} \cdot m^2 R^2$$ 
Therefore, at zero mass, the determinant with excluded zero-mode on a sphere of radius $R$ is
$$ {\det}'\Delta = e^C R^{-2\cdot\frac13+2}=e^C R^{\frac43} $$

\subsection{Dirichlet-to-Neumann operators: explicit examples}

\subsubsection{Example: disk}
For $\Sigma$ a disk of radius $R$ with flat metric, the Dirichlet-to-Neumann operator acts diagonally in the basis $\phi_n(\theta)=e^{in\theta}$ in the space of $L^2$ functions on the boundary circle (with $\theta$ the polar angle):
\begin{equation*}
D_\Sigma:\; \phi_n(\theta) \mapsto \lambda_n\cdot \phi_n(\theta) 
\end{equation*}
%with $0\leq \theta<2\pi$ the polar coordinate parameterizing the boundary circle, with $\phi_n(\theta)=e^{in\theta}$ and 
with $n\in\ZZ$ and with  eigenvalues
\begin{equation}\label{lambda DN disk}
\lambda_n=m \frac{I'_n(mR)}{I_n(mR)} = \frac{m}{2}\, \frac{I_{n+1}(mR)+I_{n-1}(mR)}{I_n(mR)}
\end{equation}
where $I_n$ is the modified Bessel's function. This follows from the fact that the general solution of Helmholtz equation on the disk can be written, via separation of variables in polar coordinates, as
$\phi(\theta,r)=\sum_{n=-\infty}^\infty c_n \phi_n(\theta) I_n(mr)$ with $c_n$ constant coefficients.

\begin{comment}
\marginpar{rename $\lambda_n^\varkappa\ra \omega_n$ to have agreement of notations with sec. 9.4.1}
We note that the operator $\varkappa=(\Delta+m^2)^{1/2}$ on the circle $\partial \Sigma$ is also diagonalized in the basis $\phi_n(\theta)$, with eigenvalues 
\begin{equation}\label{lambda kappa}
\lambda^\varkappa_n=\left(\frac{n^2}{R^2}+m^2\right)^{1/2}
\end{equation} 
A straightforward computation shows that
\begin{equation}
 \frac{\lambda_n}{\lambda_n^\varkappa}\quad \underset{n\ra\infty}{\sim} \quad 1-\frac{(mR)^2}{2n^3}+O(n^{-4})  
\end{equation}
Thus, $K D_\Sigma$ is the identity plus a pseudodifferential operator of order $-3$.

For comparison, note that for a cylinder of height $H$ (see Section 9.4.1 [INSERT CORRECT REFERENCE]), considering the block of the Dirichlet-to-Neumann operator connecting one of the two boundary circles to itself, we get
\begin{equation}
\frac{\lambda_n}{\lambda_n^\varkappa}=\coth H \lambda_n^\varkappa  \quad \underset{n\ra\infty}{\sim} \quad 1+ O(n^{-\infty})
\end{equation}
Thus, here $KD_\Sigma$ is the identity plus a smoothing operator. This is always the case when the metric on $\Sigma$ near the boundary is the product metric on $\partial \Sigma\times [0,\epsilon)$. [REFERENCE?]
\end{comment}

\subsubsection{Example: hemisphere}
Consider $\Sigma$ a hemisphere of radius $R$ with standard metric. The Dirichlet-to-Neumann operator again acts diagonally in the basis of functions $\phi_n(\theta)=e^{in\theta}$, with $\theta$ the polar angle parameterizing the equator (the boundary of the hemishpere):
\begin{equation*}
D_\Sigma:\; \phi_n(\theta) \mapsto \lambda_n \cdot \phi_n(\theta)
\end{equation*}
with $n\in\ZZ$ and eigenvalues
\begin{equation}\label{lambda DN hemisphere}
\lambda_n=\frac{2}{R} \frac{\Gamma\left(\frac{n+1+\alpha_1}{2}\right)\Gamma\left(\frac{n+1+\alpha_2}{2}\right)}{\Gamma\left(\frac{n+\alpha_1}{2}\right)\Gamma\left(\frac{n+\alpha_2}{2}\right)}
\end{equation}
with $\alpha_{1,2}$ the two roots of the quadratic equation $\alpha^2-\alpha+(mR)^2=0$. One proves this similarly to (\ref{lambda DN disk}) -- from the separation of variables for the Helmholtz equation in spherical coordinates.

\begin{comment}
Comparing the eigenvalues (\ref{lambda DN hemisphere}) to the ``reference'' eigenvalues of the root of Helmholtz operator on the boundary (\ref{lambda kappa}) at large $n$, we get
\begin{equation}
\frac{\lambda_n}{\lambda^\varkappa_n} \quad \underset{n\ra\infty}{\sim} \quad 1-\frac{(mR)^2}{4n^4}+O(n^{-5})
\end{equation}
Thus, we see that in this case $KD_\Sigma$ is the identity plus pseudodifferential operator of order $-4$ (note the higher regularity than in the case of a disk).
\end{comment}

We remark that the zeta-regularized determinant of $D_\Sigma$ for a hemisphere can calculated explicitly, yielding
\begin{equation}\label{det DN hemishpere}
{\det}_\mr{reg}D_\Sigma= \underbrace{{\det}_\mr{reg}(\varkappa)}_{2\sinh (\pi m R)}\cdot \prod_{n=-\infty}^\infty \frac{\lambda_n}{\omega_n} =\quad 2\cos \pi \Big(\frac14-(mR)^2\Big)^{1/2}
\end{equation}
where $\omega_n=(\frac{n^2}{R^2}+m^2)^{1/2}$ are the eigenvalues of the operator $\varkappa=(\Delta+m^2)^{1/2}\big|_{\partial \Sigma}$. Here to compute the second factor in the middle expression (the Fredholm determinant of $\varkappa^{-1}D_\Sigma$), the crucial observation is that (\ref{lambda DN hemisphere}) can be written in the form $\lambda_n=\frac{2}{R}\frac{f_{n+1}}{f_{n}}$, which allows one to compute the finite product 
$\prod_{n=-N}^N \frac{\lambda_n}{\lambda_n^\varkappa} = 2^{2N+1} \frac{f_{N+1}}{f_{-N}} 
\prod_{n=-N}^N  ( (n+imR) {(n-imR)})^{-1/2}$ 
-- this is a certain combination of Gamma functions, and the limit $N\ra\infty$ can be evaluated straightforwardly.

By the BFK gluing formula for determinants, the expression (\ref{det DN hemishpere}) appears as a ratio of the determinant of the Helmholtz operator on a sphere $S^2$ and the product of determinants of the Helmholtz operators on the upper and lower hemispheres $H^+$, $H^-$:
\begin{multline*}
\frac{{\det}_{\mr{reg}}(\Delta+m^2)_{S^2}}{{\det}_{\mr{reg}}(\Delta+m^2)_{H^+}\cdot {\det}_{\mr{reg}}(\Delta+m^2)_{H^-}}= \frac{\Big(\prod_{l\geq 0}\big(\frac{l(l+1)}{R^2}+m^2\big)^{2l+1}\Big)_\mr{reg}}{\left(\Big(\prod_{l\geq 0}\big(\frac{l(l+1)}{R^2}+m^2\big)^{l}\Big)_\mr{reg}\right)^2} \\
= \Big(\prod_{l\geq 0} \left(\frac{l(l+1)}{R^2}+m^2\right)\Big)_\mr{reg} \quad \underset{\mr{BFK}}{=} \quad {\det}_\mr{reg} (D_{H^+}+D_{H^-})= 2\cos \pi \Big(\frac14-(mR)^2\Big)^{1/2}
\end{multline*}

%A related result to the computation (\ref{lambda DN hemisphere}) is that the Green's function for the Helmholtz operator on a sphere is:
%\begin{equation}\label{G_sphere}
%G_{S^2}(x,y) = \frac{1}{4\cosh \pi \left((mR)^2-\frac14\right)^{1/2}}\cdot   {}_2F_1\Big(\alpha_1,\alpha_2;1;\cos^2\frac{d(x,y)}{2R}\Big)
%\end{equation}
%with $\alpha_{1,2}$ as in (\ref{lambda DN hemisphere}) and $d(x,y)$ the geodesic distance in the sphere metric; ${}_2F_1$ is the hypergeometric function. Green's function on a hemisphere can be obtained by the image charge method: 
%$$G_{H^+}(x,y)=G_{S^2}(x,y)-G_{S^2}(x,\hat{y})$$ 
%where  $\hat{y}$ is the reflection of $y$ through the equatorial plane.
%
%Note also that the zeta-regularized tadpole on the sphere can be computed from $y\ra x$ asymptotics of the Green's function (\ref{G_sphere}) (by first calculating the point-splitting tadpole and then using Corollary \ref{corollary:tadpole difference}) and yields
%\begin{equation}\label{sphere tadpole}
%\tau^\mr{reg}=\frac{\log R^2-\psi(\alpha_1)-\psi(\alpha_2)}{4\pi}
%\end{equation}
%where $\psi(z)=\frac{d}{dz}\log \Gamma(z)$ is the digamma function.

\subsubsection{How far are the Dirichlet-to-Neumann operators from the square root of Helmholtz operator on the boundary?} \label{App: examples delta}
The operator $\varkappa=(\Delta+m^2)^{1/2}$ on a circle is diagonalized the basis $\phi_n(\theta)$ with eigenvalues
\begin{equation} \label{varkappa eigenvalues}
\omega_n=\left(\frac{n^2}{R^2}+m^2\right)^{1/2}
\end{equation}
Using the results (\ref{lambda DN disk}, \ref{lambda DN hemisphere})
% , \textcolor{brown}{the generalization of the latter to a spherical sector with angle $\phi$,}
 and the case of the cylinder of height $H$ considered in Section \ref{sec:example cylinder}, we have the following $n\ra\infty $ asymptotics for the ratio of the $n$-th eigenvalue of the Dirichlet-to-Neumann operator on a disk/hemisphere/cylinder\footnote{In the case of a cylinder, we mean the block of Dirichlet-to-Neumann operator connecting one of the boundary circles to itself.} to $\omega_n$:
 %\footnote{We are normalizing by the eigenvalues of $\varkappa$ for the boundary circle, thus for a spherical sector, in (\ref{varkappa eigenvalues}) we replace $R\mapsto R\sin\phi$.}
\begin{equation}\label{lambda/omega examples}
\begin{aligned}
\frac{\lambda^\mr{disk}_n}{\omega_n} &\underset{n\ra\infty}{\sim}\quad  1- \frac{(mR)^2}{2n^3}+ \mathrm{O}(n^{-4}),\\
\frac{\lambda^\mr{hemisphere}_n}{\omega_n} &\underset{n\ra\infty}{\sim}\quad  1-\frac{(mR)^2}{4n^4}+\mathrm{O}(n^{-5}),\\
\frac{\lambda^\mr{cylinder}_n}{\omega_n} & =\coth  H\omega_n \underset{n\ra\infty}{\sim}\quad  1+\mathrm{O}(n^{-\infty})
\end{aligned}
\end{equation}
Thus, $\varkappa^{-1}D_\Sigma=1+\delta$ is the identity plus a pseudodifferential operator $\delta$ of negative order $N$, where:

\begin{comment}
\begin{tabular}{c|c}
surface & $N$ \\ \hline
disk & $-3$ \\
hemisphere & $-4$ \\
cylinder & $-\infty$
\end{tabular}
\end{comment}

\begin{center}
\begin{tabular}{c|ccc}
surface & disk & hemisphere &  cylinder \\ \hline
$N$ & $-3$ & $-4$  & $-\infty$
\end{tabular}
\end{center}

The result for the hemisphere can be further generalized to a result for the spherical sector with cone angle $\phi$. In this case, we have\footnote{
Explicitly, the eigenvalues of the Dirichlet-to-Neumann operator are given by $R^{-1}\frac{d}{d\phi}\log P^n_{-\alpha}(\cos\phi)$ with $\alpha$ either root of $\alpha^2-\alpha+(mR)^2=0$ and with $P^\mu_\nu(z)$ the Legendre function. The relevant asymptotics of the Legendre function was obtained in \cite{nemes2020large}.
}
\begin{equation} \label{spherical sector}
\frac{\lambda^\mr{spherical\;sector}_n}{\omega_n} \underset{n\ra\infty}{\sim}\quad  1-\frac{(mR)^2 \cos\phi\,\sin^2\phi}{2n^3}+\frac{(mR)^2 (1+3 \cos 2\phi) \sin^2\phi}{8 n^4}+\OO(n^{-5}) 
\end{equation}
where $\omega_n= \big(\frac{n^2}{R^2 \sin^2\phi}+m^2\big)^{1/2}$ -- the eigenvalues of $\varkappa$ on the boundary circle. Thus: 
\begin{enumerate}[(a)]
\item In the case $\phi\neq \frac{\pi}{2}$, $\delta$ has order $-3$, while at $\phi=\frac{\pi}{2}$ (the case of a hemisphere) the coefficient of $n^{-3}$ in (\ref{spherical sector}) vanishes and $\delta$ degenerates to order $-4$. 
\item If $\Sigma_1$, $\Sigma_2$ are two complementary spherical sectors (which glue into a full sphere), (\ref{spherical sector}) implies that in the sum $\delta_1+\delta_2$, the order $-3$ term cancels out, leaving an order $-4$ pseudodifferential operator.
\end{enumerate}

%{\color{orange}
%For a general surface, we have the following result: 
%\begin{prop}\label{prop:delta regularity}
%Let $\Sigma$ be a surface with smooth boundary, then $S =  D_\Sigma - \varkappa $ is a pseudodifferential operator of order at most $-1$. In particular, $\delta = \varkappa^{-1}S = \varkappa^{-1}D_\Sigma - 1$ is a pseudodifferential operator of order at most $-2$, and a fortiori a trace-class operator.\footnote{Recall that a pseudodifferential operator on a circle of order $N<-1$ is automatically trace-class.}
%\end{prop}
%\begin{proof} By \cite[Proposition C.1]{TayII} we have that $D_\Sigma^{m=0} - \sqrt{\Delta_{\partial\Sigma}} = B$, where $B$ is an order $0$ pseudodifferential operator with principal symbol 
%$$\sigma(B_0)(x,\xi) = \frac12\left(\operatorname{Tr}A_{\partial \Sigma} - \frac{\langle A^*_{\partial\Sigma}\xi,\xi\rangle}{\langle \xi,\xi\rangle}\right).$$
%Here $A_{\partial \Sigma}$ is the Weingarten map. Adapting slightly the proof in loc. cit. we can see this is true also in the massive case. But in dimension $n=1$, $A_{\partial \Sigma}$ is just multiplication by a real number and the expression on the right hand side vanishes trivially.
%\end{proof}
%}

For a general surface, we have the following result: 
\begin{prop}\label{prop:delta regularity}
Let $\Sigma$ be a surface with smooth boundary, then $S =  D_\Sigma - \varkappa $ is a pseudodifferential operator of order at most $-2$. In particular, $\delta = \varkappa^{-1}S = \varkappa^{-1}D_\Sigma - 1$ is a pseudodifferential operator of order at most $-3$, and a fortiori a trace-class operator.\footnote{Recall that a pseudodifferential operator on a circle of order $N<-1$ is automatically trace-class.}
\end{prop}
\begin{proof} We will adapt the proof of \cite[Proposition C.1]{TayII}, where for $m=0$, it is shown that $D_\Sigma^{} - \sqrt{\Delta_{\partial\Sigma}} = B$, where $B$ is an order at most $0$ pseudodifferential operator with principal symbol 
$$\sigma(B_0)(x,\xi) = \frac12\left(\operatorname{Tr}A_{\partial \Sigma} - \frac{\langle A^*_{\partial\Sigma}\xi,\xi\rangle}{\langle \xi,\xi\rangle}\right).$$
Here $A_{\partial \Sigma}$ is the Weingarten map. Adapting slightly the proof in loc. cit. we can see this is true also in the massive case. Since $T^*_x\partial\Sigma$ is one dimensional, $A_{\partial \Sigma}$ is just multiplication by a real number and the expression on the right hand side vanishes trivially. This shows that $S$ is an operator of order at most $-1$. In fact, the proof of \cite[Proposition C.1]{TayII} can be used to analyze the full symbol $p_{D_\Sigma}(x,\xi)$ of $D_{\Sigma}$ (see also \cite[Proposition 1.1]{LeeUhl1989}). In particular, it can be shown that 
\[p_{D_\Sigma}(x,\xi)=|\xi|+m^2/2|\xi|+\mathrm{O}(|\xi|^{-2}).
\]
Also, the full symbol $p_{\varkappa}(x,\xi)$ of $\varkappa$ has similar behavior:
\[p_{\varkappa}(x,\xi)=|\xi|+m^2/2|\xi|+\mathrm{O}(|\xi|^{-2}).
\]
This shows that $ D_\Sigma - \varkappa $ is an operator of order at most $-2$ which completes the proof. 
\end{proof}

%For a general surface, it can be shown that $\delta$ is a trace-class
%%Hilbert-Schmidt
% operator. %\marginpar{\textcolor{cyan}{Hilbert-Schmidt implies that its order as a pseudo-differential operator is $N>2$, right? Maybe we can add this remark.}}

Also, for any $\Sigma$ with a product metric near the boundary, $\delta$ is a smoothing operator (i.e. one has $N=-\infty$), see \cite{Lee2}. 
%\marginpar{\textcolor{blue}{Wrong reference?}}

Explicit examples above, %plus a straightforward generalization of the hemisphere to a spherical sector, 
plus the expectation that the singular part of the integral kernel of $\delta$ must be universally expressed in terms of local metric characteristics (curvature of $\Sigma$ and extrinsic curvature of the boundary at the point),  suggest the following.
\begin{conjecture}\leavevmode
\begin{enumerate}[(i)]
%\item For any smooth surface $\Sigma$ with smooth boundary, endowed with a Riemannian metric, the operator $\delta=\varkappa^{-1}D_\Sigma-1$ is a PDO (pseudodifferential operator) of order $-3$.
%\item If additionally the boundary of $\Sigma$ is geodesic, $\delta$ is a PDO of order $-4$.
\item Let $\Sigma$ be a smooth surface with smooth boundary, endowed with a Reimannian metric. If the boundary of $\Sigma$ is totally geodesic, then the operator $\delta=\varkappa^{-1}D_\Sigma-1$ is a PDO (pseudodifferential operator) of order $-4$.
\item If a surface $\Sigma$ is cut by a $1$-manifold $Y$ into two surfaces $\Sigma_L,\Sigma_R$, then ${\frac12(\delta_L+\delta_R)} = \frac{1}{2\varkappa_Y} D_{\Sigma_L,\Sigma_R}-1$ is a PDO of order $-4$. Here we are not assuming that $Y$ is geodesic.
\end{enumerate}
\end{conjecture}

\begin{remark}
For comparison, we comment on the behavior of the operator $S=D_\Sigma-\varkappa$ in the case $m=0$ (on two-dimensional surfaces).\footnote{
In the massless case, we prefer to talk about  $S$ rather than $\delta=\varkappa^{-1} D_\Sigma-1 = \varkappa^{-1}S$, since the latter is not everywhere defined, due to $\varkappa$ being invertible only on the orthogonal complement of constants on a given boundary component.
}
\begin{enumerate}[(a)]
\item \label{rmk S massless (a)} The operator $D_\Sigma$ is conformally invariant in the massless case, and likewise for the operators $\varkappa$ and $S$.
 More explicitly, let $f: \Sigma \ra \Sigma'$ be a conformal diffeomorphism, 
$f^*g_{\Sigma'}= \Omega\cdot  g_\Sigma $, with  $g_\Sigma$, $g_{\Sigma'}$ the two metrics and $\Omega$ the Weyl factor. Then, for $Y$ a boundary component of $\Sigma$, we have
$$ D_\Sigma = (\Omega|_Y)^{\frac12} (f|_Y)^* D_{\Sigma'},\qquad \varkappa_Y=(\Omega|_Y)^{\frac12} (f|_Y)^*\varkappa_{Y'},\qquad   S_\Sigma = (\Omega|_Y)^{\frac12} (f|_Y)^* S_{\Sigma'} $$
\item $S$ is a smoothing operator for any $\Sigma$.\footnote{
By (\ref{rmk S massless (a)}), one can straighten the metric near the boundary into a cylindrical one by a conformal mapping; and for a cylindrical metric near the boundary, $S$ is smoothing.}
\item For $\Sigma$ homeomorphic to a disk (with any metric), $S=0$.\footnote{
Using (\ref{rmk S massless (a)}), by Riemann mapping theorem, this case reduces to the case of the standard unit disk where the explicit computation is straightforward, e.g. as $m\ra 0$ limit of (\ref{lambda DN disk}),  see also \cite{TayII}.
}
\end{enumerate}
\end{remark}

\section{Trace of stress-energy tensor and tadpole. %$m\rightarrow 0$ limit and trace anomaly: the example of a sphere
Trace anomaly.
}\label{app: trace anomaly}

In this appendix we start by giving an interpretation of the tadpole in terms of the stress-energy tensor $T$ in free massive scalar theory -- as the vacuum expectation value of the trace  $\tr T$. %Comparing with the known fact that in CFT $\tr T$ has a nonzero expectation value, proportional to 
We then compare the trace of quantum stress-energy tensor, defined in terms of the variation of the partition function w.r.t. a Weyl transformation of the metric, with the expectation value of the trace of the classical stress-energy tensor (the variation of the classical action w.r.t. a Weyl transformation) and calculate the difference between the two (the ``trace anomaly''), in free and then in interacting theory.

In the massive free scalar theory, the \textit{classical} stress-energy tensor, defined as the variational derivative of the action w.r.t. the metric, is given by 
$$T_\mr{cl}(x)=\frac{2}{\sqrt{\det{g}}}\;\frac{\delta S_\Sigma}{\delta g^{-1}(x)} \;=d\phi \otimes d\phi - g\, \Big(\frac12 \langle d\phi ,d\phi \rangle +\frac{m^2}{2}\phi^2\Big)$$
where $g$ is the metric and $g^{-1}$ its inverse; $\langle,\rangle$ is the Hodge inner product. In particular, its trace is $\mr{tr}\,T_\mr{cl} = -m^2 \phi^2$. Thus, the normalized expectation value is:
\begin{equation}\label{tr T_cl = tadpole}
\big\langle \mr{tr}\, T_\mr{cl}(x) \big\rangle = - \hbar m^2 \tau(x)
\end{equation}
So, in a free massive theory, the tadpole\footnote{Throughout this section, a ``tadpole'' means a ``zeta-regularized tadpole.''} is proportional to the expectation value of the trace of the stress-energy tensor.

Consider the $m\ra 0$ asymptotics of (\ref{tr T_cl = tadpole}) averaged over the surface:
\begin{equation}\label{int <T_cl> tends to - hbar}
\int_\Sigma d^2x\, \big\langle \mr{tr}\, T_\mr{cl}(x) \big\rangle 
= -\hbar m^2 \int_\Sigma d^2x\; \tau(x) 
\quad \underset{m\ra 0}{\sim} -\hbar + \mathrm{O}(m^2)
%= -\hbar m^2 \frac{d}{d m^2} \log \det (\Delta+m^2)
\end{equation}
The reason for that is that the averaged tadpole behaves as 
\begin{equation}\label{tau via d/dm^2, small m asympt}
\int_\Sigma d^2 \tau(x)= \frac{d}{dm^2} \log \underbrace{\det (\Delta+m^2)}_{m^2 \widetilde{\det}(\Delta+m^2)} 
=\frac{1}{m^2}+\log \widetilde{\det}(\Delta+m^2)
\quad  \underset{m\ra 0}{\sim} \quad  \frac{1}{m^2}+ \mr{const}+\mathrm{O}(m^2)
\end{equation}
Here we denoted $\widetilde{\det}(\Delta+m^2)$ the zeta-regularized determinant of $\Delta+m^2$ with the lowest eigenvalue $\lambda=m^2$ excluded from the regularized product. 
%The term $1/m^2$ in the r.h.s. above arises as $\frac{d\log m^2}{d m^2}$.
Thus, the asymptotics $1/m^2$ of the averaged tadpole -- and hence the constant asymptotics in (\ref{int <T_cl> tends to - hbar}) -- are due to the lowest eigenvalue of $\Delta+m^2$.

One also has a local version of the result (\ref{tau via d/dm^2, small m asympt}):
$$\tau(x) \underset{m\ra 0}{\sim} \frac{1}{\mr{Area}(\Sigma)\cdot m^2}+\mathrm{O}(1). $$
It follows from (\ref{tau zeta-reg definition}) and the large-time asymptotics of the heat kernel $\theta_\Delta(x,t)\underset{t\ra \infty}{\sim} \frac{1}{\mr{Area}(\Sigma)}$ (cf. Remark \ref{rem: sphere propagator m->0 limit}). 
%-- the contribution of the eigenvalue $\lambda=0$ of $\Delta$.

Next, recall that the \textit{quantum} stress-energy tensor is defined as a variational derivative of the partition function of the theory w.r.t. the metric (rather than the expectation value of the classical stress-energy tensor):
\begin{equation*}
\big\langle T_q(x)\big\rangle = -\hbar \frac{2}{\sqrt{\det{g}}} \, \frac{\delta \log Z_\Sigma}{\delta g^{-1}(x)} %= \frac{\hbar}{\sqrt{\det{g}}} \frac{\delta \log\det(\Delta+m^2)}{\delta g^{-1}(x)}
\end{equation*}
Correspondingly, its trace measures the reaction of the partition function to an infinitesimal Weyl rescaling of metric $g\ra e^{\sigma}\cdot g$ (with $\sigma\in C^\infty(\Sigma)$ a Weyl scaling factor):
\begin{equation}\label{tr T_q def}
\big\langle \tr T_q(x)\big\rangle = \hbar \frac{2}{\sqrt{\det{g}}} \;\frac{\delta }{\delta \sigma}\Big|_{\sigma=0} \log Z_\Sigma^{g\ra e^\sigma\cdot g}
\end{equation}
%where $\sigma\in C^\infty(\Sigma)$. %is the Weyl factor.
Averaging over the surface, we get
\begin{equation}\label{tr T_q averaged}
\int_\Sigma d^2 x\; \big\langle \tr T_q(x)\big\rangle =  2\hbar\, \frac{d}{d \sigma}\Big|_{\sigma=0}\log Z_\Sigma^{g\ra e^\sigma\cdot g}
\end{equation}
with $\sigma$ a constant (point-independent) Weyl scaling factor.

\begin{remark}
Recall (see e.g. \cite[Section 5.4.2]{DMS}) that in conformal field theory 
the classical stress-energy tensor is traceless, $\tr T_\mr{cl}=0$, whereas on the quantum level
one has the  ``trace anomaly'' (or ``Weyl anomaly''): the trace of the stress-energy tensor has an expectation value proportional to the central charge $c$ and the scalar curvature $K$:\footnote{
In this section we use the notation $K$ rather than $\mr{R}$ for the scalar curvature, to avoid the mix up with the radius of the sphere.
}
\begin{equation}\label{Weyl anomaly CFT}
\big\langle \tr T_q \big\rangle^\mr{CFT} = \frac{c\hbar}{24 \pi} K
\end{equation}
This corresponds to the following behavior of the partition function of a CFT under finite Weyl transformations of metric (see e.g. \cite{Seg1}, \cite{gawedzki1989conformal}): %\marginpar{normalization?}
$$ Z^\mr{CFT}_{\Sigma, e^\sigma\cdot g} = e^{\frac{c}{48 \pi}\int_\Sigma \frac12 d\sigma\wedge *d\sigma+K\sigma \dvol_\Sigma }\;\;  Z^\mr{CFT}_{\Sigma,g} $$
In particular, %two-dimensional 
free massless scalar field is a conformal theory with central charge $c=1$ and, according to (\ref{Weyl anomaly CFT}), should satisfy
\begin{equation}\label{Weyl anomaly CFT scalar field}
\big \langle \tr T_q \big\rangle_{m=0}^\mr{CFT} = \frac{\hbar}{24 \pi} K
\end{equation}
%\marginpar{REMOVE?}  \textcolor{cyan}{For instance, for $\Sigma$ a flat torus, this yields $\big\langle \tr T_q \big\rangle =0$ for the massless theory, which seems to be in contradiction with (\ref{int <T_cl> tends to - hbar}).}
\end{remark}

\subsection{Trace of quantum vs. classical stress-energy tensor on a sphere}
Consider the example of $\Sigma$ a sphere of radius $R$. 
The partition function of the free massive scalar theory is $Z=\det(\Delta+m^2)^{-\frac12}$ with the determinant given by (\ref{det sphere}).

Note that (\ref{det sphere}) implies that
\begin{equation*}
\left(R^2\frac{\dd}{\dd R^2}-m^2 \frac{\dd}{\dd m^2}\right) \log \det(\Delta+m^2) = -\frac13+m^2R^2
\end{equation*}
Comparing with (\ref{tr T_q averaged}) (where a global Weyl rescaling is tantamount to changing the radius of the sphere), with (\ref{tr T_cl = tadpole}) and with Corollary \ref{corollary: tau via d/dm^2 log det}, we obtain
\begin{equation*}
\int_{S^2} d^2x\; \big\langle \tr T_q \big\rangle = \int_{S^2} d^2x\; \big\langle \tr T_\mr{cl} \big\rangle + \hbar\left(\frac13 - m^2R^2\right)
\end{equation*}
Or in the non-averaged form (here we can use that $\big\langle \tr T_q \big\rangle$, $\big\langle \tr T_\mr{cl} \big\rangle$ must be constant functions in the case of a sphere, due to isometries acting transitively):
\begin{equation}\label{tr T_q vs tr T_cl on a sphere}
\big\langle \tr T_q \big\rangle = \big\langle \tr T_\mr{cl} \big\rangle + \hbar\left(\frac{1}{12 \pi R^2} - \frac{m^2}{4\pi}\right)
\end{equation}
We interpret the second term in the r.h.s. as the trace %Weyl 
anomaly in the free massive scalar theory.

\begin{comment}
\textcolor{cyan}{ \marginpar{Remove?}
In the limit $m\ra 0$, (\ref{tr T_q vs tr T_cl on a sphere}) yields
\begin{equation}\label{tr T_q on a sphere}
\big\langle \tr T_q \big\rangle_{m=0} =  -\frac{\hbar}{6\pi R^2}
\end{equation}
where we used the asymptotics (\ref{int <T_cl> tends to - hbar}). This result is consistent with the answer (\ref{Weyl anomaly CFT scalar field}) known from CFT.
%
Note that (\ref{det sphere}) implies 
\begin{equation*}
\widetilde{\det}(\Delta+m^2)\quad \underset{m\ra 0}{\sim} \mr{const}\cdot R^{-\frac23+2} (1+\mathcal{O}(m^2R^2))
\end{equation*}
where tilde over the determinant means omitting the lowest eigenvalue $\lambda=m^2$.  Therefore, at $m=0$ one has $\det'(\Delta)=\mr{const}\cdot R^{\frac{4}{3}}$ and the partition function of the free massless scalar field on the sphere is 
$$ Z_{m=0}=({\det}'(\Delta))^{-\frac12}=\mr{const}\cdot R^{-\frac23} $$
The logarithmic derivative of this partition function w.r.t. constant Weyl transformations, according to (\ref{tr T_q averaged}), yields the result (\ref{tr T_q on a sphere}) for the trace %Weyl 
anomaly on a sphere.}
\end{comment}

\subsection{Trace anomaly for a general surface (and in the interacting theory)}
The following result is a generalization of (\ref{tr T_q vs tr T_cl on a sphere}) to an arbitrary closed surface. 
\begin{lemma}\label{lem: trace anomaly}
 For the free massive scalar theory on a closed surface $\Sigma$, the trace of the quantum stress-energy tensor at a %an interior 
 point $x$ is
\begin{equation}\label{trace anomaly massive}
\big\langle \tr T_q(x) \big\rangle = \big\langle \tr T_\mr{cl}(x) \big\rangle +\hbar\frac{1}{4\pi}\left(\frac{K(x)}{6}-m^2\right) 
\end{equation}
with $K(x)$ the scalar curvature. The first term on the r.h.s. is $\big\langle \tr T_\mr{cl}(x) \big\rangle=-\hbar m^2 \tau(x)$ with $\tau(x)$ the zeta-regularized tadpole (\ref{tau zeta-reg definition}). 
\end{lemma}
We interpret the second term on the r.h.s. of (\ref{trace anomaly massive}) as the trace anomaly.

\begin{proof} For $\sigma\in C^\infty(\Sigma)$, denote $A_\sigma=e^{-\sigma}\Delta+m^2$  the Helmholtz operator for the Weyl-rescaled metric $e^\sigma g$. First note that  $\tr e^{-t A_{\epsilon\sigma}} = \tr e^{-t A}+\epsilon t \tr \sigma \Delta e^{-t A}+\mathrm{O}(\epsilon^2)$ (cf. the proof of Lemma \ref{lem: tau zeta-reg is variational derivative in mass}).
Using this, we have
\begin{align*}
\frac{d}{d\epsilon}\Big|_{\epsilon=0} \log&\det A_{\epsilon \sigma} %(e^{-\epsilon\sigma}\Delta+m^2) 
=-\frac{d}{ds}\Big|_{s=0} \frac{1}{\Gamma(s)}\int_0^\infty dt\, t^{s-1} \frac{d}{d\epsilon}\Big|_{\epsilon=0} \tr e^{-t  A_{\epsilon\sigma}  % (e^{-\epsilon\sigma}\Delta+m^2)
} \\
&= -\frac{d}{ds}\Big|_{s=0} \frac{1}{\Gamma(s)}\int_0^\infty dt\, t^{s-1} \cdot t\,\tr \sigma\left(-\frac{\dd}{\dd t}-m^2  \right) e^{-t A} \\
&= -\frac{d}{ds}\Big|_{s=0} \frac{1}{\Gamma(s)}\int_0^\infty dt\, \left( s\, t^{s-1}-m^2 t^s  \right) \tr \sigma e^{-t A} \\
&=-\frac{d}{ds}\Big|_{s=0}\int_\Sigma d^2 x\, \sigma(x) (s \zeta_A(s,x)-m^2 s\zeta_{A}(s+1,x))\\
&=\int_\Sigma d^2 x\,\sigma(x) (-\zeta_{A}(0,x)+m^2 \tau(x))
\end{align*}
Comparing with (\ref{tr T_q def}) and using the definition of the partition function in free theory $Z=\det^{-\frac12} A$, we have
\begin{equation}\label{T_q = T_cl + zeta(0)}
\big\langle T_q(x) \big\rangle = -\hbar m^2 \tau(x) +\hbar\zeta_A(0,x)= \langle T_\mr{cl}(x) \rangle + \hbar\zeta_A(0,x) 
\end{equation}
We find the value $\zeta_A(0,x)$ from the small-$t$ expansion of the heat kernel (\ref{eq:local heat trace expansion}):
\begin{equation}\label{zeta_A local at zero computation}
\begin{aligned}
\zeta(0,x)&=\lim_{s\ra 0} \Big(\frac{1}{\Gamma(s)}\int_0^\infty dt\, t^{s-1}\, e^{-m^2 t} \left(\theta_\Delta(t,x)-\frac{1}{4\pi t}-\frac{1}{4\pi}\frac{K(x)}{6}\right) +\\ &+\frac{1}{4\pi}\frac{m^{-2(s-1)}}{s-1}+\frac{1}{4\pi}\frac{K(x)}{6}m^{-2s}\Big) \qquad =\quad  \frac{1}{4\pi} \left(\frac{K(x)}{6}-m^2\right)
\end{aligned}
\end{equation}
Putting (\ref{T_q = T_cl + zeta(0)}) and (\ref{zeta_A local at zero computation}) together, we obtain the statement.
\end{proof}

\begin{remark}[On the massless limit and comparison to CFT]\label{rem:comparison to CFT} In the limit $m\ra 0$, (\ref{trace anomaly massive}) becomes
\begin{equation*}
\big\langle \tr T_q(x) \big\rangle = -\frac{\hbar}{\mr{Area}(\Sigma)}+\frac{\hbar}{4\pi} \frac{K(x)}{6}
\end{equation*}
which seems to contradict the known result (\ref{Weyl anomaly CFT scalar field}) from CFT. The explanation is that the trace of the quantum stress-energy tensor 
in the $m\ra 0$ limit of the massive theory and in the massless theory differ by a shift due to different normalizations 
%\marginpar{this is a convention!}
 of partition functions. A reasonable normalization/regularization of the partition function of the massless free scalar field CFT is:
\begin{equation}\label{Z^CFT normalization}
\begin{aligned}
Z^\mr{CFT}_{m=0}=\lim_{m\ra 0} \left[\frac{\det (\Delta+m^2)}{m^2\mr{Area}(\Sigma)}\right]^{-1/2}&=\mr{Area}(\Sigma)^\frac12\lim_{m\ra 0} \til\det(\Delta+m^2)^{-\frac12}\\ &=\mr{Area}(\Sigma)^\frac12 \big({\det}'\Delta\big)^{-\frac12}
\end{aligned}
\end{equation}
where $\til\det$ is the determinant with the lowest eigenvalue $\lambda=m^2$ excluded.\footnote{
For example, for $\Sigma =T_\tau= \mathbb{C}/(\mathbb{Z}\oplus \tau \mathbb{Z})$ a torus with modular parameter $\tau$, the extra factor $\mr{Area}^{\frac12}$ in (\ref{Z^CFT normalization}) restores the invariance of the partition function under the modular transformation $\tau\ra -\tau^{-1}$. Indeed, by Kronecker's limit formula (see e.g. \cite{gawedzki1997lectures}, \cite{Q1993}), ${\det'}\Delta=(\mr{Im}\tau)^2|\eta(\tau)|^4$, with $\eta(\tau)$ the Dedekind's eta function. This determinant is not invariant under $\tau\ra - \tau^{-1}$. However, $\frac{1}{\mr{Area}(T_\tau)}{\det}'\Delta=(\mr{Im}\tau)^{-1}{\det}'\Delta$ is invariant, and thus the partition function with normalization (\ref{Z^CFT normalization}) is also modular invariant. \\
\indent As a side note, 
in the operator formalism, the partition function for a torus $Z^\mr{CFT}_\mr{op}=\tr_{H_{S^1}} q^{L_0-\frac{c}{24}}\bar{q}^{\bar{L}_0-\frac{\bar{c}}{24}}$, with $q=e^{2\pi i \tau}$,  is in fact ill-defined for the massless scalar CFT, due to a continuum %spectrum 
of primary fields/states (vertex operators). However, it can be regularized by allowing the scalar field to take values in a target circle of radius $r$ (see e.g. \cite{DMS}). For a large target radius, this regularized partition function behaves as $Z^\mr{CFT}_\mr{op}(r) \underset{r\ra\infty}{\sim} r\cdot Z^\mr{CFT}$ where the coefficient on the right is the partition function (\ref{Z^CFT normalization}).
} 
The extra factor $\mr{Area}(\Sigma)^{\frac12}$ here accounts, upon taking the variation w.r.t. Weyl transformations, for the  difference 
$$\big\langle \tr T_q \big\rangle^\mr{CFT} = \lim_{m\ra 0}\big\langle \tr T_q \big\rangle+\frac{\hbar}{\mr{Area}(\Sigma)}$$
Therefore, (\ref{trace anomaly massive}) and (\ref{Weyl anomaly CFT scalar field}) are in agreement and not in contradiction.
\end{remark}

Lemma \ref{lem: trace anomaly} admits the following generalization to the non-free case.
\begin{prop}\label{prop: trace anomaly}
Consider the massive scalar theory with interaction potential $p(\phi)=\sum_{n\geq 0}\frac{p_n}{n!}\phi^n$ on a closed surface $\Sigma$. 
%i.e., the theory defined by the action 
%$$S(\phi)=\int_\Sigma \frac12 d\phi\wedge *d\phi +\frac{m^2}{2}\phi^2 \dvol_\Sigma + p(\phi)\dvol_\Sigma.$$ 
The trace of the quantum stress energy, %at any point $x\in \Sigma$, 
defined via (\ref{tr T_q def}), satisfies
\begin{equation}\label{trace anomaly with p}
\big\langle \tr T_q(x) \big\rangle = \Big\langle \tr T_\mr{cl}(x) +  
\hbar\frac{1}{4\pi}\left(\frac{K(x)}{6}-m^2-\frac{\dd^2}{\dd \phi^2} p(\phi)\right) \Big\rangle
\end{equation}
%with the second term being the trace anomaly. 
at any point $x\in \Sigma$. Here $\tr T_\mr{cl}=-m^2\phi^2-2\,p(\phi)$ is the trace of the classical stress-energy tensor. %and
%\begin{equation}
%\big\langle \tr T_\mr{cl}(x) \big\rangle = -\hbar m^2 \tau(x)-2\sum_{n\geq 0} \frac{p_{2n}}{n!} \left(\frac{\hbar \tau(x)}{2}\right)^n
%\end{equation}
On the r.h.s., $\langle\cdots\rangle$ means the normalized expectation value (one-point correlation function) in the interacting theory.
\end{prop}

\begin{proof}
Consider the expression 
\begin{equation}\label{inf Weyl tranf of log Z}
\frac{d}{d\epsilon}\Big|_{\epsilon=0} \log Z_\Sigma^{g\ra e^{\epsilon \sigma}g}
\end{equation}
On one hand, by definition (\ref{tr T_q def}), it is equal to 
\begin{equation}\label{trace anomaly tr T_q}
\frac{1}{2\hbar}\int_\Sigma d^2 x \,\sigma(x)\,\big\langle \tr T_q(x) \big\rangle
\end{equation}
On the other hand, 
\begin{equation}\label{log Z expanded}
\log Z_\Sigma = -\frac12 \log\det(\Delta+m^2)+ \sum_{\Gamma\,\mr{connected}} \frac{\hbar^{-\chi(\Gamma)}}{|\Aut(\Gamma)|}F_\Gamma
\end{equation}
where the sum on the right is over connected Feynman graphs. Taking the derivative of this expression w.r.t. a Weyl transform, we get the following contributions to (\ref{inf Weyl tranf of log Z}):
\begin{enumerate}[(i)]
\item Derivative hits $-\frac12 \log\det(\Delta+m^2)$. This gives the contribution $$\int_\Sigma d^2 x\, \sigma(x)\, \big(\underbrace{- \frac{m^2}{2}\tau(x)}_{\mr{(i)'}}+
\underbrace{\frac12\,\frac{1}{4\pi}\left(\frac{K(x)}{6}-m^2\right)}_{\mr{(i)''}}\big) $$ as we obtained in Lemma \ref{lem: trace anomaly}.
\item Derivative hits $\dvol_\Sigma$ in $F_\Gamma$ corresponding to one of the vertices of $\Gamma$  (recall that the Riemannian volume form contains $\sqrt{\det g}$ which scales with Weyl transformations). This gives the contribution\footnote{Here and below when we say ``contribution $\cdots$'' while discussing the element $\Xi$ of Feynman diagrams, we mean ``the contribution to (\ref{inf Weyl tranf of log Z}) where one instance of $\Xi$ is replaced by $\frac{d}{d\epsilon}\big|_{\epsilon=0}\Xi^{g\ra e^{\epsilon \sigma} g}=\cdots$.'' }
$$\frac{d}{d\epsilon}\Big|_{\epsilon=0} \dvol_x^{g\ra e^{\epsilon \sigma }g}=\sigma(x) \dvol_x$$
\item Derivative hits one of the edges (Green's functions) in $F_\Gamma$ -- but not one of the short loops. This gives the contribution\footnote{
Indeed, 
using the formula for the derivative of the inverse of an operator $M_\epsilon$ in a parameter, $\frac{d}{d\epsilon}M^{-1}_\epsilon = - M^{-1}_\epsilon (\frac{d}{d\epsilon}M_\epsilon) M^{-1}_\epsilon$, we have:
$ \frac{d}{d\epsilon}\Big|_{\epsilon=0}  \int \dvol_{x_2} G(x_1,x_2) f(x_2)\Big|_{g\ra e^{\epsilon \sigma }g}=-\int d^2x \int d^2 x_2 G(x_1,x)(-\sigma(x)\Delta_x)G(x,x_2)f(x_2)= -m^2\int d^2 x \int d^2 x_2 G(x_1,x) \sigma(x) G(x,x_2) f(x_2)+\int d^2 x_2 G(x_1,x_2) \sigma(x_2) f(x_2)$ where the last term is due to the variation of $\dvol_{x_2}$ while the first is due to the variation of the Green's function itself. Here $f(x_2)$ is an arbitrary test function and we used that $\Delta_x G(x,x_2)=-m^2 G(x,x_2)+\delta(x,x_2)$.
} 
\begin{equation}\label{trace anomaly G var derivative}
 \frac{d}{d\epsilon}\Big|_{\epsilon=0} G^{g\ra e^{\epsilon \sigma }g}(x_1,x_2) =  \int_\Sigma d^2 x\, G(x_1,x) (-m^2 \sigma(x)) G(x,x_2) 
\end{equation}
\item Derivative hits a short loop, giving the contribution\footnote{
This is proven easiest by analyzing the point-split tadpole (\ref{tadpole:splitting}) where the first term $G(x,x')$ reacts to the Weyl transform according to (\ref{trace anomaly G var derivative}) and the variation of the singular subtraction $\frac{1}{2\pi}\log d(x,x')$ yields the second term in  (\ref{trace anomaly tau var derivative}). Then we recall that $\tau$ and $\tau^\mr{split}$ differ by a universal constant, see Corollary \ref{corollary:tadpole difference}; thus, their variational derivatives coincide.
}
\begin{equation}\label{trace anomaly tau var derivative}
 \frac{d}{d\epsilon}\Big|_{\epsilon=0} \tau^{g\ra e^{\epsilon \sigma }g}(x)=\underbrace{\int_\Sigma d^2y\, G(x,y) (-m^2 \sigma(y)) G(y,x)}_{\mathrm{(iv)'}} +
  \underbrace{\frac{1}{4\pi} \sigma(x)}_{\mr{(iv)''}}
\end{equation}
\end{enumerate}
Next, we note that
\begin{itemize}
\item Contributions of type (ii) above sum up to the expectation value 
\begin{equation}\label{trace anomaly p contribution}
\hbar^{-1}\int_\Sigma d^2x\, \sigma(x)\,\big\langle -p(\phi(x)) \big\rangle
\end{equation}
\item Contributions of types (iii), (iv)' and (i)' sum up to 
\begin{equation}\label{trace anomaly m^2 phi^2 contribution}
\hbar^{-1}\int_\Sigma d^2x\, \sigma(x) \big\langle-\frac{m^2}{2}\phi^2(x)\big\rangle
\end{equation}
\item Contributions (iv)'' sum up to %\marginpar{missing $\frac12 $ here -- symmtry factor?}
\begin{equation}\label{trace anomaly p'' contribution}
\int_\Sigma d^2x\, \sigma(x) \Big\langle -\frac{1}{4\pi}\underbrace{\sum_n \frac{p_n}{n!}  
\left(\begin{array}{c}n \\  2\end{array}
\right) \phi^{n-2}(x)}_{\frac12 \frac{\dd^2}{\dd \phi^2}p(\phi)}
\Big\rangle 
\end{equation}
\end{itemize}
Finally, we note that (\ref{trace anomaly p contribution}) and (\ref{trace anomaly m^2 phi^2 contribution}) together yield $\frac{1}{2\hbar}\int d^2x \,\sigma(x)\, \big\langle T_\mr{cl}(x)\big\rangle$, while (\ref{trace anomaly p'' contribution}) and (i)'' above together yield
$$\frac12 \int d^2 x\, \sigma(x)\, \frac{1}{4\pi}\Big\langle \frac{K(x)}{6}-m^2-\frac{\dd^2}{\dd\phi^2}p(\phi)\Big\rangle$$
Thus, the sum of these two expressions is (\ref{trace anomaly tr T_q}), which finishes the proof.
\end{proof}

%[-->  CASE WITH BOUNDARY]

In Lemma \ref{lem: trace anomaly} and in Proposition \ref{prop: trace anomaly}, one can allow $\Sigma$ to be a surface with boundary and $x$ an \textit{interior} point. Both results continue to hold in this case, where now l.h.s. and r.h.s. are understood as functions of the boundary field $\eta$ (the proof is adapted straightforwardly).

Note that the second derivative of the interaction %in the field 
$\frac{\dd^2}{\dd\phi^2}p(\phi)$ 
that we see in the trace anomaly in (\ref{trace anomaly with p}), we have encountered before in the context of RG flow, see (\ref{RG flow eq for p_Lambda}), (\ref{RG flow eq local}).

%Note: we saw the expression $p''(\phi)$ before when discussing the RG flow (...).

\bibliographystyle{abbrv}
\bibliography{bibliography}
%Institute for mathematics, University of Z\"urich,
%Z\"urich, CH\\
%$Email:$ \text{skandel1@alumni.nd.edu}

\end{document}